\PassOptionsToPackage{pdftex,hyperfootnotes=false,pdfpagelabels}{hyperref}

\documentclass[%
aps,pra,
reprint,
superscriptaddress,
twocolumn, 
nofootinbib,
floats,floatfix
]{revtex4-2}

\usepackage{amsmath}
\usepackage{pifont}
\usepackage{amssymb}
\usepackage{amsthm}
\usepackage{graphicx}
\usepackage{xcolor}
\usepackage[english]{babel}
\usepackage{csquotes}
\usepackage{braket}
\usepackage{mathtools}
\usepackage{enumitem}
\usepackage{IEEEtrantools}
\usepackage{relsize}
\usepackage{placeins}
\usepackage{comment}
\usepackage[normalem]{ulem}
\usepackage{systeme}
\usepackage{physics}
\usepackage{algpseudocode,algorithm}

\usepackage{titlesec}

\makeatletter
\newcommand\footnoteref[1]{\protected@xdef\@thefnmark{\ref{#1}}\@footnotemark}
\makeatother

\usepackage{caption}
\usepackage{algorithm}% http://ctan.org/pkg/algorithms
\usepackage{algpseudocode}% http://ctan.org/pkg/algorithmicx

\usepackage{tikz-cd} 
\usepackage{adjustbox}

\usepackage{tikz}						%geometric/algebraic description.
\usetikzlibrary{math, arrows,shapes.misc,
		       automata,backgrounds,
		       petri,topaths, decorations.pathmorphing, tikzmark}	

\usepackage{pgffor}

\usepackage{pgfplots}
\usetikzlibrary{pgfplots.groupplots}
\pgfplotsset{compat=1.11}
\usepgfplotslibrary{fillbetween}

\usepackage{tikz}
\usetikzlibrary{
  shapes,
  shapes.geometric,
	trees,
	matrix,
  positioning,
    pgfplots.groupplots,
  }

\usepackage{hyperref}
\hypersetup{
	colorlinks = true,
	linkcolor = [rgb]{0.70,0.13,0.13},
	citecolor = [rgb]{0.13,0.55,0.13},
	urlcolor = [rgb]{0.25, 0.41, 0.88}}
\usepackage[capitalize]{cleveref}

\providecommand{\calC}{\ensuremath{\mathcal{C}}}
\providecommand{\calD}{\ensuremath{\mathcal{D}}}

\providecommand{\calL}{\ensuremath{\mathcal{L}}}

\providecommand{\calO}{\ensuremath{\mathcal{O}}}
\providecommand{\calP}{\ensuremath{\mathcal{P}}}

\providecommand{\calS}{\ensuremath{\mathcal{S}}}
\providecommand{\calT}{\ensuremath{\mathcal{T}}}
\providecommand{\calU}{\ensuremath{\mathcal{U}}}

\providecommand{\bbC}{\ensuremath{\mathbb{C}}}

\providecommand{\bbE}{\ensuremath{\mathbb{E}}}

\providecommand{\bbU}{\ensuremath{\mathbb{U}}}

\providecommand{\calC}{\ensuremath{\mathcal{C}}}
\providecommand{\calD}{\ensuremath{\mathcal{D}}}

\providecommand{\calL}{\ensuremath{\mathcal{L}}}

\providecommand{\calO}{\ensuremath{\mathcal{O}}}
\providecommand{\calP}{\ensuremath{\mathcal{P}}}

\providecommand{\calS}{\ensuremath{\mathcal{S}}}
\providecommand{\calT}{\ensuremath{\mathcal{T}}}
\providecommand{\calU}{\ensuremath{\mathcal{U}}}

\providecommand{\bbC}{\ensuremath{\mathbb{C}}}

\providecommand{\bbE}{\ensuremath{\mathbb{E}}}

\providecommand{\bbU}{\ensuremath{\mathbb{U}}}

\def\01{\{0,1\}}

\makeatletter
\let\oldabs\abs
\def\abs{\@ifstar{\oldabs}{\oldabs*}}
\let\oldnorm\norm
\def\norm{\@ifstar{\oldnorm}{\oldnorm*}}
\makeatother

\newtheorem{claim}{Claim}
\newtheorem{fact}{Fact}
\newtheorem{theorem}{Theorem}
\newtheorem{lemma}[theorem]{Lemma}

\newtheorem{proposition}[theorem]{Proposition}
\newtheorem{definition}[theorem]{Definition}

\newtheorem{corollary}[theorem]{Corollary}
\newtheorem{conjecture}[theorem]{Conjecture}

\definecolor{armando}{rgb}{.5,.1,.7}

\allowdisplaybreaks

\definecolor{manuel}{rgb}{.6,.8,.1}

\begin{document}

\title{Classically estimating observables of noiseless quantum circuits}
\date{\today}

\author{Armando Angrisani}
\email{armando.angrisani@epfl.ch}
\affiliation{Institute of Physics, Ecole Polytechnique Fédérale de Lausanne (EPFL),  Lausanne CH-1015, Switzerland}
\affiliation{Centre for Quantum Science and Engineering, Ecole Polytechnique F\'{e}d\'{e}rale de Lausanne (EPFL), CH-1015 Lausanne, Switzerland}

\author{Alexander Schmidhuber}
\email{alexsc@mit.edu}
\affiliation{Center for Theoretical Physics, Massachusetts Institute of Technology, Cambridge, MA, USA}

\author{Manuel S. Rudolph}
\email{manuel.rudolph@epfl.ch}
\affiliation{Institute of Physics, Ecole Polytechnique Fédérale de Lausanne (EPFL),  Lausanne CH-1015, Switzerland}
\affiliation{Centre for Quantum Science and Engineering, Ecole Polytechnique F\'{e}d\'{e}rale de Lausanne (EPFL), CH-1015 Lausanne, Switzerland}

\author{\\M. Cerezo}
\affiliation{Information Sciences, Los Alamos National Laboratory, Los Alamos, NM 87545, USA}

\author{Zo\"e Holmes}
\affiliation{Institute of Physics, Ecole Polytechnique Fédérale de Lausanne (EPFL),  Lausanne CH-1015, Switzerland}
\affiliation{Centre for Quantum Science and Engineering, Ecole Polytechnique F\'{e}d\'{e}rale de Lausanne (EPFL), CH-1015 Lausanne, Switzerland}

\author{Hsin-Yuan Huang}
\affiliation{Google Quantum AI, Venice, CA, USA}
\affiliation{Center for Theoretical Physics, Massachusetts Institute of Technology, Cambridge, MA, USA}
\affiliation{Institute for Quantum Information and Matter, Caltech, Pasadena, CA, USA}

\begin{abstract}
\normalsize
%We present a classical algorithm for estimating expectation values of arbitrary observables on most quantum circuits across all circuit architectures and depths, including those with all-to-all connectivity. We prove that for any architecture where each circuit layer is equipped with a measure invariant under single-qubit rotations, our algorithm achieves a small error $\varepsilon$ on all circuits except for a small fraction $\delta$. The computational time is polynomial in qubit count and circuit depth for any small constant $\varepsilon, \delta$, and quasi-polynomial for inverse-polynomially small $\varepsilon, \delta$. 
%Given that most quantum circuits in an architecture exhibit chaotic and locally scrambling behavior, our work demonstrates that estimating observables of such quantum dynamics is classically tractable across all geometries.
{
We present a classical algorithm {based on Pauli propagation} for estimating expectation values of arbitrary observables on {random unstructured} quantum circuits across all circuit architectures and depths, including those with all-to-all connectivity. We prove that for any architecture where each circuit layer is {randomly sampled from a distribution} invariant under single-qubit rotations, our algorithm achieves a small error $\varepsilon$ on all circuits except for a small fraction $\delta$. The computational time is polynomial in qubit count and circuit depth for any small constant $\varepsilon, \delta$, and quasi-polynomial for inverse-polynomially small $\varepsilon, \delta$. 
Our results show that estimating observables of quantum circuits {exhibiting chaotic and locally scrambling behavior} is classically tractable across all geometries.
We further conduct numerical experiments beyond our average-case assumptions, demonstrating the potential utility of Pauli propagation methods for simulating real-time dynamics and finding low-energy states of physical Hamiltonians.
}

\end{abstract}

\maketitle

\makeatletter

\paragraph*{Introduction.} Simulating all quantum circuits with classical algorithms is believed to be computationally hard. Yet, specialized classical simulation methods exist that can exploit the properties of certain restricted kinds of quantum systems. These include tensor networks which are tailored towards simulating low-entangled systems~\cite{shi2006classical, pang2020efficient, biamonte2017tensor, hauschild2018efficient, causer2023optimal, markov2008simulating, singh2010tensor, patra2023efficient}, Clifford perturbation methods for low-magic systems~\cite{gottesman1998heisenberg,aaronson2004improved,nest2008classical,bravyi2016improved,beguvsic2023simulating}, more general group-theoretic approaches for symmetrised systems~\cite{somma2005quantum, somma2006efficient, Galitski2011Quantum, goh2023lie,anschuetz2022efficient}, specialized approaches for constant-depth circuits\ \cite{napp2022efficient, bravyi2021classical, bravyi2024classical}, data-driven machine-learning-based methods \cite{huang2021power, huang2021provably, huang2022learning, du2024efficient}, and heuristic approaches such as neural network states~\cite{carleo2017solving, yang2024can}. Understanding the regimes in which such classical algorithms are effective is crucial in the hunt for applications where we can achieve a quantum computational advantage. This is particularly important in the near term where it is not yet possible to run complex quantum algorithms that achieve a provable quantum speed-up.

In this work, we prove that for a large class of quantum circuits, it is possible to classically compute the associated expectation values to within a small additive error. In particular, we show that it is possible to estimate observables on %\textit{most} 
{ randomly sampled}
quantum circuits in any circuit architecture, where each circuit layer is {independently drawn from a distribution} %equipped by a probability measure that is
invariant under single-qubit rotations. Such ensembles of circuit layers are known as \textit{locally scrambling}. This property is satisfied by a wide range of deep and shallow unstructured parameterised quantum circuits of different topologies currently used by variational quantum algorithms (VQAs)~\cite{cerezo2020variationalreview, bharti2022noisy}. The property may also be approximately satisfied by the dynamics of certain chaotic systems~\cite{roberts2017chaos, belyansky2020minimal, geller2022quantum, zhang2024thermalization}.

\begin{figure}
    \centering
    \includegraphics[width=0.88\linewidth]{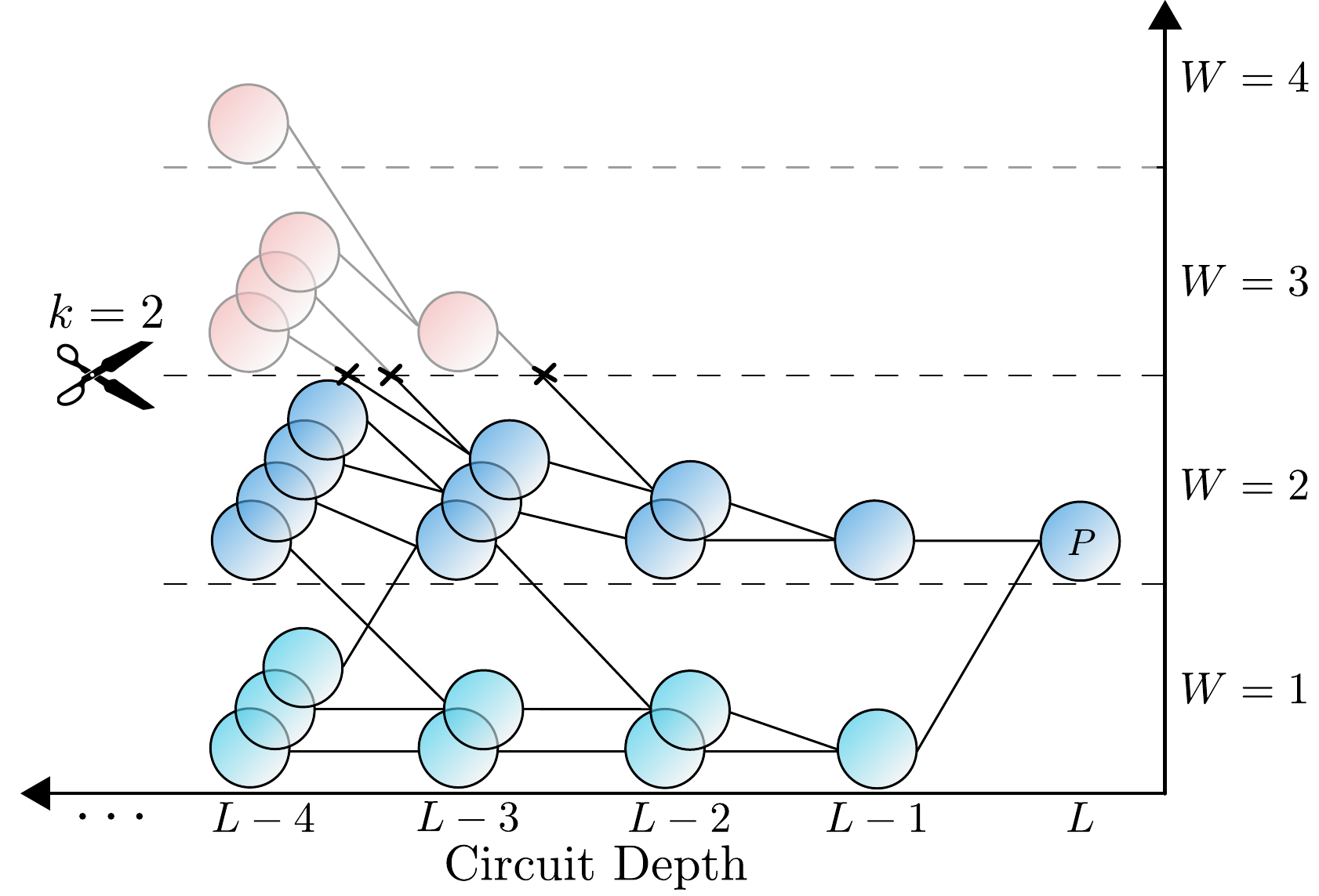}
    \caption{\small \emph{Schematic depiction of Pauli propagation equipped with weight truncation}. Pauli operators generally split into a weighted sum of several Pauli operators when acted upon by non-Clifford operations and to higher weight. We sketch the truncation of Pauli operators above a threshold of $k=2$.}
    \label{fig:schematic}
    \vspace{-20pt}
\end{figure}

To demonstrate the simulability of such circuits, {we employ a low-weight Pauli propagation algorithm~\cite{aharonov2022polynomial, rudolph2023classical, schuster2024polynomial, gonzalez2024pauli}}, which approximates the unitary evolution of the observable in the Heisenberg picture via a truncated Pauli path integral as sketched in Fig.~\ref{fig:schematic}. Although Pauli-path-based classical simulation techniques have gained attention in recent years~\cite{rakovszky2022dissipation, von2022operator}, their performance was only rigorously proven in the presence of noise~\cite{aharonov2023polynomial, fontana2023classical, rudolph2023classical,shao2023simulating,schuster2024polynomialtimeclassicalalgorithmnoisy,gonzalez2024pauli} and on near-Clifford circuits\ \cite{beguvsic2023simulating}.  Intuitively, high-weight Pauli operators are damped exponentially more by noise in the circuits, and therefore the paths containing high-weight Pauli operators can be truncated. A formal proof of their effectiveness on generic noiseless circuits stood as an open problem. 

Here, we demonstrate that even in the absence of noise, the same approach remains valid due to the scrambling action of unitaries, which causes the high-weight Pauli terms to be lost in the exponentially large operator space and exponentially suppresses their effect on the expectation values. Our results indicate that low-weight Pauli propagation approaches typically outperform alternative simulation techniques on unstructured circuits, particularly when compared to brute-force light-cone methods, which are generally only efficient for very shallow and geometrically local circuits. 

{
It is important to stress that our provable guarantees do not hold for structured circuits  violating our average-case assumptions. Nevertheless, our numerical results in Supplemental Material~\ref{app:numerics-IBM} indicate that Pauli propagation may present favorable performances also on circuits with highly correlated layers, such as those arising in quantum simulation and during the \emph{training} of variational quantum algorithms.
Exploring this potential and understanding its limits remain important future directions.
}

%It is important to stress that our provable guarantees hold only for \emph{most} instances, but not for all instances. As such, there could be randomly sampled circuits for which our classical algorithm yields a large estimation error. In our accompanying paper Ref.~\cite{bermejo2024quantumconvolutionalneuralnetworks}, we provide numerical evidence that these pathological cases do not seem to play a significant role in the optimization of some parameterised quantum circuits. In addition, many circuits of interest, such as the structured circuits with correlated parameters used for dynamical simulation, do not resemble typical locally scrambling quantum circuits. We end by providing numerical evidence that even for such circuits, Pauli propagation with weight truncation can be highly effective in estimating expectation values.

\paragraph*{Framework.}\label{sec:framework}

We consider the task of classically simulating expectation values of the form 
\begin{equation}\label{eq:loss}
   f_U(O) := \Tr[U \rho U^\dagger O]\,,
\end{equation}
for a state $\rho$, an observable $O$ and a $L$-layered quantum circuit given by
\begin{equation}\label{eq:circuits}
    U = U_L U_{L-1} \dots U_1.
\end{equation}
For our error bounds, we consider each layer $U_j$ to be %equipped with 
sampled independently from a \textit{locally scrambling} distribution~\cite{kuo2020markovian, hu2021classical, caro2022outofdistribution, gibbs2024dynamical, huang2022learning}, that is, a distribution that is invariant under rotation by random single-qubit Clifford gates.
In loose terms, locally scrambling ensembles of circuits can be thought of as ensembles of quantum circuits that look random locally on each qubit.

We further assume for our computational complexity bounds that for any Pauli operator $P$ of weight\footnote{For a Pauli operator $P \in \{I,X,Y,Z\}^{\otimes n}$, we denote by $\abs{P}$ the %$\mathrm{supp}(P)$ the set of qubits on which $P$ is non-identity, and we define the 
\textit{weight} of $P$ 
%as $\abs{P} := \abs{\mathrm{supp}(P)}$
, i.e. the number of non-identity Pauli terms in $P$.} $k$, its Heisenberg evolution $U_j^\dag P U_j$ contains at most $n^{\calO(k)}$ many distinct Pauli terms and it is classically computable in time $n^{\calO(k)}$.
This condition corresponds to the intuition that each $U_j$ layer is reasonably shallow.
We will call any distribution of circuits that satisfies these two sets of assumptions (i.e., the  locally scrambling assumption and shallow layer assumption) an \textit{$L$-layered locally scrambling circuit} distribution (cf. Definition\ \ref{def:ls-circuit}).

This family of circuits captures a wide range of deep and shallow unstructured parameterised quantum circuits of different topologies.
For example, the possible  ans\"{a}tze include circuits of SU(4) gates with arbitrary connectivity considered in Refs.~\cite{brown2012scrambling, brown2015decoupling, dalzell2021random, harrow2023approximate, zhang2023absence, napp2022quantifying,braccia2024computing}. They also cover the case where each layer $U_j$ is composed of universal single qubit rotations followed by a potentially highly entangling Clifford gates~\cite{letcher2023tight}. Finally they include a large class of Quantum Convolutional Neural Networks (QCNNs) without feed-forward~\cite{pesah2020absence}. However, it does not cover parameterised circuits that have only rotations in a single direction, e.g., a circuit with only parameterised rotations around the $Z$ axis.

Concerning the initial state $\rho$ and the measured observable $O$, we assume that $\Tr[P \rho]$ and $\Tr[P O] / 2^n$ are known (or efficiently computable) for any Pauli observable $P$.
%we distinguish between two different scenarios. We say that an operator $H$ -- which could be a quantum state or an observable -- is classically simulable if, for all Pauli operator $P$, we can exactly compute %$\Tr[PH]$ efficiently on a classical computer.  
%For cases where $\rho$ and $O$ are classically simulable, we provide a classical algorithm for approximating $f_U(O)$ for most circuits $U$. When $\rho$ (or $O$) is not classically simulable, 
When this condition is not met, we show in Supplemental Material~\ref{app:alg-sample-complexity} that $f_U(O)$ can be estimated by augmenting the classical algorithm with classical shadows of the input state $\rho$ and the observable $O$.

\paragraph*{Algorithm.}
\label{sec:alg}
Prior work has proposed
several Pauli propagation methods~\cite{aharonov2023polynomial, shao2023simulating, nemkov2023fourier, beguvsic2023fast, beguvsic2023simulating, fontana2023classical, rudolph2023classical} for classically simulating the family of expectation values described above.

At a high level, Pauli propagation methods work in the Heisenberg picture by back-propagating each Pauli in the observable of interest through the circuit. Clifford gates transform each Pauli to another Pauli, whereas non-Clifford gates generally transform a Pauli to a weighted sum of multiple Pauli operators (i.e., induce branching). In general, the number of branches, and hence the complexity, grows exponentially with the number of non-Clifford gates. However, if at each circuit layer we truncate the set of Pauli terms, %to reduce the number of these terms
then the simulation time and memory can be kept tractable.

More concretely, we consider a truncation scheme where we keep only low-weight Pauli terms. Given a positive integer $k\leq n$, the resulting \emph{low-weight Pauli propagation} algorithm, which we sketch in Fig.~\ref{fig:schematic}, is composed of the following steps.
\begin{enumerate}
    \item Given the observable $O = \sum_{P \in \{I,X,Y,Z\}^{\otimes n}} a_P P$, we compute its low-weight approximation
     $O_{L} := \sum_{ P: \abs{P}\leq k} a_P P$.
    \item For $j=L, L-1,\dots , 2$, we first compute the Heisenberg evolved observable $U^\dag_j O_{j} U_j $ in the Pauli basis,
    which we then use to compute the corresponding $k$-weight approximation:
    \begin{align}
        O_{j-1} := \, &\frac{1}{2^n} \sum_{\substack{P\in \{I,X,Y,Z\}^{\otimes n}:\\  \abs{P}\leq k}} \Tr[U^\dag_j O_{j} U_j P] P. \label{eq:O1-def}
    \end{align}
    In practice, many of the transition amplitudes $ \mathrm{Tr}[U^\dag_j O_{j} U_j P]$'s can be zero if $ O_j$ contains substantially less than $\calO(n^k)$ Pauli operators.
    \item At the end, we compute the final ``truncated'' observable
        ${O}^{(k)}_U := U_1^\dag O_1 U_1$ for the simulation,
    and we compute the inner product $\mathrm{Tr}[{O}^{(k)}_U \rho]$.
\end{enumerate}
This simple truncation strategy has been previously explored numerically in Ref.\ \cite{rudolph2023classical} and studied analytically for noisy circuits in Ref.\ \cite{schuster2024polynomialtimeclassicalalgorithmnoisy}.
Here,  we derive guarantees for low-weight Pauli propagation for the noiseless case. 
Let us denote the approximation of $f_U(O)$ obtained by truncating all Pauli operators with weight greater than $k$ by
\begin{align}
    \tilde{f}_U^{(k)}(O):= \Tr[O^{(k)}_U \rho] \, .
\end{align}
We measure the performance of the estimator $\tilde{f}^{(k)}_U$ in terms of the \emph{mean squared error} (MSE), i.e.
\begin{align}
    \bbE_U \left( \Delta f_U^{(k)}(O) \right):=\bbE_U \left[\left(f_U(O) - \tilde{f}_U^{(k)}(O)\right)^2 \right].
\end{align}
\paragraph*{Main result.}

In Supplemental Material~\ref{app:WeightTrunc}, we show that the average simulation error can be bounded as follows.
\begin{theorem}[Mean squared error]
\label{thm:errorbound}
For $k \geq 0$, we have
    \begin{align}
       & \bbE_U \left( \Delta f_U^{(k)}(O) \right)
        \leq \left(\frac{2}{3}\right)^{k+1} {\norm{O}^2}, \label{eq:msq-error}
    \end{align}
%where $\norm{O}_{\mathrm{Pauli},2} := (2^{-n}\Tr[O^\dag O])^{1/2}$ is the normalized Hilbert-Schmidt norm.
\end{theorem}
\begin{table*}%[h]
    \vspace{-3pt}
    \centering
    \begin{tabular}{|c|c|c|c|}
        \hline
        \textbf{Architecture} & \textbf{Brute-force simulation} & \textbf{Our algorithm}  & \textbf{Our algorithm} \\
        & & ($\epsilon,\delta = \Theta(1)$)& ($\epsilon,\delta = 1/\mathrm{poly}(n)$)   \\
        \hline
        Constant geometric locality & $\exp\left(L^{\mathcal{O}(1)}\right)$  & $L^{\mathcal{O}(1)}$& $L^{\mathcal{O}(\log(n))}$\\
        \hline
        All-to-all connectivity & $\min\{\exp\left(\exp(\mathcal{O}(L))\right) , \exp(\calO(n))\}$ & $L\, n^{\mathcal{O}(1)}$  & $L\, n^{\mathcal{O}(\log(n))}$\\
        \hline
    \end{tabular}
    \caption{\small Comparison of runtimes for brute-force simulation and our algorithm for $L$-layered circuits on different architectures, for additive error $\epsilon$ and success probability $1-\delta$.
    Our algorithm requires polynomial time for arbitrarily small constant precision and failure probability. For inversely polynomial precision and failure probability, the runtime of our algorithm is quasi-polynomial, yet it is substantially more efficient than brute-force simulation. }
    \vspace{-10pt}
    \label{tab:comparison}
\end{table*}

We further bound the time complexity of our algorithm by upper-bounding the number of transition amplitudes $\mathrm{Tr}[U^\dag_{j} P_{j} U_{j} P_{j-1}]$ computed by the algorithm.
In particular, we use the fact that the total number of Pauli operators with weight at most $k$ is $\calO(n^k)$. We then further tighten this bound using a light-cone argument for shallow circuits with bounded geometric dimension. Finally, we combine Theorem~\ref{thm:errorbound} with Markov's inequality to transform the average error statement into a probabilistic statement. Thus, as detailed in Supplemental Material~\ref{app:alg-time-complexity}, we obtain the follow Theorem.

\begin{theorem}[Time complexity]
\label{thm:resources}
    Let $U$ be a randomly sampled circuit from an $L$-layered locally scrambling circuit ensemble on $n$ qubits, and let $O$ be an observable satisfying $\norm{O} \leq 1$.
    There exists a classical algorithm $\mathcal{A}$ that runs in time
    $L \cdot n^{\mathcal{O}\left(\log(\epsilon^{-1} \delta^{-1})\right)},$
    and outputs a value $\mathcal{A}(U)$, such that
    \begin{align}
        \abs{\mathcal{A}(U) - \Tr[U \rho U^\dagger O] } \leq \epsilon,
    \end{align}
    for at least $1 - \delta$ fraction of the circuits.
    If $O$ is a weighted sum of polynomially many Pauli observables and $U$ is a geometrically-local circuits over a $D$-dimensional geometry, the runtime improves to $n^{\calO(1)}L \cdot \min(n, L^D)^{\mathcal{O}\left(\log(\epsilon^{-1} \delta^{-1})\right)}.$
\end{theorem}
 %For the complete proof see Supplemental Material~\ref{app:alg-time-complexity}. 

%Theorem~\ref{thm:resources} establishes that only polynomial resources are required for any small constant error and failure probability. In general, if we demand polynomially small error and failure probability then the complexity scales quasi-polynomially. However, for circuits with constant geometric locality and depth at most poly-logarithmic, the computational time is much more favourable. Namely, the algorithm runs in almost-polynomial time $n^{\calO(\log \log(n))}$ with an exponent that is exponentially smaller than $\mathcal{O}(\log n)$.  For a number of qubits $n \leq  10^{23}$, we have $\log\log(n)< 4$, hence the algorithm is efficient for many practical purposes.

% \paragraph*{Comparison to other methods.} 

Theorem~\ref{thm:resources} establishes that only polynomial resources are required for any small constant error. Although the time complexity is quasi-polynomial for inversely polynomial error, this scaling is much more favourable when compared to brute-force light-cone simulation, as we summarize in Table~\ref{tab:comparison} and discuss in more details in Supplemental Material~\ref{app:lightcone}.
%Moreover, our algorithm runs in almost-polynomial time $n^{\calO(\log \log(n))}$ for circuits with constant geometric locality and depth at most poly-logarithmic. Since for a number of qubits $n \leq  10^{23}$, we have $\log\log(n)< 4$, this result in a manageable computational cost for many practical purposes.
For circuits with highly concentrated expectation values, an extremely easy yet effective classical simulation method is to always guess zero, or more generally, guess $\Tr[O]/ 2^n$. 
%However, Theorem~\ref{thm:main} offers substantial improvements in two key scenarios. (1) When the circuit depth is insufficient for concentration to occur, the \emph{guessing zero} strategy can yield highly inaccurate results.
%(2) The circuit layers $U_j$ may contain adversarially chosen two-qubit gates. In such cases, the circuit ensemble considered in Theorem~\ref{thm:main} does not necessarily produce concentrated expectation values, even at high circuit depths.
However, when the circuit depth is insufficient for concentration to occur, the \emph{guessing zero} strategy can yield highly inaccurate results.
Previous research has shown that certain logarithmic-depth circuits with non-geometrically-local interactions do not exhibit exponentially suppressed variances, yet they are still not expected to be efficiently simulable through brute-force methods\ \cite{napp2022quantifying, zhang2023absence}. Building on these insights, in Supplemental Material\ \ref{app:XQUATH}, we argue that, {on randomly sampled circuits within these families}, low-weight Pauli propagation is significantly more accurate than \emph{guessing zero} and offers super-polynomial speed ups over brute-force light-cone simulation.

\begin{figure*}
    \centering
    %\vspace{-10pt}
    \includegraphics[width=0.98\linewidth]{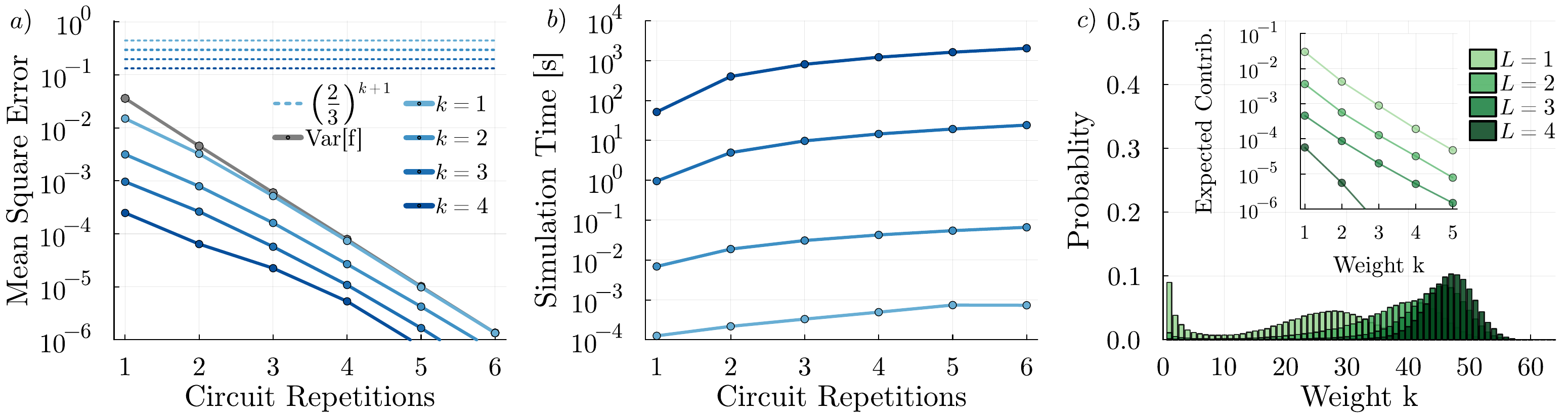}
    \caption{\small \emph{Classical simulation of a local Pauli expectation value with 64 qubits on a $8 \times 8$ grid}. The quantum circuit ansatz consists of {randomly sampled} SU(4) gates in a 2D staircase topology~\cite{zhang2023absence}. After one circuit repetition, the back-propagated observable contains fully global Pauli operators, which is pathological for approaches relying solely on small entanglement light cones. a) Average simulation error as a function of quantum circuit depth for different operator weight truncations. This error is numerically estimated using the Monte Carlo sampling approach. We compare against the general bound in Theorem~\ref{thm:errorbound}. %, which is always satisfied. 
As the parameterised expectation value exponentially concentrates with more circuit repetitions, the variance of the expectation value $\mathrm{Var}[f]$ decays exponentially. {This variance (gray line) is exactly the MSE achieved by the trivial estimator (Supplemental Section~\ref{app:XQUATH}), and therefore constitutes a baseline for quantifying the performance of our algorithm.} The average simulation errors (blue lines) also drop exponentially and becomes more accurate as $k$ increases. {These errors are always better than the trivial estimator but the relative improvement reduces as circuit depth increases.}
% {We remark that the variance (gray line) is exactly the MSE achieved by the trivial estimator (Supplemental Section~\ref{app:XQUATH}), and therefore constitutes a baseline for quantifying the performance of our algorithm.}
b) Simulation time of one expectation value using low-weight Pauli propagation. For example, three circuit repetitions can be simulated on a single CPU thread of an i7-12850HX processor on a laptop to below $10^{-4}$ MSE in approximately $10$ seconds. c) Weight distribution of Pauli operators for up to four circuit repetitions. The inset shows the expected contribution of all operators per weight over the landscape. %This is calculated by multiplying probability values with $(\frac{1}{3})^k$. 
We observe an exponentially decaying contribution of high-weight Pauli operators. 
    }   
    \label{fig:numerics_2d_staircase}
    \vspace{-10pt}
\end{figure*}

\paragraph*{Numerical error analysis.}
While our main result upper bounds the mean squared error for locally scrambling circuits, prior numerical experiments have demonstrated the effectiveness of low-weight truncation strategies for variational ans\"{a}tze that fall beyond our assumptions~\cite{rudolph2023classical, von2022operator, rakovszky2022dissipation}. This hints at the fact that weight truncation is generally even more powerful than stated by our theorems. 
%Indeed, the fact that our bounds are loose is also apparent from the fact that since the error of low-weight Pauli propagation can be trivially upper bounded by the variance of $f_U(O)$, it is suppressed exponentially in the depth for typical random circuits, whereas our upper bound only predicts an exponentially suppression in $k$. 

Fortunately, in a wide range of cases it is efficient to numerically estimate the average error of any Pauli propagation simulation. 
Specifically, in the \emph{End Matter} section, we describe a Monte Carlo method for approximating the mean squared error of any estimator expressed in the Pauli path framework, for a wide family of circuits extending far beyond locally scrambling ensemble. 

We can use such numerical error estimates %guaranteed by Theorem~\ref{thm:certacc} 
to substantiate our theoretical results {for average-case circuits} with an implementation of low-weight Pauli propagation. To do so, we pick an example that was recently reported to be out of reach for tensor network simulation~\cite{zhang2023absence}. The quantum circuit ansatz consists of {randomly sampled} SU(4) unitaries on a 2D grid topology in a so-called \textit{staircase} ordering. This circuit structure implies that the Heisenberg-evolved 1-local Pauli operator on the first (top left) qubit acts on all qubits after one circuit repetition, e.g., sequence of SU(4) gates.

Fig.~\ref{fig:numerics_2d_staircase} shows our simulation results for estimating the expectation value of $\sigma^z_1$ on a square grid of 64 qubits. Even for shallow circuits, exact simulation is not feasible as claimed in Ref.~\cite{zhang2023absence}. However, with a weight truncation of $k=3$, for example, low-weight Pauli propagation achieves MSEs of less than $10^{-3}$ at all depths. These errors are tightly numerically estimated using the Monte Carlo sampling approach in  Theorem~\ref{thm:certacc}. It is clear from Fig.~\ref{fig:numerics_2d_staircase}a) that low-weight Pauli propagation satisfies, and indeed substantially outperforms, the error guarantees provided by Theorem~\ref{thm:errorbound}. 
{We emphasize that the reported MSEs characterize the average-case accuracy of our algorithm by definition, and do not reflect the error behavior on specific, problem-dependent circuits, such as those encountered during the training of variational quantum algorithms. 
For a numerical study of Pauli propagation in such contexts, we refer the reader to Supplemental Section~\ref{app:numerics-IBM}.}
We further see in Fig.~\ref{fig:numerics_2d_staircase}b) that the simulation time can be remarkably quick, increasing exponentially in the weight truncation $k$, but only polynomially with circuit depth. 

The simulation times at low depth are unusually high because the circuit structure from Ref.~\cite{zhang2023absence} is designed to be pathologically hard. More typical circuits with, for example, fewer Pauli gates per 2-qubit block, observables on any qubit other than the first, or with smaller entanglement light cones (e.g., with commuting entangling gates) will be orders of magnitude faster to simulate. The pathological nature of the setup is underlined by Fig.~\ref{fig:numerics_2d_staircase}c), where we show the distribution of Pauli operator weights. Almost all operators become global at a few layers, but their expected contribution to the expectation landscape is suppressed exponentially. This is shown by the inset, which additionally highlights all expected contributions decaying exponentially towards the onset of barren plateaus beyond log-depth. Similar behaviour is observed in the context of %quantum convolutional neural networks 
QCNNs in our accompanying paper Ref.~\cite{bermejo2024quantumconvolutionalneuralnetworks}.

%We remind the reader that Theorem~\ref{thm:errorbound} and Theorem~\ref{thm:resources} do not hold for quantum circuits consisting of circuit layers that are not locally scrambling. 
Additionally, we numerically demonstrate that low-weight Pauli propagation works also for practical examples of quantum circuits that do not meet the ``locally scrambling'' assumptions. 
Specifically, we can often numerically show that our general bound in Theorem~\ref{thm:errorbound} is satisfied and the average error strongly (potentially exponentially) decreases with $k$. Supplemental Fig.~\ref{fig:numerics-correlated} provides evidence that low-weight Pauli propagation may indeed be a fruitful method for simulating quantum circuits with correlated angles. Here we recorded the MSE on a 16-qubit example over RX, RZ, and RZZ Pauli rotations parametrized by a single random angle. %Note that the numerical monte carlo estimation of the error at large scales is not valid for correlated angles. 
Furthermore, in Supplemental Fig.~\ref{fig:numerics-IBM} of Supplemental Material~\ref{app:numerics-IBM} we provide an example of a quantum circuit with non-locally scrambling layers on 127 qubits. The gates are generated by a transverse field Ising model, i.e. RX and RZZ gates, on a heavy-hex topology. This indicates that low-weight Pauli propagation could be a powerful classical simulation approach beyond our theoretical results.

\paragraph*{Discussion.}
Our main result, Theorem~\ref{thm:resources}, establishes that it is possible to classically estimate the expectation values of a large class of quantum circuits. Many parameterised quantum circuits used in variational quantum algorithms satisfy our assumptions
%(either exactly or approximately)
{when randomly initialized}, including some that claim to avoid barren plateaus while escaping classical simulability. This hints at the fact that the current approach to VQAs is too generic  and needs to be revised, 
and complements the conjecture connecting the provable absence of barren plateaus and classical simulability~\cite{cerezo2023does}. In particular, our benchmark in Fig.~\ref{fig:numerics_2d_staircase} shows that a {randomly initialized} circuit that has been proven barren plateau free, and that is out of reach of tensor network simulations, can be simulated via low-weight Pauli propagation. 

We reiterate that our analysis, similarly to current trainability analyses, is an average case analysis. Thus even within a family of parameterised quantum circuits for which our assumptions hold, there will be some circuits, i.e., certain rotation angles, for which our algorithm fails. Our accompanying work in Ref.~\cite{bermejo2024quantumconvolutionalneuralnetworks} numerically showcases that, despite this caveat, low-weight Pauli propagation simulations of %Quantum Convolutional Neural Networks (
QCNNs %) 
can outcompete the classification performance of faithfully simulated QCNNs on all published benchmarks so far up to 1024 qubits. 

{It remains an open question whether, for certain computational tasks, variational quantum circuits when initialized in a random, classically tractable region can evolve toward a classically intractable region that also corresponds to optimal values of the target cost function.
In Supplemental Section~\ref{app:numerics-IBM}, we deploy Pauli propagation to train a variational quantum eigensolver, witnessing a reduced accuracy near low-energy values. Nonetheless, our numerical results provide evidence of the potential utility of Pauli propagation, such as using it as a quantum-inspired optimization algorithm or as a classical method to learn quantum circuits to prepare approximate ground states on quantum hardware.
}

As expected given the generic hardness of simulating quantum systems, low-weight Pauli propagation methods have inherent limitations. For example, circuits with random rotations in a single direction, e.g., purposefully chosen RX rotations, do not generally fall under the current analysis. This is a fundamental limitation because then it is possible to encode problem instances into the circuit that are hard on average~\cite{huang2023learning, cerezo2023does, gil-fuster2024understanding}. Nonetheless, in Supplemental Material~\ref{app:numerics-IBM} we numerically show that low-weight Pauli propagation does seemingly work well on more practical (non-pathological) examples of such circuits. A similar phenomenon can be observed for quantum simulation. While Trotterization-based time evolution introduces correlated angles that are not covered by our analysis, our numerical evidence (Supplemental Fig.~\ref{fig:numerics-correlated}) suggests that low-weight Pauli propagation may be valuable for trading off simulation runtime and accuracy.

Crucially, in this work we have considered classical algorithms for estimating expectation values, rather than sampling. 
In Supplemental Material~\ref{app:XQUATH}, we also prove that, for circuits of linear depth, low-weight Pauli propagation cannot be used for refuting the XQUATH %(Linear Cross-Entropy Quantum Threshold) 
conjecture~\cite{aaronson2019classical, gao2024limitations, aharonov2023polynomial, tanggara2024classically}, which is closely related to sampling-based quantum supremacy experiments. 
In addition, a recent work showed a no-go result for approximating the Heisenberg evolution of an observable under a circuit generating high levels of magic~\cite{dowling2024magic}. %Our finding here complements previous results in Refs.~\cite{gao2024limitations, aharonov2023polynomial, tanggara2024classically}, where it was demonstrated that Pauli-path approaches can refute this conjecture for circuits of sub-linear depth. 
This highlights that approximating the Heisenberg evolution and predicting expectation values are inherently distinct tasks, as our results show that the latter may be classically easy even with high levels of magic (or entanglement).

Going forward, and despite some of the limitations highlighted above, we expect low-weight Pauli propagation to be an effective strategy in regimes well beyond those that we have so far been able to analytically guarantee. It could, for example, be combined with more problem-specific truncations that allow simulating edge cases escaping our guarantees~\cite{nemkov2023fourier,du2024efficient,beguvsic2023fast,shao2023simulating,rudolph2023classical,goh2023lie}. Symmetries and \textit{a priori} knowledge about the initial state, the circuit, and the observable will allow classical simulations to push further and challenge upcoming quantum devices.

%\end{document}

\section{End matter}
We briefly discuss the main technical tools employed in our derivations. %For a complete derivation we refer the reader to the Appendices. 
Given a Hermitian operator 
\begin{align}
    H = \sum_{P\in\{I,X,Y,Z\}^{\otimes n}} a_P P,
\end{align}
we define its associated Pauli 2-norm $\norm{H}_{\mathrm{Pauli},2}$ (also known as normalized Hilbert-Schmidt norm) as
\begin{align}
    \norm{H}_{\mathrm{Pauli},2}^2 = \sum_{P\in\{I,X,Y,Z\}^{\otimes n}} a_P^2.
\end{align}
We remark that $\norm{H}_{\mathrm{Pauli},2} \leq \norm{H}$, and in particular $\norm{P}_{\mathrm{Pauli},2} =1$ for all Pauli operators $P$. Moreover, $\norm{\ketbra{\psi}}_{\mathrm{Pauli},2} = 2^{-n/2}$ for all projectors $\ketbra{\psi}$.
\medskip
\paragraph*{Locally scrambling unitaries.}
There are a number of important properties of locally scrambling circuits that aid our analysis. First, for any locally scrambling unitary $V$ and all Pauli operators $P\neq Q$ we have
\begin{align}
    &\bbE_{V} \left[V^{\otimes 2\dag}(P \otimes Q)V^{\otimes 2}\right] = 0 &\text{(orthogonality)} \label{eq:end-ortho}.
\end{align}
This identity is extremely useful for evaluating second moments of observables expressed in the Pauli basis, e.g. $O =\sum_P a_P P$. In particular, we have
\begin{align}
    &\bbE_{V} \Tr[O V \rho V^\dag]^2 \\= &\sum_{P\in\{I,X,Y,Z\}^{\otimes n}} a_P^2 \bbE_{V} \Tr[P V \rho V^\dag]^2.
\end{align}
Moreover, we also have the following property
\begin{align}
    &\bbE_{V} (V^{\dag}PV)^{\otimes 2} =&\\ &\frac{1}{3^{\abs{P}}} \sum_{\substack{Q:\\\mathrm{supp}(Q)=\mathrm{supp}(P) }} \bbE_{V} (V^{\dag}QV)^{\otimes 2}&\text{(Pauli-mixing)}.\nonumber
\end{align}
Leveraging the Pauli-mixing property, \citet{huang2022learning} showed that
\begin{align}
    \bbE_{V} \left[\Tr[P V \rho V^\dag]^2 \right] \leq \left(\frac{2}{3}\right)^{\abs{P}} \label{eq:damping}.
\end{align}
Hence, the contribution of high-weight Pauli operators to expectation values is exponentially suppressed on average.
Therefore, we can approximate the expectation value of the observable $O =\sum_P a_P P$ with that of the truncated observable
$O^{(k)} =\sum_{P:\abs{P}\leq k} a_P P$, to give:
\begin{align}
    &\bbE_{V} \left[\Tr[(O-O^{(k)} ) V \rho V^\dag]^2 \right] \nonumber
    \\&\bbE_{V} \left[\Tr[(O ) V \rho V^\dag]^2 - \Tr[(O^{(k)} ) V \rho V^\dag]^2 \right] \label{eq:methods-wt}
    \\ \leq &\left(\frac{2}{3}\right)^{k+1}\left(\norm{O}_{\mathrm{Pauli},2}^2-\norm{O^{(k)}}_{\mathrm{Pauli},2}^2\right) .  \nonumber 
\end{align}
In the following, we discuss how this low-weight approximation for expectation values over a single locally scrambling layer can be extended to a circuit composed of a product of locally scrambling circuit layers using Pauli path integrals.  

\medskip

\paragraph*{Pauli path integral.}
Recall that any Hermitian operator $H$ can be expanded in the (normalized) Pauli basis $\calP_n = \left\{\frac{I}{\sqrt{2}},\frac{X}{\sqrt{2}},\frac{Y}{\sqrt{2}},\frac{Z}{\sqrt{2}}\right\}^{\otimes n}$ as
\begin{align}\label{eq:innerprod}
    H = \sum_{s\in\calP_n} \Tr[Hs]s \, .
\end{align}
In order to compute the expectation value $f_U(O) := \Tr[O U\rho U^\dag]$, we consider the Heisenberg evolution of the observable $O$, i.e. the evolution under the adjoint unitary channel $U^\dag(\cdot) U$. Applying iteratively Eq.~\eqref{eq:innerprod}, we obtain
\begin{align}
    \nonumber U^\dag O U & =  \sum_{s_0,\dots,s_L\in \calP_n}\left( \Tr[Os_L]\prod_{j=1}^L \Tr[U_{j}^\dag s_j U_j s_{j-1}]s_0 \right)
    \nonumber \\ & = \sum_{\gamma \in \mathcal{P}_n^{L+1}} \Phi_\gamma(U) s_\gamma. 
\end{align}
Here we labeled each Pauli path by a string $\gamma = (s_0,\dots,s_L)$; we denoted the associated Fourier coefficient by $\Phi_\gamma(U)  := \Tr[Os_L]\prod_{j=1}^L \Tr[U_{j}^\dag s_j U_j s_{j-1}]$ and we defined $s_\gamma:=s_0$.
The Pauli path integral is the summation obtained by projecting the Heisenberg evolved observable onto the initial state:
\begin{align}\label{eq:PathIntExp}
    &f_U(O)=\Tr[U^\dag O U\rho] = \sum_{\gamma \in \mathcal{P}_n^{L+1}} \Phi_\gamma(U) d_\gamma \, ,
\end{align}
where we defined $d_\gamma:=\Tr[s_\gamma \rho].$

\medskip

\paragraph*{Low-weight Pauli Propagation.}
While the Pauli path integral contains exponentially many terms, for most circuits $U$ we can approximate its value by restricting the computation to a small, carefully chosen subset of Pauli paths.
In particular, given an integer $k\geq0$, we consider the following subset $\calS_k \subseteq \calP_{n}^{L+1}$,
    \begin{align}
        \calS_k := \{\gamma =(s_0,s_1,\dots,s_L) \;|\; \forall i\neq0 : \;\abs{s_i}\leq k  \}.
    \end{align}
In addition, we  also define the following ``truncated'' observable and the associated expectation value:
\begin{align}
    &O^{(k)}_U := \sum_{\gamma \in \calS_k} \Phi_\gamma(U) s_\gamma,
    \\&\tilde{f}_U^{(k)}(O) = \Tr[O^{(k)}_U \rho].
\end{align}

The proof of Theorem~\ref{thm:errorbound} leverages two key properties of locally scrambling unitaries. 
As shown in Eq.\ \ref{eq:damping}, high-weight Pauli operators yield negligible expectation values when measured on states post-processed by locally scrambling unitaries. This property enables us to upper bound the error incurred at each iteration of the low-weight Pauli propagation algorithm. Additionally, the orthogonality property (Eq.\ \ref{eq:end-ortho}) implies that Fourier coefficients associated with different paths are uncorrelated,
\begin{equation}
    \mathbb{E}_U [f(U, O, \gamma) f(U, O, \gamma')] = 0 \label{eq:ortho-end},
\end{equation} 
whenever $\gamma \neq \gamma'$. 
As previously observed in\ \cite{aharonov2023polynomial, shao2023simulating, fontana2023classical}, this property 
-- commonly referred as ``orthogonality of Pauli paths'' --
drastically simplifies the expression of the mean-squared-error
\begin{align}
    \bbE_U \Delta f_U^{(k)} =  &\sum_{\gamma \in \calP_n^{L+1}\setminus\calS_k} \Phi_\gamma(U)^2 d_\gamma^2.
\end{align}
%from which Eq.~\eqref{eq:mseexpansion} straightforwardly follows. 
By making use of a telescoping sum, and applying iteratively Eq.\ \eqref{eq:methods-wt}, we find that total error satisfies 
\begin{align}
    \bbE_U \Delta f_U^{(k)} \leq \left(\frac{2}{3}\right)^{k+1}\norm{O}^2_{\mathrm{Pauli},2}
\leq \left(\frac{2}{3}\right)^{k+1}\norm{O}^2,
\end{align}
as stated in Theorem\ \ref{thm:errorbound}.

It remains to study for which values of $k$ and $L$ the truncated path integral can be computed efficiently.
To this end, we will need to upper bound the number of Pauli operators supported on a subset of qubits of size $M$ and weight at most $k$. We will use the following upper bound (derived in more detail in the appendices):
\begin{align}
    \sum_{\ell=0}^k 3^\ell \binom{n}{\ell} 
     \leq \left(\frac{3en}{k}\right)^k,
    \label{eq:countingbound}
\end{align}
For each consecutive layer $j+1$ and $j$, and for all $s_j, s_{j+1}$ such that $\abs{s_j},\abs{s_{j+1}}\leq k$ we need to compute the associated transition amplitude, i.e. 
\begin{align}
    \Tr[U^\dag_{j+1} s_{j+1} U_{j+1} s_j].
\end{align}
From Eq.~\eqref{eq:countingbound}, there are $\calO(n^{2k})$ pairs of Pauli operators $(s_{j+1},s_j)$ that satisfy the weight constraint.
Hence for a circuit of depth $L$, 
we need to compute at most
\begin{align}
      \mathcal{O}(n^{2k} L) = n^{\calO\left(\log(\epsilon^{-1} \delta^{-1})\right)} L
\end{align}
transition amplitudes to obtain a small error $\epsilon \norm{O}_{\mathrm{Pauli},2}$ with probability at least $1-\delta$. Therefore, the truncated Pauli path integral can be computed efficiently (i.e., in $\text{poly}(n)$ time) provided that $L =\mathrm{poly}(n)$ and $\epsilon,\delta = \Theta(1)$. %This proves the first part of Theorem~\ref{thm:resources}.

\medskip

\paragraph*{Certified error estimate.}
While our upper bounds apply to locally scrambling circuits, we also develop a Monte Carlo method to numerically estimate the mean squared error of low-weight Pauli propagation. This method can be applied to any ensemble of circuits with orthogonal Pauli paths, as described in Eq.\ \eqref{eq:ortho-end}.
This family of circuits extends well beyond locally scrambling circuits: examples of circuit ensembles with orthogonal Pauli paths include parameterised quantum circuits with random single qubit rotations only in a single direction, e.g. only parameterised  $R_Z(\theta)$ rotations.
To this end, we define the following distribution over the Pauli paths:
\begin{align}
    p(\gamma) :=  \bbE_U\Phi_\gamma(U) ^2/\norm{O}_2^2,
\end{align}
where the orthogonality property ensures that $ \sum_{\gamma \in \calP_n^{L+1}} p(\gamma) = 1$.

We observe that the mean squared error can be rearranged as follows
\begin{align}
    \bbE_U\Delta f^{(k)}_U = %\bbE_U \sum_{\gamma \in \overline\calS_k} \Phi_\gamma(U) ^2 d_\gamma^2 := 
    \sum_{\gamma \in \calP_n^{L+1}} p(\gamma) X_\gamma^{(k)},
\end{align}
where we introduced the following variable:
\begin{align}
    X_\gamma^{(k)} = \begin{cases}
    0 &\text{if $\gamma \in \calS_k$,}
       \\ \norm{O}_2^2 \cdot d_\gamma^2 &\text{if $\gamma \not\in \calS_k$,}
    \end{cases}
\end{align}
Moreover, we observe that $0 \leq X_\gamma^{(k)} \leq \norm{O}^2_{\mathrm{Pauli},2}$.
Therefore, invoking standard concentration of measure inequalities, we can prove that the mean squared error $\bbE_U\Delta f^{(k)}_U $ can be  estimated with additive error $\epsilon  \norm{O}^2_{\mathrm{Pauli},2}$ sampling $\calO(\epsilon^{-2})$ Pauli paths $\gamma$ and computing the mean of the associated function $X_\gamma^{(k)}$. This result is formalized in the following theorem, which we rigorously prove in Section\ \ref{app:certacc}.

\begin{theorem}[Certified error estimate]\label{thm:certacc}
Let $U$ be a circuit sampled from an $L$-layered ensemble with orthogonal Pauli paths. %\footnoteref{fn-ortho}
Assume that we can sample $s\in\calP_n$ with probability ${\Tr[Os]^2}/{\norm{O}^2_{2}}$ in time $\mathrm{poly}(n)$.
Moreover,  assume that for $j = L, L-1, \dots, 1$, and for all $s_j\in \calP_n$, we can sample $s_{j-1}$ with probability
$ \bbE_{U_j}\mathrm{Tr}[U^\dag_j s_jU_j s_{j-1}]^2$ in time $\mathrm{poly}(n)$.
Then, for any $\epsilon,\delta\in(0,1]$, there exists a classical randomized algorithm that runs in time $\mathrm{poly}(n) L\epsilon^{-2}\log(\delta^{-1})$ and outputs a value $\alpha$ such that
\begin{align}
    \left|{\alpha - \bbE_U\Delta f^{(k)}_U}\right| \ \leq \epsilon \norm{O}.
\end{align}
with probability at least $1-\delta$.
\end{theorem}
\section{Acknowledgements}
The authors thank Sergio Boixo, Su Yeon Chang, Soonwon Choi, Soumik Ghosh, Sacha Lerch, Jarrod R McClean, Antonio Anna Mele, Thomas Schuster and Yanting Teng for valuable discussions and feedbacks.
AA and ZH acknowledge support from the Sandoz Family Foundation-Monique de Meuron program for Academic Promotion. AS acknowledges support by the Simons Foundation (MP-SIP-00001553, AWH) and NSF grant PHY-2325080. MC acknowledges supportw by the Laboratory Directed Research and Development (LDRD) program of Los Alamos National Laboratory (LANL) under project numbers 20230527ECR and 20230049DR.  This work was also supported by LANL's ASC Beyond Moore’s Law project.

\bibliography{quantum, ref}

\clearpage 

\appendix
\onecolumngrid

\section{Preliminaries}
\subsection{Notation and basic definitions}
We briefly introduce some notations and conventions used in the present work.
\begin{itemize}
    \item \textbf{Linear operators and associated norms.} Let $\calL(\bbC^{d})$ be the set of linear operators that act on the $d$-dimensional complex vector space $\mathbb{C}^d$.
For two matrices $A,B\in \calL(\bbC^{d})$ we denote their Hilbert-Schmidt inner product as $\mathrm{Tr}[A^\dag B]$.
Furthermore, for a matrix $A\in \calL(\bbC^{d})$, the induced Hilbert-Schmidt norm is denoted by $\norm{A}_2\coloneqq \mathrm{Tr}[A^\dagger A]^{1/2}$.
We define the operator norm of $A$ as $\norm{A}:= \sup_{\norm{\ket{\psi}}_2=1}\norm{A\ket{\psi}}_2$. We recall that the normalized Hilbert-Schmidt norm is always upper bounded by the operator norm, that is 
$\norm{A}_{\mathrm{Pauli},2} \leq \norm{A}$.
Given $p>0$ and an Hermitian operator $A$, we define the associated Pauli-$p$ norm as
\begin{align}
    \norm{A}_{\mathrm{Pauli},p} := \left(\sum_{P\in\{I,X,Y,Z\}^{\otimes n}}\abs{a_P}^p \right)^{1/p}. 
\end{align}
Note that the Pauli-$2$ norm equals the normalized Hilbert-Schmidt norm: $\norm{A}_{\mathrm{Pauli},2} = 2^{-n/2} \norm{A}_{2} $.

\item \textbf{Ensemble of unitaries.} We denote by $\bbU(d)$ the $d$-dimensional unitary group. 
Given a function \( F : \mathbb{U}(d) \rightarrow \mathbb{R} \) and a distribution $\calD$ over $\bbU(d)$, we denote the expected value of \( F(U) \) with respect to \(\mathcal{D}\) as \(\mathbb{E}_{U \sim \mathcal{D}} F(U)\). For simplicity, we will write \(\mathbb{E}_{U} F(U)\) when the distribution \(\mathcal{D}\) is clear from the context.
We denote by \( U_1, U_2,\dots,U_k \sim \mathcal{D} \) that $U_1, U_2,\dots,U_k$ are independently sampled from the distribution \(\mathcal{D}\).
Furthermore, we denote by $\mathcal{U}(d)$ the Haar measure over $\bbU(d)$ and by $\mathrm{Cl}(d)$ the uniform distribution over the $d$-dimensional Clifford group. 

\item \textbf{Pauli basis.}
For a Pauli operator $P \in \{I,X,Y,Z\}^{\otimes n}$, we denote by $\mathrm{supp}(P)$ the set of qubits on which the string $P$ is non-identity and by $\abs{P} := \abs{\mathrm{supp}(P)}$ the corresponding weight.
For an operator $O=\sum_P a_P P$ we define its support as
\begin{align}
    &\mathrm{supp}(O) := \bigcup_{\substack{P:\\ a_P\neq 0}}  \mathrm{supp}(P),
    \\&\abs{O}:= \abs{\mathrm{supp}(O)}.
\end{align}
Therefore, we say that $O$ is supported on $k$ qubits if $\abs{O}=k$.
An Hermitian operator $O$ is $k$-local if it contains only Pauli operators with weight at most $k$: $O = \sum_{P:\abs{P}\leq k} a_P P$. We stress that this notion of locality does not mean that the operator is geometrically local but rather that it is low weight.

\medskip

We further make use of the \textbf{normalized} Pauli basis:
\begin{align}
    \calP_n = \left\{\frac{I}{\sqrt{2}},\frac{X}{\sqrt{2}},\frac{Y}{\sqrt{2}},\frac{Z}{\sqrt{2}}\right\}^{\otimes n} \subseteq \calL(\bbC^{2^n}),
\end{align}
whose elements are orthonormal with respect to the Hilbert-Schmidt inner product, i.e. $\forall s,t\in\calP_n:\Tr[st] = \delta_{st}$.
As a consequence, for any Hermitian operators $H,H'$ we have
\begin{align}
    &H = \sum_{s\in\calP_n} \Tr[Hs]s, \label{eq:pauliexp}
    \\&\Tr[HH'] = \sum_{s\in\calP_n} \Tr[Hs]\Tr[sH'],  \label{eq:innerproduct_app}
    \\&\norm{H}_2^2 = \Tr[H^2]= \sum_{s\in\mathcal{P}_n}\Tr[Hs]^2 \label{eq:2norm}.
\end{align}
Throughout this work, we will use both the Pauli basis $\{I,X,Y,Z\}^{\otimes n}$ and the normalized Pauli basis $\calP_n$, alternating between them to avoid introducing unnecessary renormalization factors. To ensure clarity, we consistently denote non-normalized Pauli operators with uppercase letters and normalized Pauli operators with lowercase letters.

\item \textbf{Useful relations.}
We will further make regular use of the identity
\begin{equation}
    \text{Tr}[A \otimes B ] = \text{Tr}[A]\text{Tr}[B]
\end{equation}
and hence $\text{Tr}[A^{\otimes 2}] = \text{Tr}[A]^2$. 
\end{itemize}

\subsection{Properties of locally scrambling unitaries}

We start by introducing the definition of locally scrambling unitaries.

\begin{definition}[Locally scrambling distribution]
    A distribution $\calD$ over $\bbU(2^n)$ is locally scrambling if we have %for all Hermitian operators $H,H'$ we have
    %\begin{align}
    %    \bbE_{U\sim\calD} \left[U^{\otimes 2}(H\otimes H')U^{\dag\otimes 2}\right]=\bbE_{U\sim\calD}\;\bbE_{V_1,V_2,\dots, V_n \sim \mathrm{Cl}(2) } \left[\bigotimes_{i=1}^n V_i^{\otimes 2}\right] U^{\otimes 2}(H\otimes H')  U^{\dag\otimes 2} \left[\bigotimes_{i=1}^n V_i^{\dag\otimes 2}\right],
    %\end{align}
    \begin{align}
       \bbE_{U\sim\calD} \left[U^{\otimes 2} \otimes U^{*\otimes 2}\right]=\bbE_{U\sim\calD}\;\bbE_{V_1,V_2,\dots, V_n \sim \mathrm{Cl}(2) } \left[\bigotimes_{i=1}^n V_i^{\otimes 2}\right] U^{\otimes 2} \otimes  U^{*\otimes 2} \left[\bigotimes_{i=1}^n V_i^{*\otimes 2}\right],
    \end{align}
    that is, $\calD$ is invariant under post-processing by independently sampled random single-qubit Cliffords up to the second moment.
We say that a random unitary $U$ is locally scrambling when it is sampled from a locally scrambling distribution.    
\end{definition}
 We note that the above definition of locally scrambling is slightly different to the definitions used in Refs.~\cite{kuo2020markovian, hu2021classical, caro2022outofdistribution, huang2022learning} and is close to the notion of `locally scrambled up to the second moment' used in Ref.~\cite{caro2022outofdistribution}.

We further define the family of circuits obtained by combining multiple locally scrambling unitaries.
 \begin{definition}[Locally scrambling circuit]
 \label{def:ls-circuit}
A random circuit $U=U_L U_{L-1}\dots U_1$ is sampled from an $L$-layered locally scrambling circuit distribution if the circuit layers $U_j$ are independently sampled and they satisfy the following properties:
\begin{enumerate}
    \item $U_1$ is sampled from a locally scrambling distribution over $\bbU(2^n)$;
    \item for each $j=2,\dots, L$,  let $U_j$ be expressed as $V_{A_j}\otimes I_{[n]\setminus A_j}$, where $A_j \subseteq [n]$ is the subset of qubits upon which $U_j$ acts non trivially. We assume that $V_{A_j}$ is sampled from a locally scrambling distribution over $\bbU\left(2^{\abs{A_j}}\right)$.
    \item For any Pauli operator $P$ and $j=1,2,\dots, L$, the Heisenberg evolved observable $U_j^\dag P U_j$ contains at most $n^{\calO(\abs{P})}$ many distinct Pauli terms and it is classically computable in time $n^{\calO(\abs{P})}$.
\end{enumerate}
We say that a random circuit is locally scrambling when it is sampled from a locally scrambling circuit distribution.
 \end{definition}

We emphasize that, while the first condition allows us to express our technical arguments more concisely, it is inessential for our results to hold. All our proofs can be adapted to circuits where the initial layer $U_1$ acts only on a certain subset of qubits. We remark that similar conventions are also widely employed in the literature about random quantum circuits, where the first circuit layer is typically the tensor product of random single-qubit gates\ \cite{dalzell2021random, dalzell2022randomquantum, aharonov2022polynomial}.
Furthermore, the first condition is always satisfied if $U$ is a Quantum Convolutional Neural Networks without feed-forwards with gates sampled from local 2-designs, such as those considered in our companion paper Ref.\ \cite{bermejo2024quantum}.
Moreover, the last condition is only used for our computational complexity analysis, and plays no role in the error analysis for a fixed truncation order $k$. Intuitively, this condition corresponds to each layer $U_j$ being sufficiently shallow.

We also prove the following technical lemma, which we will employ in the proof of our main result.
\begin{lemma}
\label{lem:ls-subcircuit}
    Let $U=U_L U_{L-1}\dots U_1$ be an $L$-layered locally scrambling circuit. Then for all $j=1,2,\dots, L$, the unitary $U_j U_{j-1}\dots U_1$ is sampled from a locally scrambling distribution over $\bbU(2^n)$.
\end{lemma}
\begin{proof}
    The Lemma can be easily proven by induction over $j=1,2,\dots, L$. The base step, corresponding to the case  $j=1$, follows directly by Definition\ \ref{def:ls-circuit}.
    Moreover, for all $j=1,2, \dots, L$, Definition\ \ref{def:ls-circuit} implies that
    \begin{align}
    \label{eq:ls1}
        \bbE_{U_j}\left[U_j^{\otimes 2}\otimes U_j^{*\otimes 2}\right]
        = \bbE_{U_j} \bbE_{V_{i_1},V_{i_2},\dots, V_{i_{\abs{A_j}}} \sim \mathrm{Cl}(2) } \left[I_{[n]\setminus A_j}\otimes \bigotimes_{i\in A_j} V_i^{\otimes 2}\right] U_j^{\otimes 2} \otimes  U_j^{*\otimes 2} \left[I_{[n]\setminus A_j}\otimes \bigotimes_{i \in A_j} V_i^{*\otimes 2}\right],
    \end{align}
    
    It remains to prove the inductive step. Assume that $U_{j-1}U_{j-2}\dots U_1$ is locally scrambling over $\bbU(2^n)$:
    \begin{align}
    \label{eq:ls2}
       \bbE_{U_{j-1},\dots, U_1}\left[(U_{j-1}\dots U_1)^{\otimes 2} \otimes (U_{j-1}\dots U_1)^{*\otimes 2}\right] 
       \\= \bbE_{U_{j-1},\dots, U_1}\;\bbE_{V_1,V_2,\dots, V_n \sim \mathrm{Cl}(2) } \left[\bigotimes_{i=1}^n V_i^{\otimes 2}\right] (U_{j-1}\dots U_1)^{\otimes 2} \otimes  (U_{j-1}\dots U_1)^{*\otimes 2} \left[\bigotimes_{i=1}^n V_i^{*\otimes 2}\right],
    \end{align}
    Combining Eq.\ \ref{eq:ls1} and Eq.\ \ref{eq:ls2}, we obtain that $U_j U_{j-1}\dots U_1$ is invariant (up to second moment) under post-processing by random single-qubit Cliffords, and thus it is also locally scrambling over $\bbU(2^n)$.
\end{proof}

Our technical results leverage extensively two key properties of locally scrambling unitaries, which we prove in the following lemma.
\begin{lemma}[Orthogonality and Pauli-mixing]
\label{lemma:ls}
Let $U$ be a locally scrambling  unitary. Then for all $P,Q\in\{I,X,Y,Z \}^{\otimes n}$, we have
\begin{align}
    \\&\bbE_{U} [U^{\dag\otimes 2} (P \otimes Q) U^{\otimes 2}] = \begin{cases}
        0 & \text{if $P\neq Q$ \;\;\;\;(orthogonality)} 
        \\ \frac{1}{3^{\abs{P}}}\sum_{\substack{P\in\{I,X,Y,Z \}^{\otimes n}:\\\mathrm{supp}(P)=\mathrm{supp}(Q)}} \bbE_U\left[U^{\dag\otimes 2}P^{\otimes 2} U^{\otimes 2}\right] 
        &\text{if $P = Q$ \;\;\;\;(Pauli-mixing)} 
    \end{cases}
\end{align}
\end{lemma}
\begin{proof}
Let $P = \bigotimes_{i=1}^n P_i$ and $Q = \bigotimes_{i=1}^n Q_i$. 
We have for all $i\in[n]$
\begin{align}
  &\bbE_{V_i\sim \mathrm{Cl}(2) } V_i^{\dag \otimes 2} (P_i\otimes Q_i)  V_i^{\otimes 2}
  = \bbE_{V_i\sim \mathcal{U}(2) } V_i^{\dag \otimes 2} (P_i\otimes Q_i)  V_i^{\otimes 2} 
  \\= &I^{\otimes 2} \left(\frac{\Tr[P_i\otimes Q_i] - \frac{1}{2} \Tr[P_iQ_i]}{3}\right) + \mathrm{SWAP} \left(\frac{\Tr[P_i Q_i] - \frac{1}{2}  \Tr[P_i\otimes Q_i]}{3}\right),
  \label{eq:2-design}
\end{align}
where we employed the fact that the Clifford group forms an exact 2-design and we denoted the `swap' operator by $\mathrm{SWAP}:= \frac{1}{2}(I^{\otimes 2}+X^{\otimes 2}+Y^{\otimes 2}+Z^{\otimes 2})$.
We further notice that 
\begin{align}
    &\Tr[P_i \otimes Q_i] = \begin{cases}
        4 & \text{ if $P_i=Q_i = I$ },\\
        0 & \text{ otherwise},
    \end{cases}
    \\
    &\Tr[P_i Q_i] = \begin{cases}
        2 & \text{ if $P_i=Q_i$ },\\
        0 & \text{ otherwise}.
    \end{cases}
\end{align}
Hence Eq.\ \eqref{eq:2-design} can be simplified as follows
\begin{align}
    \bbE_{V_i\sim \mathrm{Cl}(2) } V_i^{\dag \otimes 2} (P_i\otimes Q_i)  V_i^{\otimes 2}=
    \begin{cases}
        I^{\otimes 2} & \text{if $P_i=Q_i=I$},
        \\\frac{1}{3} (X^{\otimes 2} + Y^{\otimes 2} + Z^{\otimes 2})  & \text{if $P_i=Q_i\neq I$},
        \\0  & \text{if $P_i\neq Q_i$}.
    \end{cases}
    \label{eq:cliffs}
\end{align}
Moreover, the definition of locally scrambling distribution yields
\begin{align}
    \bbE_{U} [U^{\dag\otimes 2} (P \otimes Q) U^{\otimes 2}] = \bbE_{U} U^{\dag\otimes 2} \left[  \bigotimes_{i=1}^n \bbE_{V_i\sim \mathrm{Cl}(2) }V_i^{\dag\otimes 2} (P_i \otimes Q_i) V_i^{\otimes 2}\right]U^{\otimes 2} \label{eq:ls}.
\end{align}
And thus plugging Eq.\ \eqref{eq:cliffs} in Eq.\ \eqref{eq:ls} gives the desired result.
\end{proof}
The orthogonality property comes in useful when evaluating second moments of the inner product expanded in the Pauli basis, as the cross terms vanish. In particular, we prove the following lemma.

\begin{lemma}[Vanishing cross-terms]
\label{lemma:ortho}
    Let $H, H'$ be a Hermitian operators and $U$ a locally scrambling unitary. We have
    \begin{align}
        \bbE_{U} \Tr[ H U H'U^\dag]^2 
        = \bbE_{U} \sum_{s\in\calP_n} \Tr[Hs]^2\Tr[sU H'U^\dag]^2.
    \end{align}
\end{lemma}
\begin{proof}
We start by expanding the inner product in the normalized Pauli basis using Eq.~\eqref{eq:pauliexp}:
\begin{align}
    \bbE_{U} \Tr[ H U H'U^\dag]^2 =& \bbE_{U} \left( \sum_{s\in\calP_n} \Tr[H s]\Tr[sU H'U^\dag]\right)^2
    \\= &\bbE_{U} \sum_{s,t\in\calP_n} \Tr[H s]\Tr[sU H'U^\dag]\Tr[H t]\Tr[tU H'U^\dag]
    \\ = &\bbE_{U} \sum_{s,t\in\calP_n} \Tr[ H'^{\otimes 2} U^{\dag\otimes 2} (s \otimes   t)  U^{\otimes 2}]\Tr[Hs]\Tr[Ht]
    \\ = &\bbE_{U} \sum_{s\in\calP_n} \Tr[Hs]^2\Tr[s U H U^\dag]^2,
\end{align}
where in the final line we used that $\bbE_U U^{\dag\otimes 2} (s \otimes   t)  U^{\otimes 2} = 0$ for $s\neq t$.
\end{proof}
As an application of Lemma\ \ref{lemma:ortho}, we upper bound the mean squared error that arises when approximating a target observable \( O \) with another observable \( \Tilde{O} \).
\begin{lemma}[Mean-squared-error for observables]
\label{lemma:odg}
Let $H$ be an Hermitian operator and let $U$ a locally scrambling unitary.  
For all input state $\rho$ we have:
\begin{align}
   \bbE_U \left( \Tr[H U\rho U^\dag]^2\right)\leq \norm{H}^2_{\mathrm{Pauli},2}.
\end{align}
In particular, setting $H= O-\Tilde{O}$ we have
\begin{align}
   \bbE_U \left( \Tr[ (O-\Tilde{O}) U\rho U^\dag]^2\right)\leq \norm{ O-\Tilde{O}}^2_{\mathrm{Pauli},2}.
\end{align}
\end{lemma}
\begin{proof}
The lemma follows by applying Lemma\ \ref{lemma:ortho} and using the fact that $\Tr[s\rho]^2\leq 2^{-n}$ for all normalized Pauli $s\in\calP_n$.
    \begin{align}
        \bbE_U \left( \Tr[H U\rho U^\dag]^2\right])
        = &\sum_{s\in\calP_n}  \Tr[H s]^2 \bbE_U\Tr[sU\rho U^\dag]^2
        \\ \leq &\max_{\hat s \in \calP_n} \bbE_U\Tr[\hat sU\rho U^\dag]^2 \sum_{s\in\calP_n}  \Tr[Hs]^2 
       \\ \leq &\norm{H}^2_{\mathrm{Pauli},2}. 
    \end{align}
This completes the proof.
\end{proof}

\begin{comment}
We also assume that the distribution $\calD$ satisfies the Pauli-mixing property:
\begin{align}
    &&&\forall s\in\calP_n:\bbE_{U\sim\mathcal{D}}\left[U^{\dag\otimes 2}s^{\otimes 2} U^{\otimes 2}\right] = \frac{1}{3^{\abs{s}}}\sum_{\substack{t\in\calP_n:\\\mathrm{supp}(t)=\mathrm{supp}(s)}} \left[U^{\dag\otimes 2}t^{\otimes 2} U^{\otimes 2}\right]  &\text{(Pauli-mixing)} \label{eq:mixing}. 
\end{align}

\begin{fact}
Let $U \sim \calD$, where  $U = \left(\bigotimes_{i=1}^n U_i \right) V$, and where each $U_i$ is sampled independently from a single-qubit 2-design. Then $\calD$ satisfies both Pauli-invariance and Pauli-mixing
\end{fact}

\begin{fact}
Let $\calD$ a distribution over $\bbU(2)$ that satisfies both Pauli-mixing and orthogonality.% Moreover, assume that $\bbE_{U\sim\calD} U X U^\dag = \Tr[X] \frac{I}{2}$. 
Then $\calD$ is a 2-design.
\end{fact}
\begin{proof}
Let $P,Q \in \{I,X,Y,Z\}$. Orthogonality and Pauli-mixing imply:
\begin{align}
    \bbE_{U\sim\calD}U^{\otimes 2} (P\otimes Q) U^{\dag \otimes 2} = \begin{cases}
    0 & \text{if $P\neq Q$},\\
    I^{\otimes 2} & \text{if $P= Q = I$},\\
    \frac{1}{3}\left( X^{\otimes 2}+Y^{\otimes 2}+Z^{\otimes 2} \right) & \text{if $P= Q \neq I$}.
    \end{cases}
\end{align}
Tracing out the second system we obtain:
\begin{align}
    \bbE_{U\sim\calD}U P U^\dag = \bbE_{U\sim\calD} \Tr_2 [U^{\otimes 2} (P\otimes I) U^{\dag \otimes 2}] = \begin{cases}
    0 & \text{if $P\neq I$},\\
    I & \text{if $P = I$}.
    \end{cases}
\end{align}
\end{proof}
\end{comment}

We now discuss an important consequence of the Pauli-mixing property.
Informally, the Pauli-mixing property says that the second moment of a Pauli operator evolved under a locally scrambling unitary $U$ depends only on the support of the Pauli. As previously observed in Ref.\ \cite{huang2023learning}, this implies that the second moment of the inner product $\Tr[U^\dag P U \rho]$ decreases exponentially in the Pauli weight $\abs{P}$.
\begin{lemma}
\label{lemma:2/3}
Let $U$ a locally scrambling unitary and let $P\in\{I,X,Y,Z\}^{\otimes n}$.
We have
\begin{align}
    \bbE_{U} \Tr[U^\dag P U \rho]^2 
    \leq  \left(\frac{2}{3}\right)^{\abs{P}}.
\end{align}
\end{lemma}
\begin{proof}
By Pauli-mixing, we have
\begin{align}
    \bbE_{U} \Tr[U^\dag P U \rho]^2 = \bbE_{U} \frac{1}{3^{\abs{P}}}\sum_{\substack{Q:\\ \mathrm{supp}(Q)= \mathrm{supp}(P)}} \Tr[U^\dag QU\rho]^2 \label{eq:pml}%\leq \frac{1}{2^n}\left(\frac{2}{3}\right)^{\abs{s}} ,
\end{align}
We define the reduced state:
\begin{align}
    \tilde{\rho}_{{\mathrm{supp}(s)}}:= \Tr_{[n]\setminus {\mathrm{supp}(P)}}[U\rho U^\dag].
\end{align}
We can now relate the sum in Eq.\ \eqref{eq:pml} to the purity of $\tilde{\rho}_{{\mathrm{supp}(P)}}$.
\begin{align}
    \Tr[(\tilde{\rho}_{{\mathrm{supp}(P)}})^2]  
    =& 
    \frac{1}{2^{\abs{P}}} \sum_{\substack{Q \in \{I,X,Y,Z\}^{\otimes k}}} \Tr[Q \tilde{\rho}_{{\mathrm{supp}(P)}}]^2
    \\=&  \frac{1}{2^{\abs{P}}} \sum_{\substack{Q \{I,X,Y,Z\}^{\otimes n} :\\ \mathrm{supp}(Q) \subseteq \mathrm{supp}(P)}} \Tr[QU\rho U^\dag]^2
    \geq \frac{1}{2^{\abs{P}}} \sum_{\substack{Q \{I,X,Y,Z\}^{\otimes n} :\\ \mathrm{supp}(Q) =\mathrm{supp}(P)}} \Tr[QU\rho U^\dag]^2.
     %\geq \frac{2^n}{2^{\abs{s}}} \sum_{\substack{t \in \calP_n\\\mathrm{supp}(t) = \mathrm{supp}(s)}} \Tr[tU\rho U^\dag]^2.
\end{align}
Putting all together, we obtain
\begin{align}
    \bbE_{U} \frac{1}{3^{\abs{P}}}\sum_{\substack{Q:\\ \mathrm{supp}(Q)= \mathrm{supp}(P)}} \Tr[U^\dag QU\rho]^2 \leq \left( \frac{2}{3}\right)^{\abs{P}} \bbE_{U}  \Tr[(\tilde{\rho}_{{\mathrm{supp}(P)}})^2] \leq \left( \frac{2}{3}\right)^{\abs{P}}.
\end{align}
\end{proof}
The properties described above have been exploited by\ \citet{huang2023learning} for producing a low-weight approximation for an arbitrary observable, as we also show in the following lemma.

\begin{lemma}[Approximate inner product] 
\label{lem:apx-inner-product}
Let $U$ be a locally scrambling unitary, $O$ an observable and $\rho$ a quantum state. For all $k\geq 0$, denote the high-weight and low-weight components of $O$ as follows
\begin{align}
    O^{(\mathrm{low})}:= \sum_{\substack{s\in \calP_n\\\abs{s}\leq k}} \Tr[Os] s,\;\;\;\;
    O^{(\mathrm{high})}:= \sum_{\substack{s\in \calP_n\\\abs{s}> k}} \Tr[Os] s.
\end{align}
We have
\begin{align}
        \bbE_{U}   \Tr[U^\dag O^{(\mathrm{high})}U\rho]^2
      =&\bbE_{U}    \Tr[OU\rho U^\dag]^2 -  \bbE_{U} \Tr[O^{(\mathrm{low})}U\rho U^\dag]^2 \\ \leq &\left(\frac{2}{3}\right)^{k+1}  \left(\norm{O}_{\mathrm{Pauli},2}^2 - \norm{O^{(\mathrm{low})}}_{\mathrm{Pauli},2}^2 \right).
    \end{align}
\end{lemma}
\begin{proof}
The proof closely follows that of Corollary 13 in Ref.\ \cite{huang2023learning}.
  The first identity can be proven employing the orthogonality property. By means of Lemma\ \ref{lemma:ortho}, we obtain
  \begin{align}
      \bbE_{U}   \Tr[U^\dag O^{(\mathrm{high})}U\rho]^2
      = \bbE_{U} \sum_{\abs{s}>k} \Tr[ Os]^2\Tr[U^\dag sU\rho]^2 = \bbE_{U}    \Tr[OU\rho U^\dag]^2 -  \bbE_{U} \Tr[O^{(\mathrm{low})}U\rho U^\dag]^2  .
  \end{align}
As for the inequality, we employ again Lemma\ \ref{lemma:ortho} and  Lemma\ \ref{lemma:2/3},
\begin{align}
    &\bbE_{U}   \Tr[U^\dag O^{(\mathrm{high})}U\rho]^2 = \bbE_U \sum_{s\in\calP_n}\Tr[O^{(\mathrm{high})}s]^2\Tr[U^\dag sU\rho]^2 \\\leq &\left(\frac{2}{3}\right)^{k+1} \frac{1}{2^n}
    \sum_{s\in\calP_n}\Tr[O^{(\mathrm{high})}s]^2 = \left(\frac{2}{3}\right)^{k+1} \frac{1}{2^n}\left(\sum_{s\in\calP_n}\Tr[O s]^2 -\sum_{s\in\calP_n}\Tr[O^{(\mathrm{low})}s]^2  \right)
    \\ = & \left(\frac{2}{3}\right)^{k+1} \left(\norm{O}_{\mathrm{Pauli},2}^2 - \norm{O^{(\mathrm{low})}}_{\mathrm{Pauli},2}^2 \right).
\end{align}

\end{proof}
%We observe that $\norm{O^{(\mathrm{high})}}^2_{\mathrm{Pauli, 2}} \leq \norm{O}^2_{\mathrm{Pauli, 2}} \leq \norm{O}_\infty$.
An immediate consequence of this result is that, on average, we can approximate any observable $O$ with its low-weight component $O^{(\mathrm{low})}$ and incur in a small additive error on locally scrambled inputs.
Combining the above result with Jensen's inequality we obtain,
\begin{align}
    \bbE_U \left|\Tr[U^\dag \left(O - O^{(\mathrm{low})}\right)U\rho]\right| = \bbE_U \left|\Tr[U^\dag O^{(\mathrm{high})}U\rho]\right| \leq \left(\frac{2}{3}\right)^{(k+1)/2}\norm{O}_{\mathrm{Pauli, 2}}.
\end{align}
%where we also used the fact that $\norm{O}_{\mathrm{Pauli},2}^2 - \norm{O^{(\mathrm{low})}}_{\mathrm{Pauli},2}^2 \leq \norm{O}_{\mathrm{Pauli, 2}}^2 $.

In the following sections, we will demonstrate that this low-weight approximation can be generalized to the Pauli path integral.

\section{The Pauli path integral}
In this section, we introduce the Pauli path integral, which is the Feynman path integral written in the Pauli basis, and we discuss some of its properties.

Given an observable $O$, an initial quantum state $\rho$,  and  an $L$-layered quantum circuit $U = U_L U_{L-1} \dots U_1$, %of depth $L$.%, such that each layer is sampled independently from a distribution satisfying both Pauli-invariance and Pauli-mixing.
we want to compute the following inner product:
\begin{align}
    f_U(O):=\Tr[O U \rho U^\dag]. %= \Tr[ U_1^\dag \dots U_{L-1}^\dag U_L^\dag O U_L U_{L-1} \dots U_1\rho ]
\end{align}
To this end, we consider the Heisenberg evolution of the observable $O$, i.e. the evolution under the adjoint unitary channel
\begin{align}
     U^\dag(\cdot) U = U_1^\dag U_2^\dag\dots U_L^\dag(\cdot)U_L\dots U_2U_1.
\end{align}
Applying iteratively Eq.~\eqref{eq:innerproduct_app}, we obtain
\begin{align}
    U^\dag O U = &\sum_{s_0,\dots,s_L\in \calP_n} \Tr[Os_L]\Tr[U_{L}^\dag s_L U_L s_{L-1}]\Tr[U_{L-1}^\dag s_{L-1} U_{L-1} s_{L-2}]\dots\Tr[U_{1}^\dag s_{1} U_{1} s_{0}]s_0
      \nonumber \\  = &\sum_{\gamma \in \mathcal{P}_n^{L+1}} \Phi_\gamma(U) s_\gamma. 
\end{align}
Here we labeled each Pauli path by a string $\gamma = (s_0,\dots,s_L)$; we denoted the associated Fourier coefficient by $\Phi_\gamma(U)  := \mathrm{Tr}[Os_L]\prod_{j=1}^L \mathrm{Tr}[U_{j}^\dag s_j U_j s_{j-1}]$ and, in a slight abuse of notation, we defined $s_\gamma:=s_0$.
We remark that the definition of Fourier coefficient used in\ Ref.\ \cite{aharonov2023polynomial} slightly differs from ours, since they incorporate the product $\Tr[s_\gamma\rho]$ inside the coefficient. %However, defining the final product $d_\gamma$ separately has a central role in our analysis.
%Moreover, the expectation value ${f}_U(\rho, O)$ can be viewed as an inner product and hence 
% \begin{align}
%     \Tr[O U \rho_0 U^\dag] = \Tr[O U_L U_{L-1} \dots U_1\rho_0 U_1^\dag \dots U_{L-1}^\dag U_L^\dag]
%\end{align}
\begin{comment}
The Pauli path integral formalism is based around the observation that any inner product between two Hermitian operators $H$ and $H'$ can be expanded in the (normalized) Pauli basis $\calP_n = \left\{\frac{I}{\sqrt{2}},\frac{X}{\sqrt{2}},\frac{Y}{\sqrt{2}},\frac{Z}{\sqrt{2}}\right\}^{\otimes n}$ as
\begin{align}\label{eq:innerprod}
    \Tr[HH'] = \sum_{s\in\calP_n} \Tr[Hs]\Tr[sH'] \, .
\end{align}
The expectation value ${f}_U(\rho, O)$ can be viewed as an inner product and hence 
% \begin{align}
%     \Tr[O U \rho_0 U^\dag] = \Tr[O U_L U_{L-1} \dots U_1\rho_0 U_1^\dag \dots U_{L-1}^\dag U_L^\dag]
% \end{align}
% Using Eq.~\eqref{eq:innerprod}, 
we see that for each $j \in \{0,1,\dots,L\}$ we can write
\begin{align}
    \Tr[U^\dag O U\rho] = \sum_{s_{j}\in \calP_n} &\Tr[(U^\dag_{j+1} \dots U^\dag_{L} O U_{L}\dots U_{j+1})s_{j}] \nonumber \\ & \times \Tr[s_{j} U_{j} \dots U_1 \rho U_{1}^\dag \dots U_{j}^\dag].
    \label{eq:int-j}
\end{align}
\end{comment}
The Pauli path integral is the summation obtained by projecting the Heisenberg evolved observable onto the initial state:
\begin{align} % \label{eq:PathIntExp}
    &f_U(O)=\Tr[U^\dag O U\rho] = %\sum_{s_0,\dots,s_L\in \calP_n} \Tr[Os_L]\Tr[U_{L}^\dag s_L U_L s_{L-1}] \nonumber \\ &\times \Tr[U_{L-1}^\dag s_{L-1} U_{L-1} s_{L-2}]\dots \Tr[U_{1}^\dag s_{1} U_{1} s_{0}]\Tr[s_0\rho]\nonumber \\ & = 
    \sum_{\gamma \in \mathcal{P}_n^{L+1}} \Phi_\gamma(U) d_\gamma \, ,
\end{align}
where we defined
\begin{equation}\label{eq:dgamma}
    d_\gamma:=\Tr[s_\gamma \rho] \, .
\end{equation}
We now define a key property satisfied by a wide class of distribution over quantum circuits.
\begin{definition}[Orthogonality of Pauli paths]
Let $U$ be a random quantum circuit. We say that $U$ has orthogonal Pauli paths if for any $\gamma\neq\gamma'$ we have
\begin{align}
    \bbE_{U} \Phi_\gamma(U)  \Phi_{\gamma'}(U) = 0,
\end{align}
that is, the Fourier coefficients associated to two distinct Pauli paths are uncorrelated.
\end{definition}
Crucially, this property is satisfied by locally scrambling circuits, as we prove in the following lemma.
\begin{lemma}
Let $U$ be a random circuit sampled from a locally scrambling circuit distribution. Then $U$ has orthogonal Pauli paths, i.e. for any $\gamma\neq\gamma'$ we have
\begin{align}
    \bbE_{U} \Phi_\gamma(U)  \Phi_{\gamma'}(U) = 0.
\end{align}
\label{lem:ortho2}
\end{lemma}

\begin{proof}
We can express the product of the two coefficients $\Phi_\gamma(U)  \Phi_{\gamma'}(U)$ as follows:
\begin{align}\label{eq:orthogderiv}
    &\Phi_\gamma(U)  \Phi_{\gamma'}(U) \\=  &\Tr[Os_L]\Tr[Os_L']\Tr[U_{L}^\dag s_L U_L s_{L-1}]\Tr[U_{L}^\dag s_L' U_L s_{L-1}']\dots\\\dots&\Tr[U_{L-1}^\dag s_{L-1} U_{L-1} s_{L-2}]\Tr[U_{L-1}^\dag s_{L-1}' U_{L-1} s_{L-2}']\Tr[U_{1}^\dag s_{1} U_{1} s_{0}]\Tr[U_{1}^\dag s_{1}' U_{1} s_{0}'].
\end{align}
Given $\gamma \neq \gamma'$, let $j$ be lowest index such that $s_j\neq s_j'$.
%If $\gamma \neq \gamma'$, then there exists $j$ such that  
In order to prove the lemma, it suffices to show that 
\begin{align}
\bbE_{U} \Tr[U_{j}^\dag s_{j} U_{j} s_{j-1}]\Tr[U_{j}^\dag s_{j}' U_{j} s_{j-1}'] = 0
\end{align}
If $j=1$, then $U_j$ is a locally scrambling unitary over $\bbU(2^n)$. We have
\begin{align}
    &\bbE_{U} \Tr[U_{j}^\dag s_{j} U_{j} s_{j-1}]\Tr[U_{j}^\dag s_{j}' U_{j} s_{j-1}'] 
\\= &\bbE_{U} \Tr[U_{j}^{\dag\otimes 2} (s_{j} \otimes s_{j}')U_{j}^{\otimes 2} (s_{j-1}\otimes s_{j-1}')]
    = 0\label{eq:singleorthogterm},
\end{align}
where the final equality follows from $\bbE_{U_j}U_{j}^{\dag\otimes 2} (s_{j} \otimes s_{j}')U_{j}^{\otimes 2} =0$, as we showed in Lemma\ \ref{lemma:ls}. The desired result then follows from substituting Eq.~\eqref{eq:singleorthogterm} back into Eq.~\eqref{eq:orthogderiv}.

We now consider the case where $j>1$.
Let $s_j = \bigotimes_{i=1}^n s_j^{(i)}$ and $s_j' = \bigotimes_{i=1}^n s_j'^{(i)}$, where $s_j^{(i)}$ and $s_j'^{(i)}$ are single-qubit normalized Pauli acting on the $i$-th qubit.
There exists an $i\in[n]$ such that 
\begin{align}
    s_j^{(i)} \neq s_j'^{(i)}
\end{align}
Since we have that $s_{j-1}= s_{j-1}'$, then the unitary $U_j$ acts non-trivially on the $i$-th qubit.
Thus, we can write $U_j = I_{[n]\setminus A_j} \otimes V_{A_j}$ where $V_{A_j}$ is a locally scrambling unitary over $\bbU\left(2^{\abs{A_j}}\right)$ and $i\in A_j \subseteq [n]$. We obtain that

%In particular, this implies that $V_{A_j}$ is invariant under post-processing by a random single-qubit Clifford acting on the $i$-th qubit. Leveraging this property, we obtain that
\begin{align}
    &\bbE_{U} \Tr[U_{j}^\dag s_{j} U_{j} s_{j-1}]\Tr[U_{j}^\dag s_{j}' U_{j} s_{j-1}'] 
\\= &\bbE_{U} \Tr[V_{A_j}^{\dag\otimes 2} \left(\bigotimes_{i\in A_j}s_{j}^{(i)} \otimes \bigotimes_{i\in A_j}s_{j}'^{(i)}\right)V_{A_j}^{\otimes 2} \left(\bigotimes_{i\in A_j}s_{j-1}^{(i)} \otimes \bigotimes_{i\in A_j}s_{j-1}'^{(i)}\right)]
    = 0\label{eq:singleorthogterm2},
\end{align}
where the final equality follows again from Lemma\ \ref{lemma:ls}. The desired result then follows from substituting Eq.~\eqref{eq:singleorthogterm2} back into Eq.~\eqref{eq:orthogderiv}. 
This concludes the proof of the Lemma.
\end{proof}
We emphasize that there exist families of random circuits that exhibit orthogonal Pauli paths despite not having locally scrambling layers. The above proof leverages only the orthogonality property and the independence of the layers. Therefore, circuits with independent Pauli-invariant layers have orthogonal Pauli paths, as also observed in Ref.~\cite{aharonov2023polynomial}. Moreover, parameterised quantum circuits composed of Clifford gates and single-qubit rotations with uncorrelated angles also result in orthogonal Pauli paths, as noted in~\cite{shao2023simulating, fontana2023classical}.

\bigskip

Throughout this work, we aim at approximating the exact path integral $\sum_{\gamma \in \mathcal{P}_n^{L+1}} \Phi_\gamma(U) d_\gamma$ with an efficiently computable estimator that produces a small additive error on locally scrambling circuits. A natural strategy, also employed in previous works, consists in restricting the integral to a carefully chosen subset of paths, as we formalize in the following definition.

\begin{definition}[Truncated path integral]
For a subset of paths  $\calS\subseteq \calP_n^{L+1}$, we define the associated truncated observable $O_U^{(\calS)}$ and the truncated path integral $\tilde{f}^{(\calS)}_U(O)$ as
\begin{align}
    &O_U^{(\calS)}:=\sum_{\gamma \in \calS} \Phi_\gamma(U) s_\gamma,
    \\&\tilde{f}^{(\calS)}_U(O):=\sum_{\gamma \in \calS} \Phi_\gamma(U) d_\gamma.
\end{align}
We evaluate the performance of the estimator $\tilde{f}^{(\calS)}_U(O)$ in mean squared error, which is defined as follows:
\begin{align}
    \bbE_U \Delta_U^{(\calS)} := \bbE_U \left[\left({f}_U(O) - \tilde{f}^{(\calS)}_U(O)  \right)^2 \right].
\end{align}
\end{definition}

The mean squared error arises as a natural metric to leverage the properties of circuits with uncorrelated Fourier coefficients. In particular, we will make extensive of the following lemma.
\begin{lemma}[Mean squared error]
\label{lemma:mse}
Let $U$ be a random circuit, $O$ be an observable and $\rho$ be quantum state. Assume that the Fourier coefficients associated to different paths are uncorrelated, i.e. 
\begin{align}
    \gamma\neq\gamma'\implies\bbE_U \;\Phi_\gamma(U) \Phi_{\gamma'}(U)=0.
\end{align}
Then we have
\begin{align}
    \bbE_U \Delta f_U^{(\calS)} = \bbE_U\left[\tilde{f}^{(\overline\calS)}_U(O) ^2\right]  = \bbE_U\left[{f}_U(O)^2\right] - \bbE_U\left[\tilde{f}^{(\calS)}_U(O) ^2\right] ,
\end{align}
where $\overline{\calS} := \calP_n^{L+1} \setminus \calS$ be the complement of the set $\calS$.
\end{lemma}
\begin{proof}
We can rewrite the mean squared error as follows:
\begin{align}
    \bbE_U \left[\left({f}_U(O) - \tilde{f}^{(\calS)}_U(O)  \right)^2 \right] = &\bbE_U \left[\left(\sum_{\gamma\in\overline{\calS}} \Phi_\gamma(U) d_\gamma\right)^2 \right]
    \\= &\bbE_U \sum_{\gamma\in\overline{\calS}} \Phi_\gamma(U) ^2d_\gamma^2 + \bbE_U\sum_{\substack{\gamma,\gamma'\in\overline{\calS}\\\gamma'\neq\gamma}} \Phi_\gamma(U) \Phi_{\gamma'}(U)d_\gamma d_{\gamma'}
    \\ = &\bbE_U \sum_{\gamma\in\overline{\calS}} \Phi_\gamma(U) ^2d_\gamma^2 = \bbE_U\left[\tilde{f}^{(\overline\calS)}_U(O) ^2\right],
\end{align}
where in the third step we used the fact that Fourier coefficients are uncorrelated.
We also have
\begin{align}
    \bbE_U \sum_{\gamma\in\overline{\calS}} \Phi_\gamma(U) ^2d_\gamma^2 = &\bbE_U \sum_{\gamma\in\calP_n^{L+1}} \Phi_\gamma(U) ^2d_\gamma^2 - \bbE_U \sum_{\gamma\in{\calS}} \Phi_\gamma(U) ^2d_\gamma^2
    \\=&\bbE_U\left[{f}_U(O)^2\right] - \bbE_U\left[\tilde{f}^{(\calS)}_U(O) ^2\right], 
\end{align}
which proves the lemma.
\end{proof}
\begin{comment}
\begin{corollary}[Monotonicity of MSE]
Let $U$ be a random circuit, $O$ be an observable and $\rho$ be quantum state. Assume that the Fourier coefficients associated to different paths are uncorrelated.
Moreover, let $\calS \subseteq \calT \subseteq \calP_n^{L+1} $. We have
\begin{align}
    \bbE_U \Delta f_U^{(\calS)} = \bbE_U \Delta f_U^{(\calT)} + \sum_{\gamma \in \calT \setminus \calS} \Phi_\gamma(U) ^2 d_\gamma^2 ,
\end{align}
and therefore
\begin{align}
  \bbE_U \Delta f_U^{(\calT)} \leq   \bbE_U \Delta f_U^{(\calS)} .
\end{align}
\end{corollary}
\end{comment}

\section{Low-weight Pauli propagation}\label{app:WeightTrunc}
Given an integer $k\geq0$, the corresponding low-weight Pauli propagation estimator is identified by the following subset of paths:
    \begin{align}
        \calS_k := \{\gamma =(s_0,s_1,\dots,s_L) \;|\; \forall i\neq0 : \;\abs{s_i}\leq k  \}  \subseteq \calP_{n}^{L+1}.
    \end{align}
For simplicity, we replace $\calS_k$ with $k$ in the superscripts of the associated quantities, i.e. 
%we let $O^{(k)}_U := O^{(\calS_k)}_U$, $\tilde f^{(k)}_U(O) := \tilde f^{(\calS_k)}_U(O)$
%and $\bbE_U\Delta_U^{(k)} := \bbE_U\Delta_U^{(\calS_k)} $.
\begin{align}
    &O^{(k)}_U := O^{(\calS_k)}_U, %= \sum_{\substack{s_0\in\calP_n,\\\abs{s_1},\abs{s_2},\dots,\abs{s_L}\leq k}} f(O,U,\gamma) s_\gamma,
    \\&\tilde f^{(k)}_U(O) := \tilde f^{(\calS_k)}_U(O), %= \sum_{\substack{s_0\in\calP_n,\\\abs{s_1},\abs{s_2},\dots,\abs{s_L}\leq k}} f(O,U,\gamma) d_\gamma,
    \\&\bbE_U\Delta_U^{(k)} := \bbE_U\Delta_U^{(\calS_k)}. %= \bbE_U \left[\left(f_U(O) - \tilde f^{(k)}_U(O) \right)^2\right].
\end{align}
We also introduce some additional notation to state our technical proofs in a more compact way.
Let $U_{j+1} = I^{\otimes n}$ be the $n$-qubit identity matrix.
\begin{align}
    O_{j} &:= \begin{cases}
    O &\text{ for $j=L+1$}\\
        \sum_{\abs{s}\leq k} \Tr[U^\dag_{j+1} O_{j+1} U_{j+1} s]s \equiv  (U_{j+1}^{\dag}O_{j+1}U_{j+1})^{(\mathrm{low})} &\text{ for $1\leq j \leq L$}       
    \end{cases} \label{eq:Oj-defin} \\
    \rho_{j} &:=
         U_j U_{j-1}\ldots U_1 \rho U_1^{\dag} \ldots U_{j-1}^{\dag} U_j^{\dag} \quad\quad\quad\quad\quad\quad\quad\quad\quad\quad\quad\quad\,\,\,\, \text{for $1\leq  j \leq L+1$}
\end{align}
A simple application of Lemma\ \ref{lem:apx-inner-product} yields the following corollary.
\begin{corollary}[Error for a single iteration]
For all $j = 1,\dots, L$, we have
\label{cor:iter}
    \begin{align}
        \bbE_{U} \Tr[U_{j+1}^{\dag}O_{j+1}U_{j+1}\rho_{j}]^2  - \Tr[O_{j}\rho_{j}]^2 \leq \left(\frac{2}{3}\right)^{k+1} \bbE_{U} \left(\norm{O_{j+1}}_{\mathrm{Pauli},2}^2 -  \norm{O_{j}}_{\mathrm{Pauli},2}^2 \right)
    \end{align}
\end{corollary}
\begin{proof}
As shown in Lemma\ \ref{lem:ls-subcircuit}, $V\coloneqq U_jU_{j-1}\dots U_1$ is a locally scrambling unitary over $\bbU(2^n)$. Moreover, it is independent of the observables $O_j$ and $U_{j+1}^{\dag}O_{j+1}U_{j+1}$. Hence Lemma\ \ref{lem:apx-inner-product} yields
\begin{align}
    &\bbE_{U_1, U_2,\dots, U_j}    \Tr[U_{j+1}^{\dag}O_{j+1}U_{j+1}\rho_j]^2 -  \bbE_{U_1, U_2,\dots, U_j}  \Tr[O_j \rho_j]^2
    \\\coloneqq &\bbE_{V}     \Tr[(U_{j+1}^{\dag}O_{j+1}U_{j+1})(V\rho V^\dag)]^2 -  
    \bbE_{V}  \Tr[O_j ( V\rho V^\dag )]^2
    \\ \leq &\left(\frac{2}{3}\right)^{k+1}  \left(\norm{U_{j+1}^{\dag}O_{j+1}U_{j+1}}_{\mathrm{Pauli},2}^2 - \norm{O_j}_{\mathrm{Pauli},2}^2 \right)
    \\= &\left(\frac{2}{3}\right)^{k+1}  \left(\norm{O_{j+1}}_{\mathrm{Pauli},2}^2 - \norm{O_j}_{\mathrm{Pauli},2}^2 \right),
\end{align}
where in the final step we used the unitarily invariance of the Pauli 2-norm.
\end{proof}

We upper bound the mean squared error of our estimator (thus proving Theorem\ \ref{thm:errorbound} in the main text) by applying iteratively Corollary\ \ref{cor:iter}.

%Our main result follows by iterating this approximation over all the layers.
\begin{theorem}[Approximate path integral]
\label{thm:main}
For $k\geq 0$, we have
    \begin{align}
        \bbE_U  \Delta f^{(k)}_U
        \leq \left(\frac{2}{3}\right)^{k+1} \left(\norm{O}_{\mathrm{Pauli},2}^2 - \bbE_U  \norm{O_U^{(k)}}_{\mathrm{Pauli},2}^2\right)
        \leq \left(\frac{2}{3}\right)^{k+1} \norm{O}^2.
    \end{align}
\end{theorem}
\begin{proof}
Since different Fourier coefficients are uncorrelated, by Lemma\ \ref{lemma:mse} we can express the mean squared error as
\begin{align}
    \bbE_U  \Delta f^{(k)}_U = \bbE_U \left[f_U(O)^2\right] - \bbE_U \left[\tilde f_U^{(k)}(O)^2\right] = \bbE_U \Tr[O U\rho U^\dag]^2 - \bbE_U \Tr[O_U^{(k)} \rho]^2.
\end{align}
We observe that
\begin{align}
    O_U^{(k)} =  \sum_{\substack{s_0\in\calP_n,\\\abs{s_1},\abs{s_2},\dots,\abs{s_L}\leq k}} f(O,U,\gamma) s_\gamma = 
     U^\dag_1 O_1 U_1.
\end{align}
Moreover, we have $O = O_{L+1}$, $\rho_{L+1} = U \rho U^\dag$ and $\rho_1= U_1\rho U_1^\dag$.
Then we obtain
\begin{align}
    \bbE_U  \Delta f^{(k)}_U=&\bbE_U \Tr[O_{L+1} \rho_{L+1} ]^2 - \bbE_U \Tr[O_1 \rho_1]^2  
    \\= &\bbE_U \sum_{j=1}^{L} \left(\Tr[O_{j+1}\rho_{j+1}]^2 -  \Tr[O_{j}\rho_{j}]^2\right)
    \\ = &\bbE_U \sum_{j=1}^{L} \left(\Tr[U_{j+1}^\dag O_{j+1} U_{j+1}\rho_{j}]^2 -  \Tr[O_{j}\rho_{j}]^2\right)
    \\ \leq &\left(\frac{2}{3}\right)^{k+1}\bbE_U \sum_{j=1}^{L} \left(\norm{O_{j+1}}_{\mathrm{Pauli},2}^2 -  \norm{O_{j}}_{\mathrm{Pauli},2}^2 \right)
    \\=  &\left(\frac{2}{3}\right)^{k+1} \bbE_U \left(\norm{O_{L+1}}_{\mathrm{Pauli},2}^2 -  \norm{O_1}_{\mathrm{Pauli},2}^2\right)
     \\=  &\left(\frac{2}{3}\right)^{k+1} \left(\norm{O}_{\mathrm{Pauli},2}^2 - \bbE_U  \norm{O_U^{(k)}}_{\mathrm{Pauli},2}^2\right),
\end{align}
where we wrote the mean squared error as a telescoping sum and we upper bounded each term via Corollary\ \ref{cor:iter}.
\end{proof}
We emphasize that the unitarity of the layers is inessential to the proof of Theorem\ \ref{thm:main}. Specifically, the layers $U_j(\cdot)U_j^\dag$ could be replaced by channels $\calC_j$ satisfying an analogous locally scrambling property, provided that those channels do not increase the 2-norm of operators (on average) under Heisenberg evolution. 
In particular, our simulability argument could be easily extended to circuits containing \emph{dynamic operations}, e.g. feedforwards or probabilistic resets, such as those considered in Ref.\ \cite{deshpande2024dynamic}. %probabilistic resets, that is single-qubit channels corresponding to a reset to the $\ketbra{0}{0}$ state with probability $p$ and applying the identity operation with probability $1-p$.

%Similarly, our simulation argument could also be extended to encompass randomly initialized unitary Quantum Convolutional Neural Networks (QCNNs), such as those considered in our companion paper Ref.\ \cite{bermejo2024quantum}.  

\medskip

We can further ensure that the low-weight Pauli propagation algorithm produces a low additive error with high probability. 
Specifically, combining Theorem\ \ref{thm:main} with Markov's inequality, we obtain the following Corollary.
\begin{corollary}
\label{cor:markov}
Let $U$ be a locally scrambling circuit and let 
\begin{align}
    k = \bigg\lceil\frac{ \log(2/(3\epsilon^2 \delta)) }{\log(3/2)} \bigg\rceil \in \calO\left(\log\left(\frac{1}{\epsilon \delta}\right)\right).
\end{align}
Then for all observable $O$ and state $\rho$, we have
\begin{align}
    \abs{f_U(O) - \tilde f_U^{(k)}(O)} \leq \epsilon \norm{O}_{\mathrm{Pauli},2} \leq \epsilon \norm{O},
\end{align}
with probability at least $1-\delta$ (over the randomness of $U$).
\end{corollary}
\begin{proof}
For a real random variable $X$ and $a>0$, Markov's inequality implies that
\begin{align}
    \Pr[\abs{X} \geq a] = \Pr[X^2 \geq a^2] \leq \frac{\bbE [X^2]}{a^2},
\end{align}
We replace $X$ with $f_U(O) - \tilde f_U^{(k)}(O)$ and $a$ with $\epsilon \norm{O}_{\mathrm{Pauli},2}$:
\begin{align}
    \Pr_U \left[\abs{f_U(O) - \tilde f_U^{(k)}(O)} \geq \epsilon \norm{O}_{\mathrm{Pauli},2}\right] \leq \frac{1}{\epsilon^2} \Delta f^{(k)}_U \leq  \frac{1}{\epsilon^2}\left(\frac{2}{3}\right)^{k+1}. \label{eq:markov}
\end{align}
Finally, we note that the RHS of Eq.\ \eqref{eq:markov} is upper bounded by $\delta$ if $k\geq { \log(2/(3\epsilon^2 \delta)) }/{\log(3/2)}$, which concludes our proof. 
\end{proof}

\section{Time complexity}\label{app:alg-time-complexity}
In this section, we  upper bound the time complexity of the low-weight Pauli propagation algorithm.

\begin{lemma}
Assume that for all $j \in [L]$, the observable $O_j$ is supported on $M\leq n$ qubits. Then the low-weight Pauli propagation algorithm runs in time
\begin{align}
  \calO(L) \cdot \min \{M^{2k}, M^{k}\cdot \mathrm{poly}(n) \} \label{eq:general-transitions}  
\end{align}
\end{lemma}
\begin{proof}
We start by upper bounding the number of Pauli operators supported on a subset of qubits of size $M$ and weight at most $k$, which we denote by $N_{M,k}$:
\begin{align}
    N_{M,k} =&\sum_{\ell=0}^k 3^\ell \binom{M}{\ell}
    \leq \sum_{\ell=0}^k 3^\ell \frac{M^\ell}{\ell!}
    \\\leq  &\sum_{\ell=0}^k \left(\frac{3M}{k}\right)^\ell \frac{k^\ell}{\ell!} 
    \leq  \left(\frac{3M}{k}\right)^k  \sum_{\ell=0}^\infty \frac{k^\ell}{\ell!} 
    \\= &\left(\frac{3eM}{k}\right)^k,
    \label{eq:counting}
\end{align}
where in the last step we used the fact that $\sum_{\ell=0}^\infty {k^\ell}/{\ell!} = e^k$.
% see this proof: https://math.stackexchange.com/questions/3044044/how-to-prove-upper-bound-for-partial-sum-of-binomial-coefficients    

For each layer $j\in[L]$, we need to compute the observable 
\begin{align}
    O_{j} = \sum_{\abs{s_j}\leq k}\Tr[s_j U^\dag_{j+1} O_{j+1} U_{j+1}]s_j
\end{align}
To this end, for all $s_{j+1}$ in the Pauli expansion of $O_{j+1}$ and for all $s_{j}$  such that $\abs{s_{j}},\abs{s_{j+1}}\leq k$, we need to compute the associated transition amplitude, i.e. 
\begin{align}
    \Tr[U^\dag_{j+1} s_{j+1} U_{j+1} s_j].
\end{align}
By equation\ \eqref{eq:counting}, we know the the Pauli expansions of $O_j$ and $O_{j+1}$ contains at most $\left({(3eM)}/{k}\right)^k \in \calO(M^k)$ Pauli operators. 
This already implies that the number of transition amplitudes to be computed for a single iteration is at most $\calO(M^{2k})$. %and therefore the number of transition amplitude to be computed for $L$ layers is upper bounded by $\calO(M^{2k}L)$.
Moreover, we can quadratically tighten the dependence on $M^k$ by recalling that we assume %(\cref{sec:framework}) 
that each Heisenberg evolved Pauli operator $U^\dag_{j+1} s_{j+1} U_{j+1}$ contains at most $\mathrm{poly}(n)$ Pauli terms.
Thus, the number of transition amplitudes to be computed for a single iteration is at most $\calO(M^{k}\cdot \mathrm{poly}(n))$. Therefore the number of transition amplitudes to be computed for $L$ layers is upper bounded by
\begin{align}
    \calO(L) \cdot \min \{M^{2k}, M^{k}\cdot \mathrm{poly}(n) \} .
\end{align}
This completes the proof.
\end{proof}
 We will subsequently replace $M$ with an appropriate value depending on the circuit architecture.
In the following, we upper bound the total number of transition amplitudes to be computed. We provide four distinct bounds, covering the general case and relevant classes of more structured circuits.

\bigskip
\noindent\textbf{General case.}  
In the most general case, we can upper bound $M$ by $n$, obtaining a runtime of $\calO(n^{k +\calO(1)}L)$.
Invoking Corollary\ \ref{cor:markov}, we obtain that a runtime of $Ln^{\calO\left(\log(\epsilon^{-1}\delta^{-1})\right)}$ suffices to estimate the expectation value $\Tr[OU\rho U^\dag]$ with precision $\epsilon\norm{O}$ and success probability $1-\delta$. This proves the first part of Theorem\ \ref{thm:resources}.

\bigskip
\noindent\textbf{Circuits with $\calO(1)$-qubit gates.} 
Let us now assume that each $U_{j+1}$ consists of non-overlapping gates that act on at most  $\calO(1)$ qubits. 
In this case, we can quadratically tighten the previous upper bound.

We observe that each $\calO(1)$-qubit gate can map a (non-identity) Pauli operator to a constant number of Pauli operators. Moreover, assuming that $\abs{s_j}\leq k$, then at most $k$ non overlapping gates act non-trivially on $s_j$. Thus, the Heisenberg evolved Pauli $U^\dag_{j+1} s_{j+1} U_{j+1}$ contains at most $2^{\calO(k)}$ different Pauli operators.
Therefore, the total number of transition amplitudes to be computed during a single iteration of the algorithm scales as
\begin{align}
    2^{\calO(k)}  \left(\frac{3en}{k}\right)^k = \calO(n^k),
\end{align}
provided that $k$ is a sufficiently large constant.
Hence, for a circuit of depth $L$, we need to compute at most
$\calO(n^{k} L)$
transition amplitudes. 

\bigskip
\noindent\textbf{Circuits with constant geometric locality.}
The above bounds can be considerably tightened whenever the Heisenberg evolved observable is not supported on the entire set of qubits, but rather on a small subset that we can upper bound by light-cone argument.
To this end, we provide the following notion of geometric locality, which is analogous to the definitions given in in Refs.\ \cite{yu2023learning, huang2023learning}, and it is implies by more rigorous definitions such as that in Ref.\ \cite{harrow2023approximate}.
\begin{definition}[Geometric dimension of a graph]
Given a  graph $G = (V,E)$, we denote by $\gamma_G(L)$ the largest cardinality of a set of vertices obtained from a single vertex set $S_0=\{x_0\}$ of $G$ in $L$ steps, where at each step $j \leq L$, we could get an $S_i$  by adding at most one neighbor vertex, if it is not in $S_{i-1}$, for each $p \in S_{i-1}$.
We say that a graph $G$ has geometric dimension $D$ if $\gamma_G(L) = \calO(L^D)$.
\end{definition}

\begin{definition}[Circuit geometry]
 A geometry over $n$ qubits is defined by a graph $G = (V,E)$ with $n=\abs{V}$ vertices. A geometrically-local
two-qubit gate can only act on an edge of $G$. A depth-$L$ quantum circuit embedded in $G$ has
$L$ layers, where each layer consists of non-overlapping geometrically-local two-qubit gates.
Moreover, the geometric locality of a circuit embedded in $G$ is the geometric dimension of $G$.

\end{definition}
As a consequence, if a circuit $U = U_L U_{L-1}\dots U_1$ has geometric locality $D>0$ then, for all observable $O$ supported on $k$ qubits and for all $j=0,1,\dots,L-1$, the Heisenberg evolved observable $U_{j+1}^\dag U_{j+2}^\dag\dots U^\dag_L O U_L \dots U_{j+2}U_{j+1}$ is supported on at most $\calO(k L^D)$ qubits.

Proceeding as in the general case, we can show that the total number of transition amplitudes to be computed during a single iteration of the algorithm scales as
\begin{align}
    \left(\frac{3eM}{k}\right)^{2k} = 2^{\calO(k)} L^{2Dk}.
\end{align}
If we further assume that each layer consists in non-overlapping $\calO(1)$-qubit gates, the upper bound can be tightened to
\begin{align}
    \left(\frac{3eM}{k}\right)^{2k} = 2^{\calO(k)} L^{Dk}.
\end{align}
Invoking Corollary\ \ref{cor:markov}, we obtain that a runtime of $L^{\calO\left(D\log(\epsilon^{-1}\delta^{-1})\right)}$ suffices to estimate the expectation value $\Tr[PU\rho U^\dag]$ with precision $\epsilon$ and success probability $1-\delta$, for any Pauli operator $P \in \{I,X,Y,Z\}^{\otimes n}$. This proves the second part of Theorem\ \ref{thm:resources}.

\bigskip
\noindent\textbf{Circuits with all-to-all connectivity.}
We also derive a bound for shallow circuits with long-range interactions. Crucially, we also assume that each layer $U_j$ consists in non-overlapping $\calO(1)$-qubit gates.
Then we have
\begin{align}
    \abs{U^\dag_j O U_j} \leq \calO(1) \abs{O}.
\end{align}
Then, assuming again that $O$ is a Pauli operator of weight at most $k$, we have that, for all $j=0,1,\dots,L-1$, the Heisenberg evolved observable $U_{j+1}^\dag U_{j+2}^\dag\dots U^\dag_L O U_L \dots U_{j+2}U_{j+1}$ is supported on $2^{\calO(L)} k$ qubits.
Therefore, the total number of transition amplitudes to be computed is at most
\begin{align}
    L\cdot (2^{\calO(L)})^k = 2^{\calO(kL)}
\end{align}
In particular, for $k \in \calO(1)$,  the total number of transition amplitudes to be computed is at most $2^{\calO(L)}$.

\section{Sample complexity of quantum-enhanced classical simulation }
\label{app:alg-sample-complexity}
In this section, we demonstrate that the low-weight Pauli propagation algorithm can be used for the task of $\mathsf{CSIM_{QE}}$ (Classical Simulation enhanced with Quantum Experiments)\ \cite{cerezo2023does}.
In this setting, one is allowed to use a quantum computer for an initial data acquisition phase.
In particular, when the observable $O$ or the state $\rho$ are unknown (or not classically simulable), we will demonstrate how to compute $\tilde{f}_U(O)$ combining our simulation algorithm with randomized measurements, whose number scales logarithmically with the system size.
To this end, we exploit the randomized measurement toolbox developed in previous works\ \cite{huang2020predicting, elben2022randomized, huang2022learning}.

\subsection{Unknown input state}
When the initial state is non-classically simulable, we can estimate an approximate state $\Tilde\rho$ by means of the ``classical shadows'' protocol\ \cite{huang2020predicting}. As we prove in the following lemma, we can obtain a small mean squared error with randomized Pauli measurements on copies of $\rho$.
\begin{lemma}[Shadow state]
\label{lem:shadow-state}
Let $\rho$ be an unknown $n$-qubit state, let $O = \sum_P a_P P$ be an observable and let $U$ be a locally scrambling unitary consisting in non-overlapping $\calO(1)$-qubit gates.
Given $k>0$, denote by $O^{(\mathrm{low})} = \sum_{P:\abs{P}\leq k} a_P P$ the low-degree approximation of $O$.
Using  $N =  \exp\left( \calO(k)\right) { \epsilon^{-1} \log(n/\delta)}$ random Pauli measurements on copies of $\rho$, we can output an operator $\tilde\rho$ such that,
\begin{align}
    \bbE_{U} \left|\Tr[O^{(\mathrm{low})} U(\rho - \Tilde{\rho})U^\dag]\right|^2 \leq  \epsilon \;\norm{ O^{(\mathrm{low})}}^2_{\mathrm{Pauli},2},
\end{align}
with probability at least $1-{\delta}$.
\end{lemma}

\begin{proof}
Let $c\in\calO(1)$ be a positive integer, and say that $U$ consists in non overlapping $c$-qubit gates. 
We apply the ``classical shadows'' protocol with randomized Pauli measurements on the state $\rho$ (\cite{huang2020predicting}, Theorem 1). 
Since the shadow norm of a Pauli $Q$ equals $3^{\abs{Q}/2}$ (\cite{huang2020predicting}, Lemma 3),
we can estimate some values $\hat{o}_P$ satisfying $\abs{\hat{o}_P -\Tr[P\rho]} \leq \sqrt{\epsilon}$ for all Pauli operators $P: \abs{P}\leq c\cdot k$, with probability at least $1-\delta$,  using  $N = \exp\left( \calO(k)\right) { \epsilon^{-1} \log(N_{n,ck}/\delta)}$ random Pauli measurements on copies of $\rho$ where we recall that $N_{n,ck} \in \mathcal{O}(n^k)$ from Eq.~\eqref{eq:counting}.% for the full light-cone of width $M = n$.
Then, we define the operator:
\begin{align}
    \Tilde\rho := \frac{1}{2^n}\left(I + \sum_{1\leq \abs{P}\leq c\cdot k} \hat{o}_P P \right),
\end{align}
satisfying $\Tr[P\Tilde\rho] = \hat{o}_P$ for all $P: \abs{P}\leq c\cdot k$.

Since the unitary $U$ consists in non-overlapping $c$-qubit gates and $O^{(\mathrm{low})} $ contains only Pauli operators with weight at most $k$, then $U^\dag O^{(\mathrm{low})} U$ contains only 
 Pauli operators with weight at most $c\cdot k$.

With probability at least $1-\delta$, we have
\begin{align}
    &\bbE_{U} \left(\Tr[ O^{(\mathrm{low})}U(\rho - \Tilde{\rho})U^\dag]^2 \right) 
     \\=&\bbE_U \left(\sum_{s\in\calP_n^{L+1}}\Tr[U^\dag  O^{(\mathrm{low})} U s]\Tr[s(\rho - \Tilde{\rho})]\right)^2
     \\= &\bbE_U \sum_{s: \abs{s}\leq c\cdot k}\Tr[U^\dag  O^{(\mathrm{low})} U s]^2\Tr[s(\rho - \Tilde{\rho})]^2
      \\=&\max_{P: \abs{P}\leq c \cdot k} \Tr[P(\rho - \Tilde{\rho})]^2 \bbE_U \norm{U^\dag  O^{(\mathrm{low})} U}^2_{\mathrm{Pauli},2}
     \\ =& \max_{P: 1\leq \abs{P}\leq c \cdot k} \left(\Tr[P\rho] - \hat{o}_P\right)^2 \norm{ O^{(\mathrm{low})}}^2_{\mathrm{Pauli},2} \leq \epsilon \norm{ O^{(\mathrm{low})}}^2_{\mathrm{Pauli},2},
\end{align}
where the second identity holds because $U$ is locally scrambling (Lemma\ \ref{lemma:ortho}).
\end{proof}

\subsection{Unknown observable}  
When the observable \( O \) being measured is not classically simulable, we estimate an approximate observable \( \tilde{O} \) by measuring \( O \) on tensor products of randomly chosen single-qubit stabilizer states. 

Our approach leverages a modified version of the algorithm originally proposed in Ref.\ \cite{huang2022learning}.

To establish the efficiency of our algorithm, we first introduce some preliminary lemmas. Our analysis leverages the well-known Medians-of-Means estimator.
\begin{lemma}[Median-of-Means, \cite{nemirovskij1983problem, jerrum1986random}]
\label{lemma:mom}
Let $X$ be a random variable with variance $\sigma^2$. Then, $K = 2 \log(2/\delta)$
independent sample means of size $N = 34 \sigma^2/\epsilon^2 $ suffice to construct a median of means estimator $\hat{\mu}(N,K)$ that
obeys 
\begin{align}
    \Pr[\abs{\hat{\mu}(N,K) - \bbE[X]}\geq \epsilon ] \leq \delta,
\end{align}
for all $\epsilon, \delta > 0$.
\end{lemma} 
We will further make use of some technical tools developed in\ Ref.\ \cite{huang2022learning}.

\begin{lemma}[Adapted from Lemma\ 16 in Ref.\ \cite{huang2022learning}]\label{lem:extract-Pauli}
    Let $O = \sum_{P \in \{I, X, Y, Z\}^{\otimes n}} a_P P$ be an observable and $\calD$ be the uniform distribution over tensor products of single-qubit stabilizer states.
    For any Pauli observable $P \in \{I, X, Y, Z\}^{\otimes n}$,
    we have
    \begin{equation}
    \bbE_{\rho \sim \mathcal{D}} \Tr[O \rho] \Tr[P \rho] = \left(\frac{1}{3}\right)^{|P|} a_P.
    \end{equation}
\end{lemma}

Combining Lemma\ \ref{lem:extract-Pauli} with the Median-of-Means estimator, we obtain an algorithm for estimating a $k$-local  Pauli component of an observable. 
\begin{lemma}[Learning Pauli coefficients]
\label{lemma:coeffs}  
Let $O =\sum_{P\in\{I,X,Y,Z\}^{\otimes n}} a_P P$ be an observable and let $Q$ be a Pauli operator with weight $\abs{Q}=k$. Then, using $68 \cdot 9^k\epsilon^{-2} \log(2/\delta)$ measurements of $O$ on tensor products of random single-qubit stabilizer states, it is possible to estimate a value $\alpha$ satisfying
\begin{align}
    \abs{a_Q - \alpha} \leq \epsilon \, \norm{O}_{\mathrm{Pauli},2} ,
\end{align}
with probability at least $1-\delta$.
\end{lemma}

\begin{proof}
Let $K = 2\log(2/\delta)$ and $N = 34 \cdot 9^k\epsilon^{-2}$.
Sample $N\cdot K$ i.i.d.  tensor products of random single-qubit stabilizer states $\rho_{1},\rho_{2},\dots, \rho_{NK}$. Let $x_{i}$ be the outcome obtained by measuring $O$ on $\rho_{i}$, and consider the rescaled random variable $X_{i}$ defined as
\begin{align}
    X_{i} = x_{i} \Tr[Q \rho_{i}] 3^k
\end{align}
We now consider the first and second moments of $X_{i}$, with respect both the randomness of the measurement and that of the initial state.
The first moment of $X_{i}$ is
\begin{align}
    \bbE X_{i} = \bbE_{\rho\sim\calD} {\Tr[O \rho] \Tr[Q \rho]} 3^k = a_Q,
\end{align}
via Lemma~\ref{lem:extract-Pauli}.
In order to compute the second moment of $x_{i}$, we write the observable as a weighted sum of projectors $O = \sum_v \lambda_v \ketbra{v}$, such that
\begin{align}
    \Pr[\text{$\lambda_v$ is measured on input $\rho$}] = \Tr[\ketbra{v}\rho].
\end{align}
Then we have
\begin{align}
    \bbE (x_{i}^2) = &\bbE_{\rho \sim \calD} \sum_{v}  \Tr[\ketbra{v}\rho] \lambda_v^2 
    \\=  &\sum_v  \lambda_v^2   \Tr[\ketbra{v}\bbE_{\rho \sim \calD}[\rho]] \\= 
    & \sum_v  \lambda_v^2   \Tr[\ketbra{v}\frac{I}{2^n}]
    = \norm{O}^2_{\mathrm{Pauli},2},
\end{align}
where we used the fact that random stabilizer states form a $1$-design, and that the squared Hilbert-Schmidt norm of an Hermitian operator is the sum of its squared eigenvalues.

As for the rescaled variable $X_{i}$, we have
\begin{align}
    \bbE \left[X_{i}^2\right] = &9^k\bbE \left[ \left(x_{i,j} \Tr[Q \rho]\right)^2 \right]
    \\\leq  &9^k\bbE\left[x_{i}^2\right] 
     \leq  9^k\norm{O}^2_{\mathrm{Pauli},2} ,
\end{align}
where we used the fact that $\abs{\Tr[Q \rho]} \leq 1$ in the second step.
By Lemma\ \ref{lemma:mom}, the Median-of-Means estimator $\hat{\mu}(N,K)$ satisfies
\begin{align}
    \abs{\hat{\mu}(N,K) - a_Q}\leq \epsilon \,  \norm{O}_{\mathrm{Pauli},2},
\end{align}
with probability at least $1-\delta$.
\end{proof}
We now provide two distinct algorithms for estimating the $k$-local components of an observable $O$. We start with an algorithm that achieves arbitrarily small constant error with a polynomial number of randomized  measurements.

We also give a refined algorithm specialized on $\calO(1)$-local observables, which achieve arbitrarily small constant error with logarithmically many measurements.

\begin{lemma}[Shadow observable]
\label{lem:shadow-obs-2}
Let $O= \sum_P a_P P$ be an unknown observable. % and let $U$ be a locally scrambling unitary. 
Given $k>0$, denote by $O^{(\mathrm{low})} = \sum_{\abs{P}\leq k} a_P P$  the low-weight component of $O$.
Then using $N$ measurements of $O$  on tensor products of random single-qubit stabilizer states, where
    \begin{equation}
        N \in \calO(n^k \epsilon^{-1} \log(n/\delta)) ,
    \end{equation}
it is possible to learn 
an observable $\Tilde{O}$ that satisfies
\begin{align}
    \norm{\Tilde{O} - O^{(\mathrm{low})}}^2_{\mathrm{Pauli},2} \leq \epsilon \,\norm{O}_{\mathrm{Pauli},2}^{2},
\end{align}
with probability $1-\delta$.
\end{lemma}

\begin{proof}
Let $N_{n,k} = \sum_{j=0}^k\binom{n}{j}3^j \leq \left(3en/k \right)^k$ be the number of Paulis with weight at most $k$.
By means of Lemma\ \ref{lemma:coeffs}, it is possible to estimate $a_P$ up to additive error $\sqrt{\epsilon/N_{n,k}}\norm{O}_{\mathrm{Pauli},2} $ with probability $1 -\delta/N_{n,k}$ using $N$ randomized measurements, where
\begin{align}
    N= 2^{\calO(k)} \, \epsilon^{-1} {N_{n,k}} \log(M^k/\delta) \in  \calO(n^k \epsilon^{-1} \log(n/\delta)).
\end{align}
By union bound, $N$ measurements suffice to output some coefficients $x_P$'s for all Paulis of weight at most $k$ satisfying
\begin{align}
   \max_{P:\abs{P}\leq k} \abs{x_P - a_P} \leq \sqrt{\epsilon/N_{n,k}}\, \norm{O}_{\mathrm{Pauli},2}
\end{align} 
with probability $1 - \delta$.
We condition on this event happening.
Summing over all the low-weight Pauli operators yields the desired result:
\begin{align}
    \norm{\Tilde{O} - O^{(\mathrm{low})}}^2_{\mathrm{Pauli},2} =  &\sum_{P:\abs{P}\leq k} \abs{x_P - a_P}^2 \\\leq &N_{n,k} \cdot \frac{\epsilon \norm{O}_{\mathrm{Pauli},2}^2}{N_{n,k}} \leq \epsilon \norm{O}_{\mathrm{Pauli},2}^2.
\end{align}
This concludes the proof of the lemma.
\end{proof}

\subsubsection{\texorpdfstring{Improved algorithm for $\calO(1)$-local observables}{Improved algorithm for O(1)-local observables}}

Here, we describe an alternative algorithm tailored on $\calO(1)$-local observables, that is observables that contain only Pauli terms of weight at most $\calO(1)$. This specialized algorithm allows to achieve arbitrarily small constant error with logarithmic (in system size) sample complexity.
Local observables play a prominent role in many-body physics, where they are used to represent physical quantities such as the average magnetization.

We start by restating the definition of Pauli-$p$ norm of an Hermitian operator $A = \sum_{P\in\{I,X,Y,Z\}^{\otimes n}} a_P P$
\begin{align}
    \norm{A}_{\mathrm{Pauli},p} := \left(\sum_{P\in\{I,X,Y,Z\}^{\otimes n}}\abs{a_P}^p \right)^{1/p}. 
\end{align}

We also recall the following norm inequality.

\begin{lemma}[Corollary 3 in Ref.\ \cite{huang2022learning}]
\label{lem:norm}
Given an $n$-qubit $k$-local Hermitian operator $H = \sum_{P:\abs{P}\leq k} a_P P$, we have 
\begin{align}
    \norm{H}_{\frac{2k}{k+1}} \leq B(k) \norm{H},
\end{align}
where $B(k) \in \exp\left(\calO(k \log k) \right)$.
\end{lemma}

The proposed algorithm implements a variant of the ``filtering lemma'' adopted in Ref.\ \cite{huang2022learning}.
\begin{lemma}[Filtered shadow observable]
\label{lem:shadow-obs}
Let $O= \sum_P a_P P$ be an unknown observable. % and let $U$ be a locally scrambling unitary. 
Given $k>0$, denote by $O^{(\mathrm{low})} = \sum_{\abs{P}\leq k} a_P P$  the low-weight component of $O$.
%Moreover, let $r = \frac{2k}{k+1} \in [1,2)$.
Then using $N$ measurements of $O$ on tensor products of random single-qubit stabilizer states, where
    \begin{equation}
        N \in \exp\left(\calO(k^2 \log k) \right) \epsilon^{-(k+1)} \log(n/\delta).
    \end{equation}
it is possible to learn an observable $\Tilde{O}$ that satisfies
\begin{align}
    \norm{\Tilde{O} - O^{(\mathrm{low})}}^2_{\mathrm{Pauli},2} \leq \epsilon \,\norm{O}^{\frac{2}{k+1}}_{\mathrm{Pauli},2} \,  \norm{O^{(\mathrm{low})}}^{\frac{2k}{k+1}},
\end{align}
with probability $1-\delta$. 

\smallskip
\noindent In particular, if $O$ is a $k$-local observable satisfying $\norm{O}\leq 1$, then the observable $\Tilde{O}$ satisfies
\begin{align}
    \norm{\Tilde{O} - O^{(\mathrm{low})}}^2_{\mathrm{Pauli},2} \leq \epsilon
\end{align}
with probability $1-\delta$. 
\end{lemma}

\begin{proof}
Let $N_{n,k} = \sum_{j=0}^k\binom{n}{j}3^j \leq \left({3en}/{k}\right)^k$ be the number of Pauli operators with weight at most $k$.
Given $P:\abs{P}\leq k$, by means of Lemma\ \ref{lemma:coeffs}, it is possible to estimate $a_P$ up to additive error ${\epsilon'}\norm{O}_{\mathrm{Pauli},2} $ with probability $1 -\delta/N_{n,k}$ using $N$ randomized measurements, where
\begin{align}
    N\in 2^{\calO(k)} \, \left(\epsilon'\right)^{-2} \log(N_{n,k}/\delta) \in \exp\left(\calO(k)\right) \left(\epsilon'\right)^{-2} \log(n/\delta) \label{eq:obs-sample}.
\end{align}
By union bound, $N$ measurements suffice to output some coefficients $x_P$'s for all Pauli operators of weight at most $k$ satisfying
\begin{align}
    \abs{x_P - a_P} \leq \epsilon' \, \norm{O}_{\mathrm{Pauli},2} := \hat{\epsilon} \label{eq:delta-bound},
\end{align}
 for all $P$ such that ${P:\abs{P}\leq k}$ with probability at least $1-\delta$, where we set $\hat{\epsilon}:=\epsilon'\norm{O}_{\mathrm{Pauli},2}$ to ease the notation.
In the remaining part of this proof, we condition on this event happening.

We define the observable $\Tilde{O}$ as follows:
\begin{align}
    \Tilde{O} =  {\sum_{ \substack{P:\abs{P}\leq k, \\ \abs{x_P} \geq 2 \hat{\epsilon}}}x_P P}.
\end{align}
The observable $ \Tilde{O}$ is obtained by filtering out the Pauli operators if their associated coefficient  $x_P$ is below an appropriate threshold.
When upper bounding the Pauli-2 distance, we will deal separately with the contributions of the filtered and unfiltered coefficients.
\begin{align}
     \norm{\Tilde{O} - O^{(\mathrm{low})}}^2_{\mathrm{Pauli},2} = \underbrace{\sum_{\substack{P:\abs{P}\leq k,\\\abs{x_P} \geq 2 \hat{\epsilon} }} \abs{x_P-a_P}^2}_{\text{unfiltered}} +
    \underbrace{\sum_{\substack{P:\abs{P}\leq k,\\\abs{x_P} < 2 \hat{\epsilon} }} \abs{a_P}^2}_{\text{filtered}}      
\end{align}
We make some preliminary observations.
When $\abs{x_P} < 2 \hat{\epsilon}$, the triangle inequality yields
\begin{align}
    \abs{a_P} = \abs{a_P - x_P + x_P} \leq \abs{a_P - x_P} + \abs{x_P} \leq \hat\epsilon + 2 \hat\epsilon = 3 \hat\epsilon.
\end{align}
This allows us to upper bound the contribution of the coefficients below the threshold:
\begin{align}
    \sum_{\substack{P:\abs{P}\leq k,\\\abs{x_P} < 2 \hat{\epsilon} }} \abs{a_P}^2   
    = \sum_{\substack{P:\abs{P}\leq k,\\\abs{x_P} < 2 \hat{\epsilon} }} \abs{a_P}^{2-r}  \abs{a_P}^{r}  \leq (3\hat\epsilon)^{2-r} 
    \sum_{\substack{P:\abs{P}\leq k,\\\abs{x_P} < 2 \hat{\epsilon} }} \abs{a_P}^{r} ,
\end{align}
where, so we can later apply Lemma~\ref{lem:norm}, we defined $r = \frac{2k}{k+1} \in [1,2)$.
On the other hand, when $\abs{x_P} \geq 2 \hat{\epsilon}$, the triangle inequality yields
\begin{align}
    2\hat{\epsilon}\leq \abs{x_P} = \abs{x_P - a_P + a_P} \leq \abs{x_P - a_P} + \abs{a_P} \leq &\hat{\epsilon} + \abs{a_P}
    \\ \implies  &\hat{\epsilon} \leq  \abs{a_P}.
\end{align}
Combining it with Eq.\ \eqref{eq:delta-bound}, we obtain
\begin{align}
    \abs{x_P - a_P} \leq\hat\epsilon\leq \abs{a_P},
\end{align}
for all $P$ satisfying $\abs{P}\leq k $ and $ \abs{x_P} \geq 2 \hat{\epsilon} $.
This allows us to upper bound the contribution of the coefficients above the threshold:
\begin{align}
    \sum_{\substack{P:\abs{P}\leq k,\\\abs{x_P} \geq 2 \hat{\epsilon} }} \abs{x_P-a_P}^2
    \leq  \sum_{\substack{P:\abs{P}\leq k,\\\abs{x_P} \geq 2 \hat{\epsilon} }} \hat{\epsilon}^2 
    =  \sum_{\substack{P:\abs{P}\leq k,\\\abs{x_P} \geq 2 \hat{\epsilon} }} \hat{\epsilon}^{2-r} \cdot \hat{\epsilon}^{r}
    \leq \hat\epsilon^{2-r}  \sum_{\substack{P:\abs{P}\leq k,\\\abs{x_P} \geq 2 \hat{\epsilon} }}\abs{a_P}^{r}.
\end{align}
Putting all together and applying Lemma\ \ref{lem:norm}, we obtain
\begin{align}
    \norm{\Tilde{O} - O^{(\mathrm{low})}}^2_{\mathrm{Pauli},2} &\leq (3\hat\epsilon)^{2-r} \sum_{\substack{P:\abs{P}\leq k}}\abs{a_P}^{r} \\&= (3\hat\epsilon)^{2-r}  \norm{O^{(\mathrm{low})}}_{\mathrm{Pauli},r}^r\\&\leq \hat\epsilon^{2-r} B(k)^r \norm{O^{(\mathrm{low})}}^r
    \\ &:= (\epsilon')^{2-r} \norm{O}^{2-r}_{\mathrm{Pauli},2} B(k)^r  \norm{O^{(\mathrm{low})}}^r.
\end{align}
Setting $\epsilon':= \epsilon^{{1}/({2-r})} B(k)^{-{r}/({2-r})} = \epsilon^{(k+1)/2} B(k)^{-1/k}$, we obtain the desired error:
\begin{align}
     \norm{\Tilde{O} - O^{(\mathrm{low})}}^2_{\mathrm{Pauli},2} \leq
      \epsilon \, \norm{O}^{\frac{2}{k+1}}_{\mathrm{Pauli},2} \,  \norm{O^{(\mathrm{low})}}^{\frac{2k}{k+1}} \label{obs-filt-prec}
\end{align}
Plugging the value of $\epsilon'$ inside Eq.~\eqref{eq:obs-sample}, we find that the required number of measurements is upper bounded as follows
\begin{align}
    N \in \exp\left(\calO(k^2 \log k) \right) \epsilon^{-(k+1)} \log(n/\delta).
\end{align}
It remains to prove the last part of the lemma. Assuming $O$ is a $k$-local observable satisfying $\norm{O}\leq 1$, we have
\begin{align}
    &\norm{O^{(\mathrm{low})}} = \norm{O} \leq 1,
     \\&\norm{O}_{\mathrm{Pauli},2} \leq \norm{O} \leq 1,
\end{align}
which together with Eq.\ \eqref{obs-filt-prec} imply that $ \norm{\Tilde{O} - O^{(\mathrm{low})}}^2_{\mathrm{Pauli},2} \leq \epsilon$.
\end{proof}

We emphasize that the first part of Lemma\ \ref{lem:shadow-obs} hold for generic -- possibly non-local -- observables, and thus Lemma\ \ref{lem:shadow-obs} might outperform 
Lemma\ \ref{lem:shadow-obs-2} whenever there exists a non-trivial bound for $\norm{O^{(\mathrm{low})}}$.
In the most general case, one can always upper bound $\norm{O^{(\mathrm{low})}}$ by Minkowski's inequality:
\begin{align}
    \norm{O^{(\mathrm{low})}} = &\left\|\sum_{P:\abs{P}\leq k}a_P P \right\| \leq \sum_{P:\abs{P}\leq k} \abs{a_P} \norm{P}
    \\=  &\norm{O^{(\mathrm{low})}}_{\mathrm{Pauli},1} \in \calO(n^k).
\end{align}
However, using this simple upper bound would result in a  sample complexity larger than that of Lemma\ \ref{lem:shadow-obs-2}.
\subsection{General case}
Combining the tomographic tools presented in this section with the accuracy guarantees of the low-weight Pauli propagation algorithm, we prove the following theorem.
\begin{theorem}[Quantum-enhanced classical simulation]
Let $U = U_L U_{L-1}\dots U_1$ be a circuit sampled from an $L$-layered locally scrambling circuit ensemble,  let $O$ be an unknown observable and $\rho$ be an unknown quantum state.
Moreover, assume that $U_1$ consists in non-overlapping $\calO(1)$-qubit gates.
%Given $\epsilon, \delta \in (0,1)$, set
%\begin{align}
%    \ell =  \lfloor \log(3\sqrt{3}\epsilon^{-2}\delta^{-1})/(\log(3/2)) \rfloor \in \calO\left(\log\left(\frac{1}{\epsilon \delta}\right)\right),
%\end{align}
%and let $O^{\mathrm{(low)}} = \sum_{|P| \leq \ell} \alpha_P P$ be the low-degree approximation of $O$, and $r = \tfrac{2\ell}{\ell + 1} \in [1, 2)$.

Given $\epsilon, \delta \in (0,1)$, after an initial data-collection phase consisting in $n^{\calO\left(\log({\epsilon^{-1}\delta^{-1}})\right)}$ randomized measurements,
%\begin{align}
    %&N_1 =  \log\left({n}\right) \,  2^{\mathcal{O}\left(\log^2({\epsilon^{-1}\delta^{-1}})\right)},
%    N = n^{\calO\left(\log({\epsilon^{-1}\delta^{-1}})\right)}
%\end{align}  
there is classical algorithm that runs in time $n^{\calO\left(\log({\epsilon^{-1}\delta^{-1}})\right)}$ and outputs a value $\alpha$ such that
\begin{align}
    \abs{\alpha - f_U(O)} \leq  \epsilon \norm{O}_{\mathrm{Pauli},2} \leq \epsilon \norm{O}
\end{align}
with probability at least $1-\delta$. 
The probability is both over the randomness of the circuit $U$ and that of the initial measurements.
\end{theorem}
\begin{proof}

 We first make some preliminary observations.
 Given some $k\geq 0$, recall that $O^{(k)}_U$ is the truncated observable obtained with the low-weight Pauli propagation algorithm. Moreover, $O^{(k)}_U$ can be expressed as $U_1^\dag O_1 U_1$, where $O_1$ contains only Pauli terms of weight at most $k$ and $U_1$ consists
    in non-overlapping $\calO(1)$-qubit gates. 
    
In order to approximate $f_U(O)$, we perform the following steps:
\begin{enumerate}
    \item    By means of Lemma\ \ref{lem:shadow-state}, we learn a ``shadow state'' $\tilde{\rho}$ satisfying 
    \begin{align}
       \bbE_{U} \left|\Tr[O^{(k)}_U (\rho - \Tilde{\rho})]\right|^2 
       \leq  (\epsilon')^2 \;\norm{ O^{(k)}_U}^2_{\mathrm{Pauli},2}
       \leq  (\epsilon')^2 \;\norm{ O }^2_{\mathrm{Pauli},2}
       \label{eq:acc1}
    \end{align}
    with probability at least $1-\delta'$.
    The number of required randomized measurements is upper bounded by
    \begin{align}
       \exp\left(\calO(k)\right) \log\left(\frac{n}{\delta'}\right)  (\epsilon')^{-2}  .
    \end{align}
    \item We learn a ``shadow observable'' $\Tilde{O}$ satisfying 
    \begin{align}
       \norm{\Tilde{O} - O^{(\mathrm{low})}}^2_{\mathrm{Pauli},2} \leq (\epsilon')^2 \,\norm{O}_{\mathrm{Pauli},2}^{2}. \label{eq:acc-shadow-obs}
    \end{align}
    with probability at least $1-\delta'$. Here, we denoted by $O^{(\mathrm{low})} = \sum_{s:\abs{s}\leq k} \Tr[Os]s$ the low-weight approximation of $O$.
    This can be done using either the procedure given in Lemma~\ref{lem:shadow-obs-2}. We select the procedure that achieve the lowest sample complexity for the given values of $\epsilon', \delta'$. Then the required number of measurements is upper bounded by
    \begin{align}
       \calO(n^k) \log\left(\frac{n}{\delta'}\right) (\epsilon')^{-2} 
    \end{align}
    Combining Lemma~\ref{lemma:odg} with Eq.~\ref{eq:acc-shadow-obs}, we have that 
    \begin{align}
        \bbE_{U} \left|\Tr[(O^{(\mathrm{low})} -\Tilde{O}) U\rho U^\dag]\right|^2 \leq 
         \norm{\Tilde{O} - O^{(\mathrm{low})}}^2_{\mathrm{Pauli},2} \leq (\epsilon')^2 \,\norm{O}_{\mathrm{Pauli},2}^{2} \label{eq:acc2}
    \end{align}
    with probability at least $1-\delta'$.
    \item Finally, we run the $k$-weight Pauli propagation algorithm on inputs $U, \Tilde{O}, \Tilde{\rho}$. Specifically, we compute the approximate Heisenberg-evolved observable $\Tilde{O}^{(k)}_U$ and project it onto the state $\Tilde{\rho}$, obtaining $\Tr[\Tilde{O}^{(k)}_U \Tilde{\rho}]$.

\end{enumerate}
In the following of this proof, we condition on Eqs.~\eqref{eq:acc1} and~\eqref{eq:acc2} being satisfied simultaneously. By union bound, this happens with probability at least $1-2\delta'$.
We can now upper bound the mean squared difference between $\Tr[\tilde{O}_U^{(k)} \tilde\rho]$ and $f_U(O) := \Tr[O U \rho U^\dag]$.

\begin{align}
    &\bbE_U \left(f_U(O) -  \Tr[\tilde{O}_U^{(k)} \tilde\rho]\right)^2
    \\= &\bbE_U \left(f_U(O) - \tilde{f}_U^{(k)}(O) + \tilde{f}_U^{(k)}(O) - \Tr[\tilde{O}_U^{(k)} \rho] + \Tr[\tilde{O}_U^{(k)} \rho]-  \Tr[\tilde{O}_U^{(k)} \tilde\rho]\right)^2
    \\ \leq &3\bbE_U\left\{\left(f_U(O) - \tilde{f}_U^{(k)}(O)\right)^2
    + \left(\tilde{f}_U^{(k)}(O) - \Tr[\tilde{O}_U^{(k)} \rho]\right)^2
    + \left( \Tr[\tilde{O}_U^{(k)} \rho]-  \Tr[\tilde{O}_U^{(k)} \tilde\rho]\right)^2\right\},
\end{align}
where in the last step we applied the inequality $(\sum_{i=1}^m a_i)^2 \leq m \sum_{i=1}^m a_i^2$, which is a special case of the Cauchy-Schwarz inequality.
By Theorem~\ref{thm:errorbound}, we have 
\begin{align}
    \bbE_U\left(f_U(O) - \tilde{f}_U^{(k)}(O)\right)^2 \leq \left(\frac{2}{3}\right)^{k+1} \norm{O}^2_{\mathrm{Pauli},2}.
\end{align}
Therefore, putting all together we obtain
\begin{align}
    &\bbE_U \left(f_U(O) -  \Tr[\tilde{O}_U^{(k)} \tilde\rho]\right)^2
    \leq 3\left(\left(\frac{2}{3}\right)^{k+1} + 2(\epsilon')^2\right) \norm{O}_{\mathrm{Pauli},2}^{2}
\end{align}
Markov's inequality yields
\begin{align}
    \Pr_U \left[ \left|f_U(O) -  \Tr[\tilde{O}_U^{(k)} \tilde\rho]\right| \geq \epsilon \norm{O}_{\mathrm{Pauli},2}^{2} \right] &\leq \frac{3}{\epsilon^2} \left(\left(\frac{2}{3}\right)^{k+1} + 2(\epsilon')^2\right)\\&\leq \frac{ (3\epsilon')^2}{\epsilon^2} = \frac{\delta}{3},
\end{align}
where we set $k = \lfloor 2\log(\epsilon'^{-1})/(\log(3/2)) \rfloor$ and $\epsilon'= \epsilon \sqrt{\delta}/(3\sqrt{3})$.
Moreover, we also choose $\delta' = \delta/3$.

Therefore, by union bound, the estimated expectation value $\Tr[\tilde{O}_U^{(k)} \tilde\rho]$ satisfies
\begin{align}
    \left|f_U(O) -  \Tr[\tilde{O}_U^{(k)} \tilde\rho]\right| \leq \epsilon \norm{O}_{\mathrm{Pauli},2}^{2},
\end{align}
with probability at least $1-\delta$.
\end{proof}

While the above theorem does not require any structural assumptions on the unknown observable $O$, we remark that the sample complexity can be considerably tightened if $O$ is either known or it is $\calO(1)$-local.

In order to achieve this exponential improvement, it is sufficient observe that (i) the sample complexity for learning the ``shadow state'' is logarithmic in $n$ for constant error (Lemma\ \ref{lem:shadow-state}), and (ii) the sample complexity for learning the  ``shadow observable'' with constant error is also logarithmic in $n$ if $O$ is $\calO(1)$-local (Lemma\ \ref{lem:shadow-obs}).

\section{Comparison with light-cone simulation}\label{app:lightcone}
In this section, we compare low-weight Pauli propagation to more conventional classical light cone simulation, which is based on the observation that the expectation value of a local observable only depends on the gates and qubits within its backward light cone. 
Assume for the sake of simplicity that $O$ is a local Pauli observable, i.e. a Pauli observable with weight 1.
For a depth-$L$ geometrically local circuit in $D$ dimension, the backward light cone of$O$ contains order $\calO(L^D)$ qubits. A brute-force statevector simulation can thus exactly compute the expectation value of $O$ with $2^{\calO(L^D)}$ memory and time. 
On the other hand, Pauli backpropagation of the observable equipped with weight truncation approximates the expectation value of $O$ up to an additive error $\epsilon$ with a success probability at least $1-\delta$ by only keeping track of Pauli operators within the light-cone with weight at most $k = \calO\left(\log(\epsilon^{-1}\delta^{-1})\right)$. 
The time complexity of weight truncation in this setting is upper bounded by $L^{\calO\left(D\log(\epsilon^{-1} \delta^{-1})\right)}$, cf. \Cref{thm:resources}. 

\medskip
\noindent The following list compares statevector simulation to Pauli propagation with weight truncation in various regimes. 

\begin{itemize}
    \item \underline{Constant (non-zero) error and constant failure probability}: For arbitrary depth $L$, the cost of light cone simulation scales as $2^{\calO(L^D)}$,  whereas the cost of weight truncation scales as $\mathrm{poly}(L^D)$. Thus if a constant $\epsilon > 0$ error and a constant $\delta > 0$ failure probability suffices, low-weight Pauli propagation asymptotically requires exponential in $L$ less computation time.
    \item \underline{Inverse-polynomially small error and failure probability}: Let $c\geq 0$ and $L =\log^c(n) = 2^{c \log(\log(n))}$. The cost of state vector light cone simulation scales as $2^{O(\log^{cD}(n))} = n^{\calO(\log^{cD-1}(n))}$, whereas the cost of weight truncation scales as ${2}^{c D \log(\log(n)) \log(\text{poly}(n))} = n^{\calO(cD\log\log(n))}$.
    In particular, for $cD \geq 2$, the scaling is $n^{\mathrm{poly} \log(n)}$ versus $n^{\calO( \log \log(n))}$.
Since $\log\log(n)< 4$ for $n < 5\times 10^{23}$, the above scaling could be significantly more feasible for practical purposes.

    \item \underline{For near-exact simulation}: Both statevector and Pauli propagation simulation have cost $2^{\calO(L^D)}$, however the statevector simulation caps out at $2^{L^D}$ memory and Pauli propagation at $4^{L^D}$ with the cross-over happening at $k=\frac{n}{2}$. Additionally, statevector simulation is remarkably fast on modern computers with specialized hardware components such as GPUs and TPUs. Thus statevector light cone simulation may be faster for near-exact simulation of general circuits. If the circuit contains some structure, however, it is conceivable that there are additional truncations that path-based simulation methods can leverage and statevector methods cannot -- in turn making the comparison more subtle.
\end{itemize}

\section{Comparison with the trivial estimator}\label{app:XQUATH}

In this section, we consider an extremely simple estimator for the expectation value $f_U(O) = \Tr[O U\rho U^\dag]$, which we refer as the \emph{trivial estimator}. 
\begin{definition}[Trivial estimator]
Let $U$ be a random circuit and $O$ be an observable. Then the corresponding \emph{trivial estimator} is an algorithm that always outputs the first moment of the expectation value $\mu:=\bbE_{U} [f_U(O)]$.
\end{definition}
Therefore the corresponding mean squared error is given by the variance of the expectation value $f_U(O)$:
\begin{align}
    &\bbE_{U} \left(f_U(O) - \mu\right)^2 = \bbE_{U} [f_U(O)^2] +\mu^2 - 2\mu \bbE_{U} [f_U(O)]
    \\= &\bbE_{U} [f_U(O)^2] - \bbE_{U} [f_U(O)]^2 = \mathrm{Var}_{U} f_U(O).
\end{align}
Under many circumstances, random quantum circuits exhibit highly concentrated expectation values\ \cite{mcclean2018barren, larocca2024review}. Consequently, the trivial estimator, despite its simplicity, provides a small mean squared error. In this work, we extensively exploit the properties of random circuits to investigate the performance of low-weight Pauli propagation. This naturally leads to a comparison between low-weight Pauli propagation and the trivial estimator.

In the following, we present two key insights: first, we prove that low-weight Pauli propagation consistently outperforms the trivial estimator in terms of mean squared error. Moreover, we precisely quantify this improvement for \( k = 1 \) in random brickwork circuits. Second, we demonstrate that for typical random circuits at high circuit depths, low-weight Pauli propagation  becomes indistinguishable from the trivial estimator. Based on this, we argue that low-weight Pauli propagation cannot be employed for refuting the XQUATH (Linear Cross-Entropy Quantum Threshold) conjecture on random circuits of depth $c\cdot n$, for a sufficiently large constant $c>0$. This also outlines the limitations of Pauli propagation methods for classically spoofing linear cross-entropy benchmarking (linear XEB) on circuits of linear depth.

\subsection{Improvement over the trivial estimator}
The following Proposition shows that low-weight Pauli propagation always outperforms the trivial estimator.
\begin{proposition}\label{prop:improvement}
Let $U$ be a circuit with independent locally scrambling layers, $O$ an observable and $\rho$ a quantum state.  We have
\begin{align}
   \bbE_U \Delta f_U^{(k)} =  \mathrm{Var}_{U} f_U(O)  - \mathrm{Var}_U\tilde f_U^{(k)}(O) 
\end{align}
In particular, since the $\mathrm{Var}_{U} f_U(O)$ is the mean squared error of the trivial estimator, then $\mathrm{Var}_{U} \Tilde{f}_U^{(k)}(O)$ quantify the improvement over the trivial estimator.
\end{proposition}
\begin{proof}
   For all locally scrambled circuit $U$ and observable $O$, the first moment of the expectation value is given by
\begin{align}
    \bbE_U f_U(O) = \frac{\Tr[O]}{2^n}.
\end{align}
Moreover, it is easy to see that, for all $k\geq 0$,
\begin{align}
    \bbE_U \tilde f_U^{(k)}(O) = \frac{\Tr[O]}{2^n}.
\end{align}
This observation allows us to simplify Lemma\ \ref{lemma:mse} as follows
\begin{align}
    \bbE_U \Delta_U^{(k)}(O) = \mathrm{Var}_{U} f_U(O) - \mathrm{Var}_{U} \Tilde{f}_U^{(k)}(O),
\end{align}
that is the mean squared error is given by the difference between the variance of the expectation value and that of the low-weight Pauli propagation estimator.
\end{proof}

\subsection{Improvement of weight-1 Pauli propagation over the trivial estimator}

In the following, we exactly compute $\mathrm{Var}_{U} \Tilde{f}_U^{(1)}(O)$ for random circuits with local 2-designs with arbitrary geometric dimension.
%We target random brickwork circuits for the sake of simplicity, but we anticipate that similar arguments can be applied to any random circuit made of one and two-qubit gates sampled from local 2-designs.

\begin{proposition}\label{prop:impr-k=1}
    Let $O$ be an observable and let $U= U_L U_{L-1}\dots U_1$ be a $L$-layered locally scrambling circuit with input $\ketbra{0^n}$.  Assume that each layer $U_j$ consists in non-overlapping single-qubit and 2-qubit gates sampled from local 2-designs
     and let $O^{(1)}:= \sum_{s\in \calP_n : \abs{s}=1} \Tr[O s]s$ be the weight-1 component of the observable $O$. We have
    \begin{align}
        \mathrm{Var}_U \tilde f_U^{(1)}(O) \geq \frac{1}{5}\left(\frac{2}{5}\right)^L \norm{O^{(1)}}^2_{\mathrm{Pauli},2}.
    \end{align}
\end{proposition}

\begin{proof}
By using the fact that the Fourier coefficients are uncorrelated, we have
\begin{align}
     \mathrm{Var}_U \tilde f_U^{(1)}(O)
     = \bbE_U\sum_{\substack{s_1,s_2,\dots, s_L \in \calP_n: \\ \abs{s_1},\abs{s_2},\dots,\abs{s_L}=1}} \Tr[O s_L]^2 \Tr[U^\dag_L s_L U_L s_{L-1}]^2 \Tr[U^\dag_{L-1} s_{L-1} U_{L-1} s_{L-2}]^2 
     \dots \Tr[U^\dag_{1} s_{1} U_{1}  \rho]^2 %
     \\= \sum_{\substack{s_{L} : \\\abs{s_L} =1 }}\left(\Tr[O s_L]^2 \sum_{\substack{s_{L-1} : \\\abs{s_L-1} =1 }} \left(\bbE_{U_L}\Tr[U^\dag_L s_L U_L s_{L-1}]^2 \dots \sum_{\substack{s_1 :\\ \abs{s_1} =1 }}\left( \bbE_{U_2}\Tr[U^\dag_2 s_2 U_2 s_1]^2\bbE_{U_1}\Tr[U^\dag_1 s_1 U_1 \ketbra{0^n}]^2 \right)  \right) \right)
\end{align}

Exploiting basic properties of local 2-designs, we can lower bound all the terms of the truncated path integral. 
In particular, let $s_j$ be a normalized Pauli string which is non-identity only 
on the $\ell$-th qubit. We distinguish between two cases:

%\begin{itemize}
\begin{enumerate}
    \item The layer $U_j$ contains a 2-qubit gate acting on the $\ell$-th qubit.
    \item The layer $U_j$ does not contain a 2-qubit gate acting on the $\ell$-th qubit.
\end{enumerate}

\bigskip
\noindent\underline{Case 1 :}  Let such 2-qubit gate act on the $\ell$-th and the $\ell'$-th qubits.
%and assume (without loss of generality) that the brickwork layer $U_j$ contains a gate acting on the $\ell$-th and $(\ell+1)$-th qubits.
We have
\begin{align}
   \bbE_{U_j} U^{\otimes 2\dag}_{j} s_{j}^{\otimes 2} U_{j}^{\otimes 2} = \frac{1}{15}\left( \sum_{\substack{s: \\\mathrm{supp}(s) = \{\ell, \ell'\}}} s^{\otimes 2} + 
   \sum_{\substack{s:\\\mathrm{supp}(s) = \{\ell\}}} s^{\otimes 2} + \sum_{\substack{s:\\\mathrm{supp}(s) = \{\ell'\}}} s^{\otimes 2} \right).
\end{align}
Therefore, we have
\begin{align}
    &\bbE_{U_{j}}\Tr[U^\dag_{j} s_{j} U_{j} \ketbra{0^n}]^2
    \\=  &\frac{1}{15}\left( \underbrace{\sum_{\substack{s\in\calP_n: \\\mathrm{supp}(s) = \{\ell, \ell'\}}} \Tr[s_{j-1}\ketbra{0^n}]}_{=1/2^n} + 
  \underbrace{\sum_{\substack{s\in\calP_n:\\\mathrm{supp}(s) = \{\ell\}}}\Tr[s_{j-1}\ketbra{0^n}]}_{=1/2^n} + \underbrace{\sum_{\substack{s\in\calP_n:\\\mathrm{supp}(s) = \{\ell'\}}} \Tr[s_{j-1}\ketbra{0^n}]}_{=1/2^n} \right)
   =  \frac{1}{5} \cdot \frac{1}{2^n},
\end{align}
where in the last step we observed that the Pauli expansion of $\ketbra{0^n}$ contains only three terms which contribute to the sums (i.e. $Z_\ell, Z_{\ell'}$ and $Z_\ell Z_{\ell'})$, and the additional factor $1/2^n$ comes from the normalization of the Pauli operators.
We also have
\begin{align}
    &\bbE_{U_{j}}\sum_{\substack{s_{j-1} \in \calP_n: \\ \abs{s_{j-1}}=1}} \Tr[U^\dag_{j} s_{j} U_{j} s_{j-1}]^2
    \\=  &\frac{1}{15}\left( \underbrace{\sum_{\substack{s_{j-1} \in \calP_n: \\ \abs{s_{j-1}}=1}}\sum_{\substack{s\in\calP_n: \\\mathrm{supp}(s) = \{\ell, \ell'\}}} \Tr[s_{j-1}s]}_{=0} + 
  \underbrace{\sum_{\substack{s_{j-1} \in \calP_n: \\ \abs{s_{j-1}}=1}} \sum_{\substack{s\in\calP_n:\\\mathrm{supp}(s) = \{\ell\}}}\Tr[s_{j-1}s]}_{=3} + \underbrace{\sum_{\substack{s_{j-1} \in \calP_n: \\ \abs{s_{j-1}}=1}}\sum_{\substack{s\in\calP_n:\\\mathrm{supp}(s) = \{\ell'\}}} \Tr[s_{j-1}s]}_{=3} \right)
   =  \frac{2}{5},
\end{align}
where we noticed that there are 3 Pauli operators supported on $\{\ell\}$ and 3 Pauli operators supported on $\{\ell'\}$.
%\end{itemize}

\bigskip
\noindent\underline{Case 2:} By the locally scrambling assumption, there is at least a single-qubit gate acting on $s_j$
We have
\begin{align}
   \bbE_{U_j} U^{\otimes 2\dag}_{j} s_{j}^{\otimes 2} U_{j}^{\otimes 2} = \frac{1}{3}
   \sum_{\substack{s:\\\mathrm{supp}(s) = \{\ell\}}} s^{\otimes 2}.
\end{align}
Therefore, we have
\begin{align}
    &\bbE_{U_{j}}\Tr[U^\dag_{j} s_{j} U_{j} \ketbra{0^n}]^2 = \frac{1}{3}\cdot \frac{1}{2^n}\\
    &\bbE_{U_{j}}\sum_{\substack{s_{j-1} \in \calP_n: \\ \abs{s_{j-1}}=1}} \Tr[U^\dag_{j} s_{j} U_{j} s_{j-1}]^2 = 1%\frac{2}{5}.
\end{align}

\bigskip
\noindent\underline{General case:}
Putting all together, we have demonstrated that
\begin{align}
    &\bbE_{U_{j}}\Tr[U^\dag_{j} s_{j} U_{j} \ketbra{0^n}]^2 \geq \frac{1}{5}\cdot\frac{1}{2^n}   \label{eq:k=1}\\
    &\bbE_{U_{j}}\sum_{\substack{s_{j-1} \in \calP_n: \\ \abs{s_{j-1}}=1}} \Tr[U^\dag_{j} s_{j} U_{j} s_{j-1}]^2 \geq \frac{2}{5}\label{eq:k=1bis}.
\end{align}

We can now lower bound the variance of the truncated path integral. By Eq.\ \eqref{eq:k=1}, the final term $\bbE_{U_1}\Tr[U^\dag_1 s_1 U_1 \ketbra{0^n}]^2 $ bears a factor of at least $1/(5\cdot2^n)$, and there are $L$ terms of the form
$\bbE_{U_{j}}\sum_{\substack{s_{j-1} \in \calP_n: \\ \abs{s_{j-1}}=1}} \Tr[U^\dag_{j} s_{j} U_{j} s_{j-1}]^2$, which bear a factor of at least $2/5$ by Eq.\ \eqref{eq:k=1bis}. Moreover, $\sum_{s_L:\abs{s_L}=1}\Tr[Os_L]^2 = \norm{O^{(1)}}^2_{2}$. Thus we obtain
\begin{align}
    \mathrm{Var}_U \tilde f_U^{(1)}(O) \geq \frac{1}{5}\left(\frac{2}{5}\right)^L \norm{O^{(1)}}^2_{\mathrm{Pauli},2}.
\end{align}
This completes the proof.
\end{proof}
We can further observe that, for Brickwall circuit, the ``case 2'' discussed in the above proof never happens, and thus we can precisely quantify the improvement of weight-1 Pauli propagation over the trivial estimator as $\mathrm{Var}_U \tilde f_U^{(1)}(O) = \frac{1}{5}\left(\frac{2}{5}\right)^L \norm{O^{(1)}}^2_{\mathrm{Pauli},2}$.
We also note that $\mathrm{Var}_U \tilde f_U^{(1)}(O)$ lower bounds the variance of the expectation value $\Tr[OU\rho U^\dag]$, and therefore Proposition\ \ref{prop:impr-k=1} implies that local observables do not suffer from barren plateaus on random circuits of logarithmic depth, independently of the geometric dimension of the circuit.
We formalize this observation in the following Corollary.
\begin{corollary}[Absence of Barren Plateaus for circuits with arbitrary connectivity]
Let $P$ be a Pauli observable with Pauli-weight $\abs{P}=1$ and let $U= U_L U_{L-1}\dots U_1$ be a $L$-layered locally scrambling circuit with input $\ketbra{0^n}$. Assume that each layer $U_j$ consists in non-overlapping single-qubit and 2-qubit gates sampled from local 2-designs. We have
    \begin{align}
        \mathrm{Var}_U \Tr[PU\rho U^\dag] \geq \frac{1}{5}\left(\frac{2}{5}\right)^L.
    \end{align}
\end{corollary}
We emphasize that the above Corollary is not entirely novel, as similar scalings were also obtained in Ref.\ \cite{napp2022quantifying} hinging on the ``statistical mechanical mapping'', a technique that allows to express the second moments of random quantum circuits as the expected values of some suitable Markov chains.

\medskip

Proposition\ \ref{prop:impr-k=1} also implies that there exists circuits where our algorithm is, simultaneously, significantly more accurate than the trivial estimator, and super-polynomially faster than brute-force simulation.
In particular, consider a Pauli observable $P$ with Pauli-weight $\abs{P}=1$ and a logarithmic depth circuit (i.e. with $L = \Theta(\log(n))$) with geometric dimension $D > 1$, satisfying the assumptions of Proposition\ \ref{prop:impr-k=1}.
\begin{enumerate}
    \item Thanks to Proposition\ \ref{prop:impr-k=1}, we know that low-weight Pauli propagation is inversely polynomially more accurate than the trivial estimator:
    \begin{align}
    \mathrm{Var}_U \tilde f_U^{(1)}(O) \geq \frac{1}{5}\left(\frac{2}{5}\right)^{L} \in \frac{1}{n^{\calO(1)}}.
    \end{align}
    \item As discussed in Supplemental Material\ \ref{app:lightcone}, exactly computing the evolution of $P$ via brute-force simulation would take time
    \begin{align}
        &n^{\calO(\log n)} \qquad \text{if $D\in\calO(1)$},
        \\& \exp(\calO(n)) \quad \text{if $D=n$}.
    \end{align}
    In contrast, low-weight Pauli propagation achieves inversely polynomially small error and failure probabilities with a runtime of 
    \begin{align}
        & n^{\calO(\log \log n)} \quad \text{if $D\in\calO(1)$},
        \\ & n^{\calO(\log n)} \,\qquad \text{if $D=n$}.
    \end{align}
\end{enumerate}
We also remark that, since this circuit has logarithmic depth, the results from Refs.\ \cite{napp2022efficient, bravyi2024classical}, which concern constant-depth circuits, are not directly applicable. 

\begin{comment}
{
\subsection{Relative errors}
Consider the relative error:  
\begin{align}
    R_U = \frac{\Delta f_U^{(k)}}{f_U(O)}
\end{align}
\begin{lemma}[MSE-to-Variance Ratio]
  
\end{lemma}
\begin{proof}
For a real random variable $X$ and $a>0$, Markov's inequality implies that
\begin{align}
    \Pr[\abs{X} \geq a] = \Pr[X^2 \geq a^2] \leq \frac{\bbE [X^2]}{a^2},
\end{align}
\begin{align}
    \bbE[Y^2] = \int_{Y^2 \leq T} Y^2 d\mu  + \int_{Y^2 \geq T} Y^2 d\mu \leq T + \Pr[Y^2 \geq T]
    \\ \implies \Pr[Y^2 \geq T] \geq \bbE[Y^2] - T
\end{align}
\begin{align}
    \Pr[Y^2 \geq \frac{\bbE[Y^2]}{2}] \geq \frac{\bbE[Y^2]}{2}.
\end{align}
\end{proof}
}
\end{comment}

\subsection{Implications for the XQUATH conjecture}

The notion of ``improving over the trivial estimator'' is closely related to the XQUATH (Linear Cross-Entropy Quantum Threshold) conjecture proposed by\ \citet{aaronson2019classical}.
Given a distribution over quantum circuits $\calD$, we consider the task of estimating the probability $f_U(\ketbra{0^n}) :=  \abs{\bra{0^n}U\ket{0^n}}^2 $ for a random circuit $U\sim\calD$.
The XQUATH conjecture tells that no efficient classical algorithm can achieve a slightly better variance compared with the trivial estimator, which in this case outputs always $\bbE_U f_U(\ketbra{0^n}) = 1/2^n$. %(i.e. ``guessing zero'' in our terminology).
This conjecture is central to the complexity theoretic foundation of the linear cross-entropy benchmark used in quantum supremacy experiments.

\begin{conjecture}[XQUATH, or Linear Cross-Entropy Quantum Threshold Assumption]
 There is no
polynomial-time classical algorithm that takes as input a quantum circuit $U\sim\nu$ and produces a number $q(C, 0^n)$ such that
\begin{align}
    \mathrm{XQ}:= \bbE_{U\sim\calD} \left[\left[f_U(\ketbra{0^n}) - 2^{-n}\right]^2 \right] - \bbE_{U\sim\calD} \left[\left[f_U(\ketbra{0^n}) - q(U,0^n) \right]^2 \right]  = \Omega\left(2^{-3n}\right).
\end{align}
\end{conjecture}
This conjecture was refuted for random Brickwork circuit of sublinear depth in Ref.\ \cite{aharonov2023polynomial} using a simple Pauli-path based propagation algorithm, which consists in following a single Pauli-path and achieve $XQ = 15^{-L}$, where $L$ is the circuit depth. In particular, this implies that the XQUATH conjecture does not hold on random Brickwork circuit of depth smaller than $\log_2(1/15) n \simeq n/4$.
It is natural to ask whether this result can be further improved by using the low-weight Pauli propagation algorithm.
To this end, we observe that a simple application of Proposition\ \ref{prop:improvement} yields the following corollary.
\begin{corollary}
Let $U$ be a random circuit with orthogonal Pauli paths. We have
     \begin{align}
        \mathrm{XQ} := \mathrm{Var}_U \tilde f_U^{(k)}(\ketbra{0^n}).
    \end{align}   
\end{corollary}
This observation allows us to quickly compute the $\mathrm{XQ}$-value for the weight-1 Pauli propagation algorithm.

\begin{corollary}
Let $U$ be a random Brickwork circuit. Then
the weight-1 Pauli propagation algorithm outputs a number $q(U,0^n)$ in time $\calO(n L)$ that achieves
\begin{align}
    \mathrm{XQ} = \frac{1}{5} \left(\frac{2}{5}\right)^L n  2^{-2n} \label{eq:xq-1},
\end{align}
Therefore XQUATH is false for random Brickwork circuits with depth $L \leq \frac{3}{4}n$.
\end{corollary}
\begin{proof}
    Let $O = \ketbra{0^n} = 2^{-n}\sum_{P \in \{I,Z\}^{\otimes n}} P$ be the projector on the computational zero-state. We have
    \begin{align}
        \norm{O^{(1)}}^2_{\mathrm{Pauli},2} :=  \sum_{\substack{P \in \{I,Z\}^{\otimes n}:\\ \abs{P}=1}}  (2^{-n})^2 = n  2^{-2n}.
    \end{align}
    By applying Proposition\ \ref{prop:impr-k=1}, we obtain
    \begin{align}
         \mathrm{XQ} := \mathrm{Var}_U \tilde f_U^{(k)}(\ketbra{0^n}) = \frac{1}{5} \left(\frac{2}{5}\right)^L \norm{O^{(1)}}^2_{\mathrm{Pauli},2} =  \frac{1}{5} \left(\frac{2}{5}\right)^L n  2^{-2n}. 
    \end{align}
    As a consequence, for random Brickwork circuits of depth at most $\frac{n + \log(n)}{\log_2(5/2)} \geq \frac{3}{4}n$, we have $\mathrm{XQ} = \calO(2^{-3n})$, and therefore the XQUATH conjecture does not hold.
    
\end{proof}
Since the RHS of Eq.\ \ref{eq:xq-1} decays exponentially in depth, this ``attack'' can be easily countered by increasing the size of the circuits. However, the value of $k$ could be augmented as well to produce a more accurate estimator. Thus, it is natural to ask whether low weight Pauli propagation methods can be used to refute the XQUATH conjecture for circuits of linear depth. Here answer to this question negatively, by leveraging previous bounds on approximate unitary designs.

\medskip

In particular, we consider the following definition of $\epsilon$-approximate unitary 2-design.
\begin{definition}[Approximate design]
A distribution $\calD$ over $\bbU(2^n)$ is an $\epsilon$-approximate unitary 2-design if
\begin{align}
    \norm{ \Psi_{\mathcal{U}} - \Psi_{\mathrm{\calD}} }_{1\rightarrow 1} \leq \epsilon ,
\end{align}
where the quantum channel $\Psi_{\mathrm{\calD}} (\cdot)$ is defined via
\begin{align}
    \Psi_{\mathrm{\calD}}(A) := \bbE_{U\sim\calD} \left[U^{\otimes 2} A U^{\dag \otimes 2} \right], 
\end{align}
and similarly for the Haar measure $\calU$. %Here, $\Psi \preceq \Psi'$ denotes that $\Psi' -\Psi$ is a completely-positive map.
For a superoperators $\Psi$, we denoted its induced 1-norm by $\norm{\Psi}_{1\rightarrow 1} := \max_{\rho} \norm{\Psi(\rho)}_1$, where the maximization is taken over the set of all quantum states.
\end{definition}
We emphasize that this is a relatively weak definition of approximate design, which is implied by stronger definitions, such as the diamond-norm based definition and the multiplicative definition used, for instance, in Refs. \cite{schuster2024random, harrow2023approximate, chen2024incompressibility}.
Given a random circuit $U$ sampled from an approximate 2-design, we upper bound the average (squared) Hilbert Schmidt norm of the truncated observable $O^{(k)}_U$, i. e.  the final observable obtained by running the low-weight Pauli propagation algorithm.

\begin{lemma}[Average norm contraction]
\label{lemma:contraction}
Let $O := \ketbra{0^n}$ be the projector onto the computational 0-state let and let be $\calD$ be an $\epsilon$-approximate 2-design . 
Let $U= U_L U_{L-1}\dots U_1$ be a random circuit sampled from $\calD$ and 
assume %that the Fourier coefficients of different Pauli paths are uncorrelated, i.e. $\bbE_U \Phi_\gamma(U) \Phi_{\gamma'}(U)=0 $ whenever $\gamma \neq \gamma'$. $
that $U$ has orthogonal Pauli paths.
We have
\begin{align}
    \bbE_U\norm{O^{(k)}_U}_2^{2} - 2^{-n} \in \calO(n^k)\left(\frac{\norm{O}_{2}^2 -2^{-n}}{4^n-1} +\frac{\epsilon}{2^n}\right).
\end{align}
\end{lemma}
\begin{proof}
As a preliminary step, we prove the following claim.
\begin{claim}
\label{claim:convergence}
For all Pauli operator $P\in \{I,X,Y,Z\}^{\otimes n}\setminus \{I^{\otimes n}\}$,
\begin{align}
        \bbE_{U\sim\calD}\Tr[U P U^{\dag} O]^2 
    \leq  \frac{2^n\norm{O}_2^2 -1}{4^n-1} +\epsilon.
\end{align}  
\end{claim}
\begin{proof}[Proof of \cref{claim:convergence}]
We have
\begin{align}
  &\bbE_{U\sim\calD}\Tr[U P U^{\dag} O]^2 = \bbE_{U\sim\calD}\Tr[P^{\otimes 2}U^{\otimes 2\dag} O^{\otimes 2} U^{\otimes 2}]
  \\= &\bbE_{U\sim\calD, V\sim\mathcal{U}}\Tr[P^{\otimes 2} (U^{\otimes 2\dag} O^{\otimes 2} U^{\otimes 2} - V^{\otimes 2\dag} O^{\otimes 2} V^{\otimes 2})] + \bbE_{V\sim\mathcal{U}}\Tr[P^{\otimes 2}V^{\otimes 2\dag} O^{\otimes 2} V^{\otimes 2}]
 \\ \leq &\bbE_{U\sim\calD, V\sim\mathcal{U}} \norm{P}\norm{U^{\otimes 2\dag} O^{\otimes 2} U^{\otimes 2} - V^{\otimes 2\dag} O^{\otimes 2} V^{\otimes 2}}_1 + \bbE_{V\sim\mathcal{U}}\Tr[P^{\otimes 2}V^{\otimes 2\dag} O^{\otimes 2} V^{\otimes 2}]
  \\\leq &\bbE_{V\sim\mathcal{U}}\Tr[VPV^{\dag} O]^2 + \epsilon.
\end{align}
where we used Hölder's inequality and the definition of $\epsilon$-approximate $2$-design.
By standard Weingarten's calculus we have\ (see, for instance, Corollary 13 in\ \cite{mele2023introduction}):
\begin{align}
    \bbE_{V\sim\mathcal{U}} V^{\dag\otimes 2} O^{\otimes 2}V^{\otimes 2} = \left( \frac{1 - 2^{-n} \norm{O}_2^2}{4^n -1} \right) I^{\otimes 2} + \left( \frac{\norm{O}_2^2 -2^{-n}}{4^n-1}\right)\mathrm{SWAP}, 
\end{align}
and therefore
\begin{align}
    \bbE_{V\sim\mathcal{U}}\Tr[VPV^{\dag} O]^2 = \bbE_{V\sim\mathcal{U}}\Tr[P^{\otimes 2}V^{\otimes 2\dag} O^{\otimes 2} V^{\otimes 2}] = \frac{2^n\norm{O}_2^2 -1}{4^n-1}.
\end{align}
Putting all together, we have
\begin{align}
    &\bbE_{U\sim \calD}
     \Tr[UP U^{\dag} O] 
     \leq \frac{2^n\norm{O}_2^2 -1}{4^n-1}+\epsilon.
\end{align}
\end{proof}
We now upper bound the expected 2-norm of $O^{(k)}_U$.
Using the orthogonality of Pauli paths, we obtain
  \begin{align}
      \bbE_U\norm{O^{(k)}_U}_2^{2} = \bbE_U \Tr[(O^{(k)}_U )^2] =\bbE_U \sum_{\gamma \in \calS_k} \Phi_\gamma(U) ^2 \label{eq:up-paths-norm}
  \end{align} 
In order to upper bound the RHS of Eq.\ \ref{eq:up-paths-norm}, we include some additional paths in the sum:
\begin{align}
           \bbE_U \sum_{\gamma \in \calS_k} \Phi_\gamma(U) ^2 \leq &\bbE_U \sum_{\abs{s_L}\leq k} \Phi_\gamma(U) ^2 = \frac{1}{2^n}+ \bbE_U \sum_{1\leq\abs{s}\leq k} \Tr[s^{\otimes 2}(U^\dag O U)^{\otimes 2}].
           \label{eq:extra-paths}
\end{align}
Applying Claim\ \ref{claim:convergence}, we obtain that
    \begin{align}
    \bbE_U \sum_{1\leq \abs{s}\leq k} \Tr[s^{\otimes 2}(U^\dag O U)^{\otimes 2}] \leq &
    \sum_{1\leq \abs{s}\leq k} 2^{-n}\left(\frac{2^n\norm{O}_2^2 -1}{4^n-1} +\epsilon\right)
    \in \calO(n^k)\left(\frac{\norm{O}_{2}^2 -2^{-n}}{4^n-1} +\frac{\epsilon}{2^n}\right).
  \end{align} 
Putting all together yields the desired upper bound:
\begin{align}
    \bbE_U\norm{O^{(k)}_U}_2^{2} - 2^{-n} \in \calO(n^k)\left(\frac{\norm{O}_{2}^2 -2^{-n}}{4^n-1} +\frac{\epsilon}{2^n}\right).
\end{align}
This completes the proof.
\end{proof}

Intuitively, the above lemma shows that the outputs of low-weight Pauli propagation and the trivial estimator are nearly indistinguishable. However, in order to show that low-weight PP does not violate the XQUATH conjecture, we need a slightly tighter bound, which we achieve by combining 2 approximate designs sequentially.

\begin{theorem}\label{thm:no-go}
Let $U= U^{(A)}U^{(B)}$ be a random circuit with orthogonal Pauli paths, where  $U^{(A)}$ and $U^{(B)}$  are two unitaries sampled independently from an $\epsilon$-approximate unitary 2-design. We have
\begin{align}
    \mathrm{Var}_U \tilde{{f}}_U^{(k)}(\ketbra{0^n}) \in  \calO(n^{4k}2^{-5n}) + \calO\left({\epsilon n^k}2^{-n} \right),
\end{align}
Therefore low-weight Pauli propagation does not violate the XQUATH conjecture on $U$ if $\epsilon \leq 1/(2^n n^k) $.
\end{theorem}

\begin{proof}
Let $O = \ketbra{0^n}$.  We have
\begin{align}
    \mathrm{Var}_U \tilde{{f}}_U^{(k)}(O) = \bbE_U \Tr[O_U^{(k)} \rho]^2 - \left\{\bbE_U\Tr[O_U^{(k)} \rho]\right\}^2 = \bbE_U \Tr[O_U^{(k)} \rho]^2  - 4^{-n}.
\end{align}
Moreover, we recall that $O^{(k)}_U$ can be rewritten as $U_1^\dag O_1 U_1$, where $O_1$ is the observable defined in Eq.\ \eqref{eq:Oj-defin}.
Then by Lemma\ \ref{lemma:odg}, we have
\begin{align}
    \bbE_U \Tr[O_U^{(k)} \rho]^2 = \bbE_U \left(\Tr[U_1^\dag O_1 U_1\rho]^2\right) \leq \bbE_U\norm{O_1}^2_{\mathrm{Pauli},2} = \bbE_U \norm{O^{(k)}_U}^2_{\mathrm{Pauli},2} 
\end{align}
The desired result follows by applying Lemma\ \ref{lemma:contraction} on both $U^{(A)}$ and $U^{(B)}$. Denote by $O^{(k)}_{U^{(A)}}$ the intermediate observable obtained after running the low-weight Pauli propagation on $U^{(A)}$.
\begin{align}
   \bbE_U &\norm{O_U^{(k)}}^2_{\mathrm{Pauli},2} -4^{-n} =  \bbE_{U^{(A)}U^{(B)}} \norm{O_{U^{(A)}U^{(B)}}^{(k)}}^2_{\mathrm{Pauli},2} - 4^{-n}
   \\\leq &\left(\bbE_{U^{(A)}} \norm{O_{U^{(A)}}^{(k)}}^2_{\mathrm{Pauli},2} - 4^{-n}\right) + \beta
   \\\leq &\alpha^2 \left(\norm{O}^2_{\mathrm{Pauli},2} - 4^{-n}\right) + \beta + \alpha\beta 
\end{align}
for $\alpha =  \calO\left(\frac{n^k}{4^n} \right)$ and $\beta =  \calO\left(\frac{\epsilon n^k}{2^n} \right)$. Since $\norm{O}^2_{\mathrm{Pauli},2} = 2^{-n}$, we have
\begin{align}
    \bbE_U &\norm{O_U^{(k)}}^2_{\mathrm{Pauli},2} -4^{-n} \in \calO(n^{4k}2^{-5n}) + \calO\left({\epsilon n^k}2^{-n} \right).
\end{align}
\end{proof}

Random quantum circuits on several generic architectures forms approximate 2-designs at linear depth. One of the best current result for Brickwork circuits was recently established in Ref.\ \cite{chen2024incompressibility}.
\begin{lemma}[Adapted from Corollary 1.7 in Ref.\ \cite{chen2024incompressibility}]
\label{lem:chen}
Random Brickwork circuits generate $\epsilon$-approximate
unitary 2-designs in depth $L = \calO(n+\log(1/\epsilon))$.
\end{lemma}
Therefore, Theorem\ \ref{thm:no-go} and Lemma\ \ref{lem:chen} together imply the following corollary.
\begin{corollary}
There exists a constant $c>0$, such that low-weight Pauli propagation algorithm cannot be used to refute the XQUATH conjecture on random brickwork circuits of depth $L\geq c\cdot n$.
\end{corollary}
Moreover, we observe that the upper bound in Eq. \eqref{eq:extra-paths} is extremely loose, since it takes into account only the truncation performed in the first iteration othe algorithm. Similarly, the proof of Theorem\ \ref{thm:no-go} only exploits the presence of two truncations, while the actual algorithm performs $L$ truncation rounds on an $L$-layered circuit. Therefore we anticipate that our proof technique could be applied to virtually any efficiently computable Pauli propagation method.

\section{Numerical error certificates}\label{app:certacc}
In Lemma\ \ref{lemma:mse}, we showed that, for all set of Pauli paths $\calS \subseteq \calP_n^{L+1}$, the mean squared error can always be expressed as
\begin{align}
    \bbE_U \Delta f_U^{(\calS)}% = &\bbE_U\left[{f}_U(O)^2\right] - \bbE_U\left[\tilde{f}^{(\calS)}_U(O) ^2\right]
    = \bbE_U\left[\tilde{f}^{(\overline{\calS})}_U(O) ^2\right] = \bbE_U \sum_{\gamma \in \overline{\calS}} \Phi_\gamma(U) ^2 d_\gamma^2 ,
\end{align}
where we denoted the complement set $\overline{\calS} := \calP_n^{L+1} \setminus \calS$. 
By leveraging this identity, we can design a randomized classical algorithm for estimating the mean squared error of any function $\tilde{f}^{(\calS)}_U$, which works by Monte-Carlo sampling the second moment of $\tilde{f}^{(\overline{\calS})}_U(O)$.

This approach hinges on the fact that the (renormalized) squared Fourier coefficients can be interpreted as probabilities, as we show in the lemma below.
\begin{lemma}[Fourier spectrum]
Let $O$ be an observable and $U=U_L U_{L-1}\dots U_1$ be a random $L$-layered circuit with orthogonal Pauli paths. We have
\begin{align}
    \sum_{\gamma \in \calP_n^{L+1}}  \bbE_U \Phi_\gamma(U) ^2 = \norm{O}_2^2,
\end{align}
where $\Phi_\gamma(U)  = \Tr[O s_L] \Tr[U_L^\dag s_L U_L s_{L-1}]\dots \Tr[U_1^\dag s_1 U_1 s_0]$.
Therefore, denoting $p(\gamma) :=  \bbE_U\Phi_\gamma(U) ^2/\norm{O}_2^2$, we have
\begin{align}
    \sum_{\gamma \in \calP_n^{L+1}} p(\gamma) = 1.
\end{align}
\end{lemma}
\begin{proof}
Consider the Heisenberg-evolved observable $U^\dag O U =\sum_{\gamma \in \calP_n^{L+1}} \Phi_\gamma(U)  s_\gamma$. We have
\begin{align}
    \norm{O}_2^2 = \bbE_U \norm{U^\dag O U}_2^2 = \bbE_U \Tr[\left(\sum_{\gamma \in \calP_n^{L+1}} \Phi_\gamma(U)  s_\gamma\right)^2] = \sum_{\gamma \in \calP_n^{L+1}} \bbE_U\Phi_\gamma(U) ^2,
\end{align}
where in the first step we used the fact that the Hilbert Schmidt norm is unitarily invariant, and the last step follows from the Fourier coefficients of different Pauli paths being uncorrelated.
%Renormalizing both the LHS and RHS yields the second part of the Lemma.
\end{proof}

This observation allows us to efficiently estimate the mean squared error, as formalized in the following theorem, which is a restatement of Theorem\ \ref{thm:certacc} of the \emph{End Matter} section.

\begin{theorem}[Certified error estimate]
Let $U=U_L U_{L-1}\dots U_1$ be a random circuit with orthogonal Pauli paths and let $O$ be an observable. Assume that we can sample $s\in\calP_n$ with probability ${\Tr[Os]^2}/{\norm{O}^2_{2}}$ in time $\mathrm{poly}(n)$.
Moreover,  assume that for $j = L, L-1, \dots, 1$, and for all $s_j\in \calP_n$, we can sample $s_{j-1}$ with probability
$ \bbE_{U_j}\Tr[U^\dag_j s_jU_j s_{j-1}]^2$ in time $\mathrm{poly}(n)$.
Then, for any $\epsilon,\delta\in(0,1]$, there exists a classical randomized algorithm that runs in time $\mathrm{poly}(n) L\epsilon^{-2}\log(\delta^{-1})$ and outputs a value $\alpha$ such that
\begin{align}
    \left|{\alpha - \bbE_U\Delta f^{(\calS)}_U}\right| \leq \epsilon \norm{O}^2_{\mathrm{Pauli},2}.
\end{align}
with probability at least $1-\delta$.
\end{theorem}
\begin{proof}
We show how to efficiently estimate $\bbE_U\Delta f^{(\calS)}_U$ for all $\calS$. 
We have
\begin{align}
    \bbE_U\Delta f^{(\calS)}_U = \bbE_U \sum_{\gamma \in \overline\calS} \Phi_\gamma(U) ^2 d_\gamma^2 := \sum_{\gamma \in \calP_n^{L+1}} p(\gamma) X_\gamma^{(\calS)},
\end{align}
where we introduced the following variable:
\begin{align}
    X_\gamma^{(\calS)} = \begin{cases}
    0 &\text{if $\gamma \in \calS$,}
       \\ \norm{O}_2^2 \cdot d_\gamma^2 &\text{if $\gamma \in \overline\calS$,}
    \end{cases}
\end{align}
Moreover, we observe that $0 \leq X_\gamma^{(\calS)} \leq \norm{O}^2_{\mathrm{Pauli},2}$.

We then implement the following protocol.
\begin{enumerate}
    \item Sample $M$ strings $\gamma_1, \gamma_2, \dots, \gamma_M$ i.i.d. with probability $p(\gamma_i) :=  \bbE_U\Phi_{\gamma_i}^2/\norm{O}_2^2$.
    \item Compute $Y =\frac{1}{M}\sum_{i=1}^M X_{\gamma_i}^{(\calS)}$ .
\end{enumerate}
By Chernoff-Hoeffding bound, we have
\begin{align}
    \Pr[\abs{Y - \bbE_U\Delta f^{(\calS)}_U} \geq \epsilon \norm{O}^2_{\mathrm{Pauli},2}]  \leq \delta,%2 \exp\left(- {2\epsilon^2 M} \right)
\end{align}
provided that $M\geq {\log(2/\delta)}/({2\epsilon^2})$.
It remains to upper bound the time required for sampling each string $\gamma_i$.
We use the following iterative algorithm:
\begin{enumerate}
    \item Sample $s_L$ from the `Pauli spectrum' of $O$ with probability 
        ${\Tr[Os]^2}/{\norm{O}^2_{2}}$,
    \item for $j = L, L-1, \dots, 1$, sample $s_{j-1}$ with probability
    \begin{align}
        \bbE_{U_j}\Tr[U^\dag_j s_jU_j s_{j-1}]^2
    \end{align}
    \item Return $\gamma = (s_1,s_2,\dots, s_L)$.
\end{enumerate}
Thus the sampling algorithm requires time $\mathrm{poly}(n)L$, and the entire procedure runs in time $\mathrm{poly}(n)L\epsilon^{-2}\log(1/\delta)$.
\end{proof}

\section{Further numerical experiments beyond our bounds}\label{app:numerics-IBM}

\begin{figure*}
    \centering
    \includegraphics[width=0.98\linewidth]{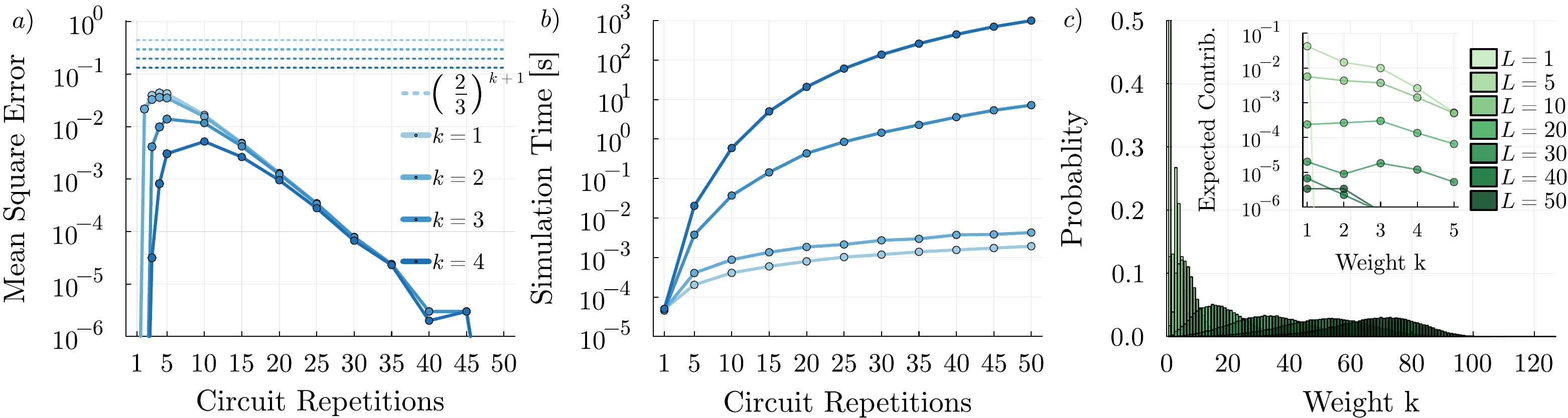}
  \captionsetup{labelformat=empty} % Remove "Figure X" label
  \caption{{Supplemental Figure 3:} \small Numerical verification of the effectiveness of weight truncation for a 127-qubit quantum circuit that does not strictly comply with our assumptions. We employ an ansatz consisting of repeated RX and RZZ rotations equivalent to a Trotter time evolution circuit of the transverse field Ising Hamiltonian. Thus, one layer of this circuit is not locally scrambling. The entangling topology is chosen to be the heavy-hex lattice, and the measurement is $\sigma^z_{63}$ in the middle of the lattice.}
    \label{fig:numerics-IBM}
\end{figure*}

In addition to Fig.~\ref{fig:numerics_2d_staircase} in the main text, we provide further numerical evidence that the assumption of circuit layers randomly drawn from a locally scrambled distribution could be relaxed in practice. 

\medskip

\noindent {\textbf{Circuits with random Pauli rotations.}}
Here we consider the case of a 127-qubit quantum circuit consisting of RX and RZZ gates, i.e., a circuit generated by the operators in a transverse field Ising Hamiltonian. Even a full circuit layer acting on all qubits is strictly speaking not locally scrambling. The entangling topology we employ is that of a \textit{heavy-hex} lattice which was, for example, used in Ref.~\cite{kim2023evidence}. The observable is $\sigma^z_{63}$ in the middle of the lattice. Our results of Pauli propagation with weight truncation are shown in Supplemental Fig.~\ref{fig:numerics-IBM}.

It becomes clear that this example is significantly easier than the pathological example shown in Fig.~\ref{fig:numerics_2d_staircase} in the main text. One qualitative difference is the initial ease of simulation inside the entangling light cone, which allows for fast and close-to-exact simulation for a few layers where the operators are mostly low-weight. At the same time, we also observe that our general error bounds in Theorem~\ref{thm:errorbound} are satisfied despite the layers not being locally scrambling.

We emphasize that these results should not be understood as strongly outperforming the quantum experiment results in Ref.~\cite{kim2023evidence}. While the authors employed a very similar setup up to only 20 layers, the parameters of their quantum circuit aimed to perform a certain quantum simulation task and were strongly correlated. All RZZ angles were set to $\pi/2$ (Clifford but maximally weight increasing) and all RX angles were chosen to be the same tunable value. This choice of parameters is thus similar to the correlated angle example shown in Supplemental Fig.~\ref{fig:numerics-correlated}, but this circuit could not be called locally scrambling even for uncorrelated parameters. Furthermore, it was later shown that only a small region of the correlated parameter space is indeed challenging for classical simulation methods, hinting at the fact that the average-case results shown here may not be representative of the hardness of the experiments in Ref.~\cite{kim2023evidence}.

\begin{figure}
    \centering
    \includegraphics[width=0.50\linewidth]{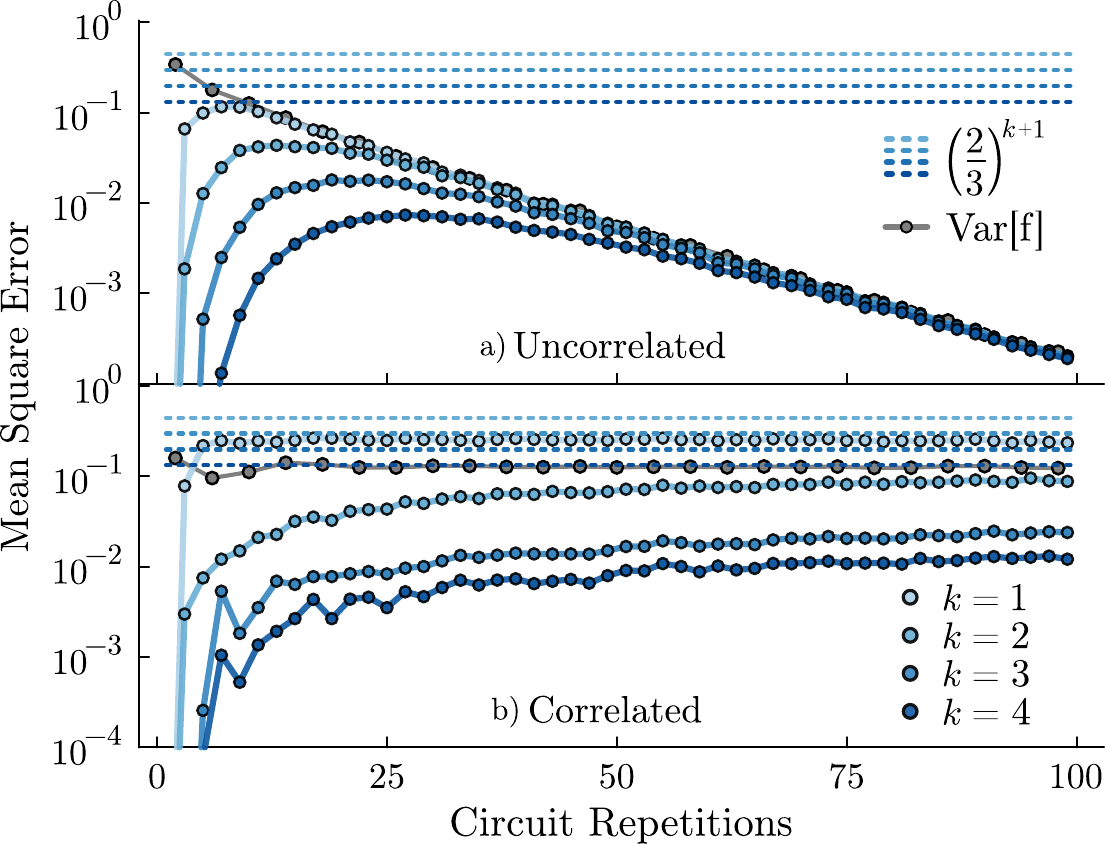}
  \captionsetup{labelformat=empty} % Remove "Figure X" label
  \caption{{Supplemental Figure 4:} \small \emph{Numerical verification of the effectiveness of weight truncation for correlated angles on 16 qubits.} The circuit ansatz consists of repeated RX and RZ rotations on each qubit followed by RZZ gates in a staircase ordering. The observable is a local Pauli Z operator on the first qubit. We either draw a) random parameters for all gates or b) one random parameter for all gates. We also report the variance of the un-truncated loss function $\mathrm{Var}[f]$, which indicates the presence and absence of exponential concentration in the case of uncorrelated and correlated parameters respectively.}
    \label{fig:numerics-correlated}
\end{figure}

\medskip

{
\begin{figure}
    \centering
    \includegraphics[width=0.9\linewidth]{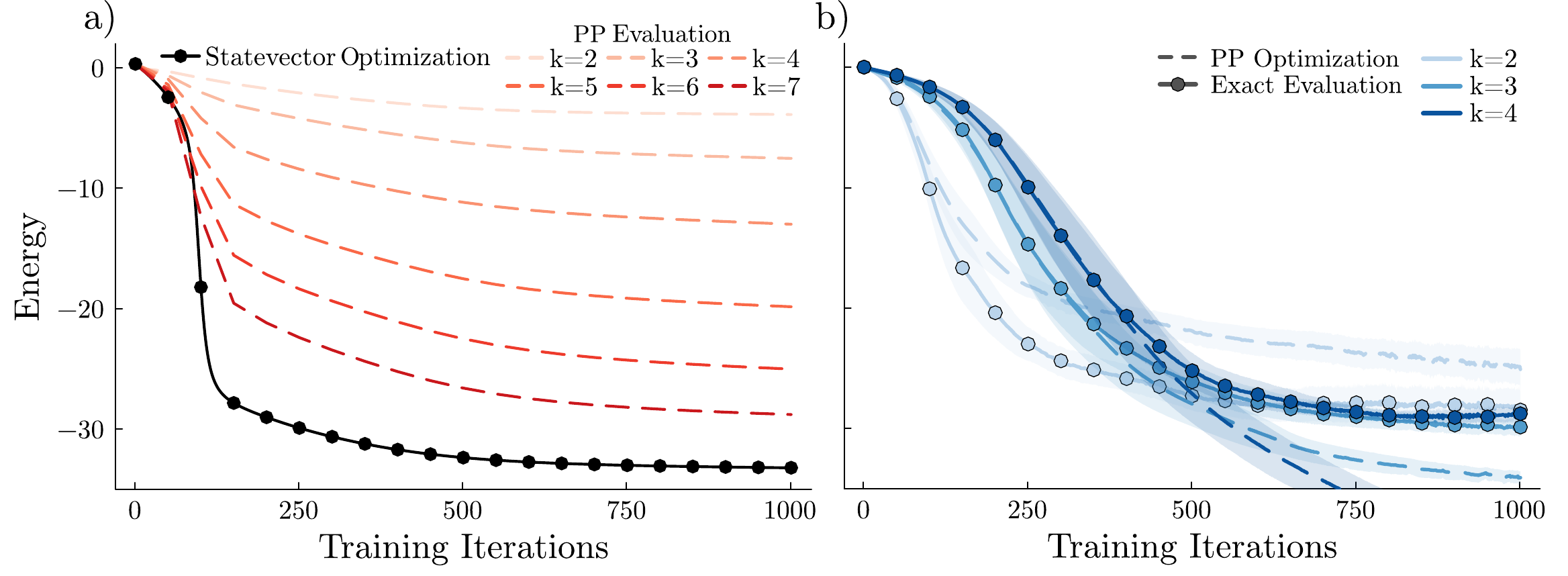}
    \captionsetup{labelformat=empty} % Remove "Figure X" label
    \caption{{Supplemental Figure 5:}
    \small 
    \emph{Variational ground state optimization with low-weight Pauli propagation.} We consider a Heisenberg Hamiltonian on a $4 \times 4$ grid, with all Hamiltonian coefficients equal to 1. As an ansatz, we choose three layers of the 2D staircase SU(4) unitary circuit discussed in the main text. a) Exact statevector optimization with Pauli propagation evaluation of the circuits throughout optimization with varying Pauli weight truncation. b) Low-weight Pauli propagation optimization with exact statevector evaluation of the circuits throughout optimization. }
    \label{fig:optimization}
\end{figure}
\textbf{Training a Variational Quantum Eigensolver.} To explore the broader applicability of low-weight Pauli propagation beyond our average-case guarantees, we investigated its use in training a Variational Quantum Eigensolver (VQE) to find the ground state of a Heisenberg Hamiltonian on a $4 \times 4$ grid, with all Hamiltonian coefficients equal to 1. As an ansatz, we choose three layers of the 2D staircase SU(4) unitary circuit from Ref.~\cite{zhang2023absence} discussed in the main text. When randomly initialized, the model meets our average-case assumptions, making it efficiently simulable with Pauli propagation with high probability. However, as training progresses and parameters are updated, the resulting circuits are expected to deviate from these assumptions.

In Supplemental Fig.~\ref{fig:optimization}a, we indeed observe that, over the course of the optimization using exact statevector simulation, the accuracy of Pauli propagation gradually declines for constant Pauli weight truncation, suggesting that our average case bounds cannot be generalized to more structured circuits such as those generating low-energy states in the Heisenberg model.

Yet, we uncover a surprising phenomenon: when optimizing the circuit parameters via low-weight Pauli propagation (supplemented with a $10^{-3}$ coefficient truncation) we find that the Pauli propagation energy estimates become unphysical, but the found parametrization evaluated via exact statevector simulation results in very low energies,  comparable to those found by statevector optimization. 
In other words, the truncated landscape may still enable efficient discovery of good minima that overlap with those of the exact energy landscape.

This result highlights the potential of Pauli propagation as a tool for variationally learning low energy quantum states. Still, we note that in many cases, even when it identified good parameters, Pauli propagation failed to accurately estimate the corresponding energy with exclusively low Pauli weight. This suggests promising hybrid quantum-classical strategies as illustrated below.
\begin{itemize}
    \item \emph{Train classically, deploy quantumly:} Pauli propagation can be used to train the parameters classically, and subsequently quantum device can be used to to evaluate the final energy, and deploy it for further processing.
    \item \emph{Warm-starts:} Pauli propagation methods can be used for warm-starting quantum devices, by finding nearly optimal values and completing the training using a quantum device.
\end{itemize}  
We also emphasize that the training using the exact simulator may still converge faster to a good solution, leaving the door open for a potential practical quantum advantage for this task. Nevertheless, given that current quantum devices are severely affected by hardware and shot noise, we anticipate that Pauli propagation methods could play a key role in solving physically motivated computation problems in the near-term era.

}

\end{document}